\documentclass[preprint]{imsart}

\usepackage{amsthm,amsmath,natbib}
\RequirePackage[colorlinks,citecolor=blue,urlcolor=blue]{hyperref}

\usepackage{amsmath,amsfonts,amssymb}
\usepackage{algorithm2e}
\usepackage{mathtools}
\usepackage{bm}
\usepackage[mathscr]{euscript}
\usepackage[T1]{fontenc}
\usepackage{lmodern}
\usepackage{enumitem}
\usepackage{xcolor}
\usepackage{graphicx}
\usepackage{multirow}
\usepackage{booktabs}
\usepackage{microtype}
\usepackage{fix-cm}
\usepackage{colortbl}
\usepackage[utf8]{inputenc}
\usepackage{tikz}
\usepackage{pbox}

\startlocaldefs

\newcommand{\Real}{\mathbb R}

\newcommand{\NatInt}{\mathbb N}

\newcommand{\CalO}{\mathcal O}

\newcommand{\CalH}{\mathcal H}
\newcommand{\CalX}{\mathcal X}
\newcommand{\BX}{\bold X}
\newcommand{\Bx}{\bold x}
\newcommand{\By}{\bold y}
\newcommand{\Balpha}{\boldsymbol{\alpha}}
\newcommand{\Bbeta}{\boldsymbol{\beta}}

\newcommand{\cmtS}[1]{{\color{red}{(Simon: #1)}}}

\newtheorem{theorem}{Theorem}

\newtheorem{lem}{Lemma}

\newtheorem{prop}{Proposition}

\newcommand\smallO{
  \mathchoice
    {{\scriptstyle\mathcal{O}}}% \displaystyle
    {{\scriptstyle\mathcal{O}}}% \textstyle
    {{\scriptscriptstyle\mathcal{O}}}% \scriptstyle
    {\scalebox{.7}{$\scriptscriptstyle\mathcal{O}$}}%\scriptscriptstyle
  }
\endlocaldefs

\newcommand{\crd}[1]{{\color{red}{#1}}}

\begin{document}

\begin{frontmatter}

% "Title of the paper"
\title{BdryGP: a new Gaussian process model for incorporating boundary information}

% indicate corresponding author with \corref{}
% \author{\fnms{John} \snm{Smith}\corref{}\ead[label=e1]{smith@foo.com}\thanksref{t1}}
% \thankstext{t1}{Thanks to somebody} 
% \address{line 1\\ line 2\\ printead{e1}}
% \affiliation{Some University}

\author{\fnms{Liang} \snm{Ding}$^*$},
\author{\fnms{Simon} \snm{Mak}$^\dagger$\thanksref{t2}} and
\author{\fnms{C. F. Jeff} \snm{Wu}$^\ddagger$\thanksref{t1} \corref{}}\\
\textit{\textnormal{Hong Kong University of Science and Technology$^*$, Duke University$^\dagger$ and Georgia Institute of Technology$^\ddagger$}}
\thankstext[1]{t2}{Mak's work is supported by U. S. Army Research Office grant W911NF-17-1-0007.} 
\thankstext[2]{t1}{Wu's work is supported by NSF grant DMS 1914632.} 
%\address{\printead{e2}}
%\affiliation{Hong Kong University of Science and Technology and Georgia Institute of Technology }

\runauthor{L. Ding \and S. Mak \and  C. F. J. Wu}

\begin{abstract}
Gaussian processes (GPs) are widely used as surrogate models for emulating computer code, which simulate complex physical phenomena. In many problems, additional boundary information (i.e., the behavior of the phenomena along input boundaries) is known beforehand, either from governing physics or scientific knowledge. While there has been recent work on incorporating boundary information within GPs, such models do not provide theoretical insights on improved convergence rates. To this end, we propose a new GP model, called BdryGP, for incorporating boundary information. We show that BdryGP not only has improved convergence rates over existing GP models (which do not incorporate boundaries), but is also more resistant to the ``curse-of-dimensionality'' in nonparametric regression. Our proofs make use of a novel connection between GP interpolation and finite-element modeling.
\end{abstract}

\begin{keyword}[class=MSC]
\kwd[Primary ]{62G08, 62M40}
\end{keyword}

\begin{keyword}
\kwd{Gaussian Process}
\kwd{Kriging}
\kwd{Computer Experiment}
\kwd{Uncertainty Quantification}
\kwd{Finite-element Modeling}
\kwd{High-dimensional Inputs}
% \kwd{Experiment design}
% \kwd{High dimensional input}
\end{keyword}

\end{frontmatter}
\section{Introduction}

With advances in mathematical modeling and computation, complex phenomena can now be simulated via computer code. This code numerically solves a system of governing equations which represents the underlying science of the problem. Due to the time-intensive nature of these numerical simulations \citep{yeh2018common}, Gaussian processes (GPs; \citealp{sacks1989}) are often used as surrogate models to emulate the expensive computer code. Let $\Bx \in \mathcal{X} = [0,1]^d$ be a vector of $d$ code inputs, and let $f(\Bx)$ be its corresponding code output. The idea is to adopt a GP prior for $f(\cdot)$, then use the posterior process given data to infer code output at an unobserved input. GP emulators are now widely used to study a broad range of scientific and engineering problems, such as rocket engines \citep{Mak18}, universe expansions \citep{kaufman2011efficient} and high energy physics \citep{goh2013prediction}.

In many applications, there is additional knowledge on the phenomenon than simply computer code output, and incorporating such knowledge can improve GP predictive performance. This ``physics-integrated'' GP modeling has garnered much attention in recent years \citep{Wheeler14,Golchi15,WangBerger16}. 
% This notion of ``physics-integrated'' GP modeling has garnered much attention in the past few years. One flavor is the monotone GPs studied in \cite{Golchi15}, \cite{WangBerger16} and \cite{Matt17}, which capture known monotone behavior in the black-box computer code. \cite{Wheeler14} proposed a mechanistic hierarchical GP model which incorporates information from a linear ordinary differential equation. Recently, \cite{MattnLi18} proposed a GP model for parabolic PDEs via the principle of superposition.
We consider here a specific type of information called \textit{Dirichlet boundaries} \citep{bazilevs2007weak}, which specifies the values of $f$ along certain input boundaries. Dirichlet boundaries are often available from governing physics or from simple physical considerations \citep{Matt18}. One example is the simulation of viscous flows \citep{white2006viscous}, widely used in climatology and high energy physics. Such flows are dictated by the complex Navier-Stokes equations \citep{temam2001navier}, and can be very time-consuming to simulate. At the limits of certain variables (e.g., zero viscosity or fluid incompressibility), the Navier-Stokes equations can be greatly simplified for efficient, even closed-form, solutions \citep{kiehn2001some,humphrey2016introduction}. Incorporating this boundary information within the GP can allow for improved predictive performance.

Despite its promise, the integration of GPs with boundary information is largely unexplored in the literature, with the only reference being a recent paper by \cite{Matt18}. In this paper, a flexible \textit{Boundary Modified Gaussian Process} (BMGP) is proposed, which can integrate a broad range of boundaries by modifying the mean and variance structure of a stationary GP. Due in part to its modeling flexibility, the BMGP model is quite complicated and difficult to analyze theoretically. This raises an important open question: to what extent does incorporating boundary information improve convergence rates for GPs?

To address this, we propose a new GP model, called BdryGP, which has provably improved error rates when incorporating boundary information. The key novelty is a new Boundary Constrained Mat{\'e}rn (BdryMat\'ern) covariance function, which incorporates boundary information of the form:
\begin{equation}
\mathcal{F}_j^{[0]} := \{f(\Bx) : x_j = 0\}, \quad \text{or} \quad \mathcal{F}_j^{[1]} := \{f(\Bx) : x_j = 1\}.
\label{eq:bound}
\end{equation}
The BdryMat\'ern covariance inherits the same smoothness properties as the tensor Mat\'ern kernel, while constraining GP sample paths to satisfy \eqref{eq:bound} almost surely.
% In one dimension, this covariance function is derived by expressing the Mat\'ern kernel as the solution of a differential equation via its reproducing property, then solving for the modified kernel which satisfies the same equation under boundary constraints. The resulting BdryMat\'ern covariance inherits the same smoothness properties as the original Mat\'ern kernel, while also imposing a process variance of 0 along known boundaries. This ensures that sample paths from BdryGP satisfy the desired boundary conditions almost surely.
% \begin{figure}
% \centering
% \includegraphics[width=\textwidth]{figures/proofsk.png}
% \caption{A roadmap of major theorems in the paper.}
% \label{fig:roadmap}
% \end{figure}
Assuming boundaries of the form \eqref{eq:bound} is known for each variable $j = 1, \cdots, d$, we prove two main results for BdryGP. The first is a \textit{deterministic} $L^p$ convergence rate for a \textit{fixed} function $f \in \CalH^{1,c}_{mix}(\mathcal{X})$:
\begin{equation}
\|f - \hat{f}_n^{\rm BM}\|_{L^p} = \smallO(n^{-1}), \quad 1 \leq p < \infty.
\label{eq:lp}
\end{equation}
Here, $n$ is the sample size, $\hat{f}_n^{\rm BM}$ is the BdryGP predictor, and $\CalH^{1,c}_{mix}(\mathcal{X})$ is the Sobolev space with mixed first derivatives satisfying \eqref{eq:bound}. The second is a \textit{probabilistic} uniform bound for a \textit{random} function $Z(\cdot)$ following a GP with sample paths in $\CalH^{1,c}_{mix}(\mathcal{X})$:
\begin{equation}
\sup_{\Bx\in[0,1]^d}|Z(\Bx)-\mathcal{I}_n^{\rm BM}Z(\Bx)|=\CalO_{\mathbb{P}}(n^{-1}[\log n]^{2d-\frac{3}{2}}),
\label{eq:unif}
\end{equation}
where $\mathcal{I}_n^{\rm BM}$ is the BdryGP interpolation operator satisfying $\mathcal{I}_n^{\rm BM}f = \hat{f}_n^{\rm BM}$. Both rates require a sparse grid design \citep{Bungartz04}. Compared to existing GP rates (which do not incorporate boundary information), our BdryGP rates decay much faster in sample size $n$. (A full comparison is given in Section \ref{sec:comp}.) Furthermore, by incorporating boundaries, our rates are also more resistant to the well-known ``curse-of-dimensionality'' in nonparametric regression \citep{geenens2011curse}. Our proof makes use of a novel connection between GP interpolation and finite-element modeling (FEM).

% The roadmap of our paper is shown in figure \ref{fig:roadmap}. We will first derive the BdryMat\'ern kernel $k^{\rm BM}_\omega$ and the Brownian kernel $k^{\rm BR}$ which has an equivalent \text{reproducing kernel Hilbert space} (RKHS) of $k^{\rm BM}_\omega$. In following sections, we use theorem \ref{thm:LagrangeKrigingEquivalent} to show the connection between Brownian kernel and FEM, which is the basis for our proof for convergence rates. In section 4, we will prove the equivalence of four functions spaces: we first prove the equivalence between the RKHSs of $k^{\text{BR}}$ and the first order mixed Sobolev space with boundary condition; we then show that the RKHS of $k^{\text{BR}}$ can be decompose by subspaces of different level which we will call hierarchical difference spaces and is closely related to sparse grid. In the last subsection of section 4, we use theorem \ref{thm:diffBdryGP} to show the difference between the interpolators of BdryMat\'ern kernel and Brownian kernel so that we can only focus on Brownian kernel for our proof of convergence rates. In section 5, we can use the result from theorem \ref{thm:hierarchicalInterpolation2Kriging} and the hierarchical difference representation of $\CalH^{1,c}_{mix}$ to prove our convergence rates.

This paper is organized as follows. In Section 2, we present the new BdryGP model and derive the BdryMat\'ern kernel. In Section 3, we establish a novel connection between the BdryGP predictor and the FEM interpolator. In Section 4, we connect the function space for FEM with the native space for the BdryMat\'ern kernel. Using these results, we then derive in Section 5 the main convergence rates for BdryGP, and verify these rates in Section 6 via numerical simulations. Section 7 concludes the paper.

\section{The BdryGP model}
We first give a brief review of GP modeling, then present a model specification for BdryGP.

\subsection{Gaussian process modeling}
Let $\Bx \in \mathcal{X}$ be an input vector on domain $\mathcal{X} = [0,1]^d$, with $f(\Bx)$ denoting its corresponding computer code output. In Gaussian process emulation \citep{sacks1989,santner2003}, $f(\cdot)$ is assumed to be a realization of a Gaussian process with mean function $\mu: \mathcal{X} \to \Real$, and covariance function $k: \mathcal{X} \times \mathcal{X} \to \Real$. Further details on GP modeling can be found in \cite{Adler81}.

Suppose the code is evaluated at $n$ input points $\BX=\{\Bx_1,\cdots,\Bx_n\}\subset \mathcal{X}$, yielding observations $f(\BX)=[f(\Bx_1),\cdots,f(\Bx_n)]^\intercal$. Given data $f(\BX)$,  one can show that the conditional process $f(\cdot)|f(\BX)$ is still a GP, with mean function:
\begin{equation}\label{eq:Kriging_posteriorMean}
    \hat{f}_n(\Bx)=\mu(\Bx)+k(\Bx,\BX)k^{-1}(\BX,\BX)\left[f(\BX)-\mu(\BX)\right],
\end{equation}
and covariance function:
\begin{equation}\label{eq:Kriging_posteriorCov}
    k_n(\Bx,\By)=k(\Bx,\By)-k(\Bx,\BX)k^{-1}(\BX,\BX)k(\BX,\By).
\end{equation}
Here, $k(\BX,\Bx) = [k(\Bx,\Bx_1),\cdots,k(\Bx,\Bx_n)]^\intercal$ denotes the covariance vector between design $\BX$ and a new point $\Bx$, $k(\Bx,\BX) = k(\BX,\Bx)^\intercal$, and $k(\BX,\BX)$ is the covariance matrix ${{[k(\Bx_i},\Bx_j)]_{i=1}^n}_{j=1}^n$ over design points. The posterior mean $\hat{f}_n(\cdot)$ is typically used as a predictor (or emulator) for unknown code output $f(\cdot)$, since it is optimal under quadratic and absolute error loss \citep{santner2003}. The posterior variance $k_n(\Bx,\Bx)$ then quantifies the uncertainty of the predictor $\hat{f}_n(\Bx)$ at a new input setting $\Bx$.
% We will denote this predictor as $\hat{f}_n(\cdot) = \hat{\mu}_n(\cdot)$. 

The kernel $k$ is also associated with an important function space $\CalH_k(\mathcal{X})$, called the \textit{reproducing kernel Hilbert space} (RKHS) or \textit{native space} of $k$ \cite{Wendland10}. For a symmetric, positive definite kernel $k$, the RKHS $\CalH_k(\mathcal{X})$ of $k$ is defined as the closure of the linear function space:
\begin{equation}
\left\{\sum_{i=1}^n c_ik(\cdot,\Bx_{i}):c_i\in\Real,\Bx_{i}\in \CalX, n \in \mathbb{N} \right\}.
\label{eq:rkhs}
\end{equation}
% In words, the RKHS $\CalH_k(\mathcal{X})$ consists of functions which can be well-approximated with linear combinations of kernel $k$.
This RKHS is also endowed with an inner product $\langle \cdot,\cdot\rangle _k$ which satisfies the so-called \textit{reproducing property}:
\begin{equation}
f(\Bx)=\langle f,k(\cdot,\Bx)\rangle_k, \quad \Bx \in \mathcal{X},
\end{equation}
for any function $f\in\CalH_k(\CalX)$. Both the RKHS $\CalH_k(\mathcal{X})$ and its reproducing property will play a key role in the derivation and analysis of the BdryGP.
% Further details on the RKHS can be found in \cite{Wahba90} and \cite{Wendland10}.

\subsection{Boundary information}

In many problems, boundary information on $f$ is available from governing physical principles or scientific knowledge. We consider here a common type of boundary called Dirichlet boundaries \citep{bazilevs2007weak}, which specify the values of $f$ along certain boundaries of the input domain $\mathcal{X} = [0,1]^d$. In particular, we consider boundaries of the form $\mathcal{F}_j^{[0]}$ or $\mathcal{F}_j^{[1]}$ in \eqref{eq:bound}, which quantify the values of $f$ along the left hyperplane $\mathcal{S}_j^{[0]} := \{\Bx:x_j = 0\}$ or right hyperplane $\mathcal{S}_j^{[1]} := \{\Bx:x_j = 1\}$ of a variable $x_j$, respectively. We will call $\mathcal{F}_j^{[0]}$ and $\mathcal{F}_j^{[1]}$ the \textit{left} and \textit{right} boundary condition of variable $x_j$. Boundaries of this form arise naturally in many limiting simplifications of physical systems, and provide closed-form expressions for the BdryMat\'ern kernel.

% More specifically, let $\Bx_{-i}$ denote the vector $(x_1,\cdots,x_{i-1},x_{i+1},\cdots,x_p)$ and let $\partial D$ denote the boundary of of the hypercube $D$:
% $$\partial D:=\bigcup_{i=1}^P\{\Bx_{-i}\in [0,1]^{p-1}, x_i=0,1\}.$$
% Full boundary information for $f$ means that the value $f(\Bx)$ is given for any $\Bx\in \partial D$. On the other hand, partial boundary information for $f$ means that the value $f(\Bx)$ is given for any $\Bx\in U$ where $U= \partial D-\{x_i=0\}_{i\in I}\bigcup \{x_j=1\}_{j\in J}$ with some index set $I$ and $J$. More precisely, in the case of partial boundary information, the value of $f$ is known only on some specific boundary hyperplanes $\{x_i\}_{i\in I}$ where $I\subset \{1,\cdots,D\}$.

% For convenient, we will call the boundary information $f(\Bx_{-i};x_i=0)$ the left boundary condition on dimension $i$ and, similarly, $f(\Bx_{-i};x_i=1)$ the right boundary condition on dimension $i$. 

\begin{figure}
\begin{center}
\includegraphics[width=0.325\textwidth]{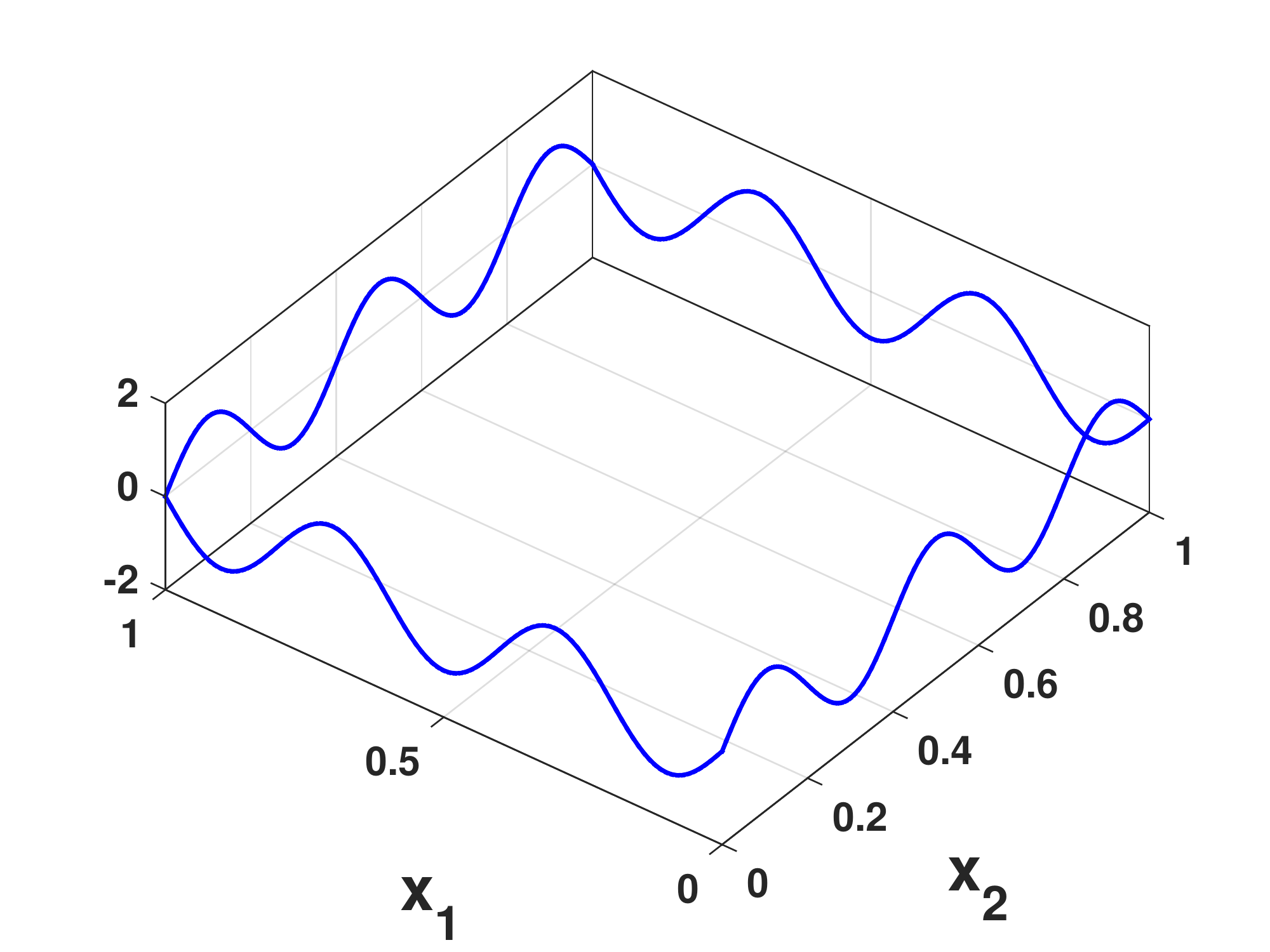}
\includegraphics[width=0.325\textwidth]{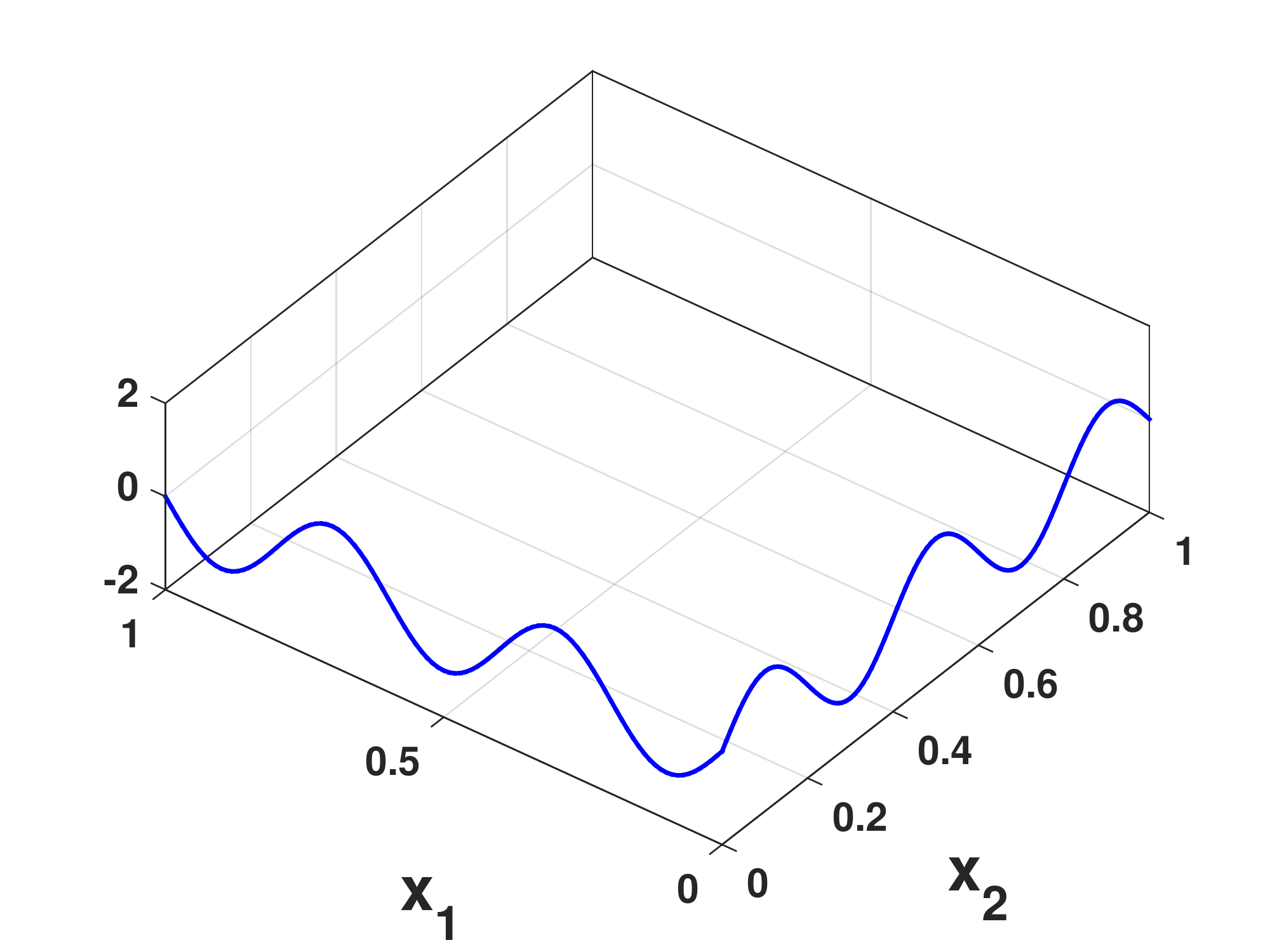}
\includegraphics[width=0.325\textwidth]{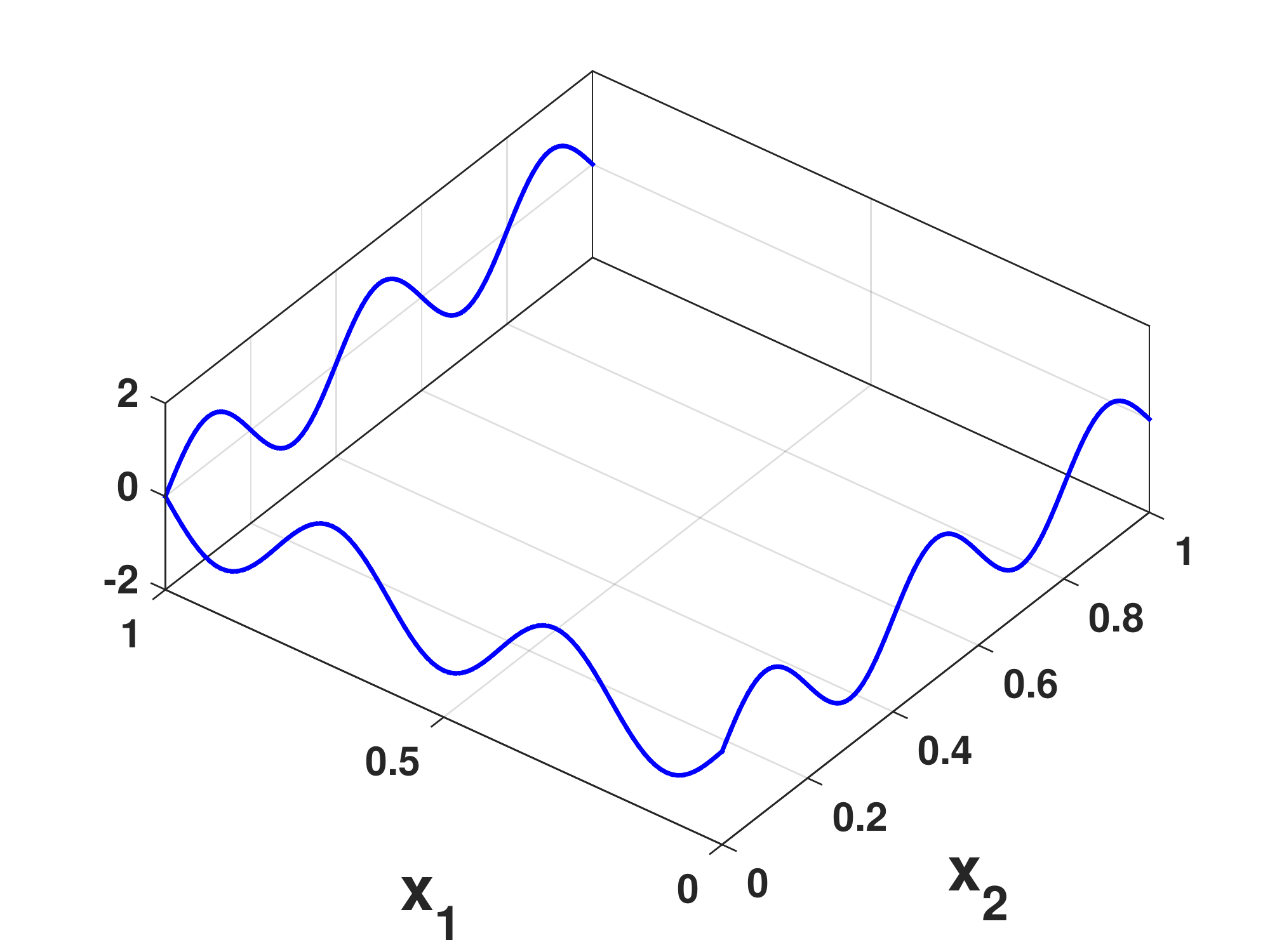}
\end{center}
\caption{Visualizing $I^{[0]} = I^{[1]} = \{1,2\}$ (full boundaries; left), $I^{[0]} = \{1,2\}, I^{[1]} = \emptyset$ (left boundaries; middle), and $I^{[0]} = \{1,2\}$, $I^{[1]} = \{1\}$ (right) for a 2-d function.}
\label{fig:boundary}
    \begin{center}
\includegraphics[width=0.325\textwidth]{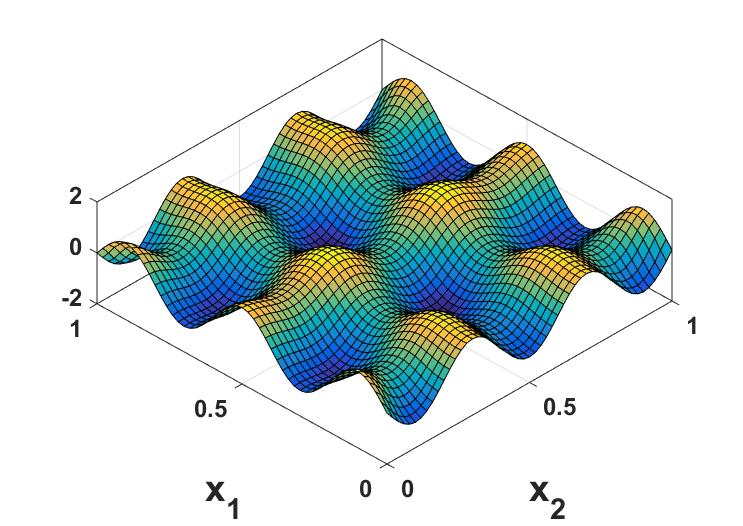}
\includegraphics[width=0.325\textwidth]{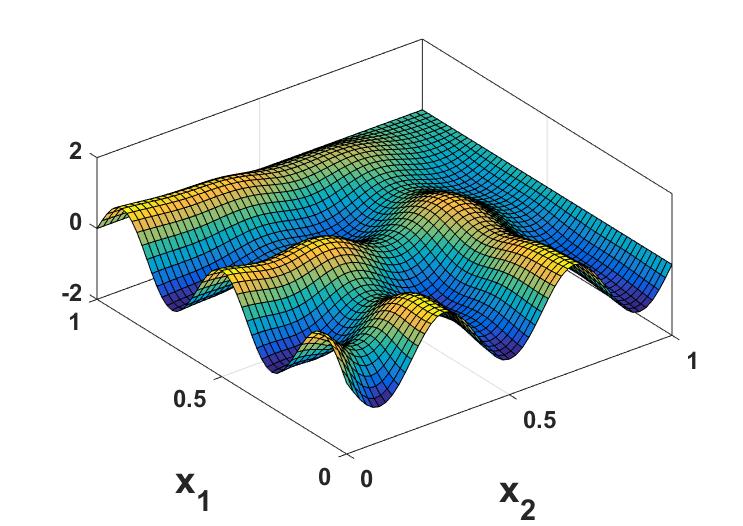}
\includegraphics[width=0.325\textwidth]{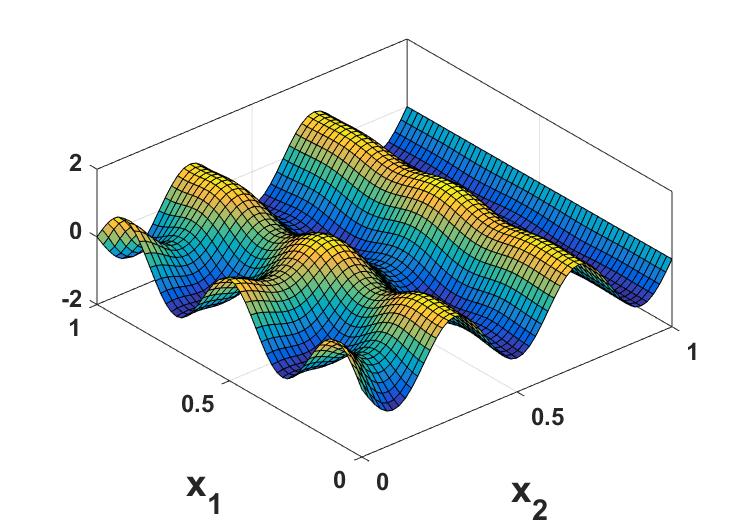}
\end{center}
\caption{Proposed mean function $\mu(\cdot)$ in \eqref{eq:ourmu} for the three boundary cases in Figure \ref{fig:boundary}.}
\label{fig:mean_func_boundary}
\end{figure}

To distinguish which boundaries are known beforehand, let $I^{[0]} \subseteq [d] = \{1, \cdots, d\}$ denote the variables with known \textit{left} boundary condition, and let $I^{[1]} \subseteq [d]$ denote the variables with known \textit{right} boundary condition. With $I^{[0]} = \emptyset$ and $I^{[1]} = \emptyset$, this reduces to the standard setting with no boundary information; with $I^{[0]} = [d]$ and $I^{[1]} = [d]$, this implies knowledge of $f$ along the full boundary of $\mathcal{X}$. Figure \ref{fig:boundary} illustrates this in two dimensions, with known boundaries of $f$ in blue. The left plot shows the case of $I^{[0]} = I^{[1]} = \{1, 2\}$, with the value of $f$ known on all boundaries of $[0,1]^2$. The middle plot shows $I^{[0]} = \{1,2\}, I^{[1]} = \emptyset$, with $f$ known only on the left boundaries of each variable. The right plot shows $I^{[0]} = \{1,2\}$, $I^{[1]} = \{1\}$, with $f$ known on both boundaries of $x_1$ and on the left boundary of $x_2$.

% In our study, we will consider the partial boundary information that at least one of the right and left boundary condition is given for all dimensions. We can further assume without loss of generality that at least the left boundary information on any dimension is given for our analysis which is equivalent to our previous assumption. Now the value of $f$ is known for all left boundary condition plus some right boundary condition:
% $$U_I:=\big[\bigcup_{i=1}^p\{\Bx_{-i}\in [0,1]^{p-1}, x_i=0\}\big]\ \bigcup\ \big[\bigcup_{i\in I}\{\Bx_{-i}\in [0,1]^{p-1}, x_i=1\}\big]$$
% where $I\subseteq\{1,\cdots,p\}$. That is,  $f(\Bx)$ is known  when $x_i=0$ for any $i=1,\cdots,p$ or $x_i=1$ for any $i\in I$.  When $I=\{1,\cdots,p\}$, then $U_I=\partial D$.

To integrate such boundary information, we need to specify two ingredients for BdryGP. First, a mean function $\mu(\cdot)$ is needed which satisfies known boundary conditions on $f$. Second, a covariance function $k(\cdot,\cdot)$ is needed which satisfies $k(\Bx,\Bx)=0$ for any $\Bx \in \mathcal{S}_j^{[0]}$, $j \in I^{[0]}$ and any $\Bx \in \mathcal{S}_j^{[1]}$, $j \in I^{[1]}$. This ensures the BdryGP model satisfies the desired boundary information on $f$ almost surely.

\subsection{Mean function specification}
\label{sec:mean_func}
% Suppose the value of $f$ is given on the partial boundary set $U_I$, we now show how to build a mean function $\mu$ that exactly matches $f$ on $U_I$ and has the same order of smoothness as $f$.
Consider first the specification of the BdryGP mean function $\mu(\cdot)$. We adopt a simple strategy for constructing $\mu(\cdot)$ via an interpolator on known boundary conditions. For a point $\Bx\in \mathcal{X}$, let $\mathcal{P}_j^{[0]}\Bx$ and $\mathcal{P}_j^{[1]}\Bx$ denote the projection of $\Bx$ onto the subspaces $\mathcal{S}_j^{[0]}$ and $\mathcal{S}_j^{[1]}$, respectively. These projected points can be written explicitly as:
\begin{align}
\begin{split}
  &[\mathcal{P}_j^{[0]}\Bx]_k=\begin{cases}
x_k, &\text{if} \ j\neq k\\
0,&\text{if} \ j = k
\end{cases}\ \ \ \ \text{for}\ j \in I^{[0]},\\
&[\mathcal{P}_j^{[1]}\Bx]_k=\begin{cases}
x_k, &\text{if} \ j \neq k\\
1,&\text{if} \ j = k
\end{cases}\ \ \ \ \text{for}\ j \in I^{[1]}.
\end{split}
\end{align}
Furthermore, let $\mathbf{P}(\Bx)=\{ \mathcal{P}_j^{[0]}\Bx : j \in I^{[0]} \} \cup \{ \mathcal{P}_j^{[1]}\Bx : j \in I^{[1]} \}$ be the set of all such projected points of $\Bx$ on known boundaries.

With this, the mean function $\mu(\cdot)$ can then be constructed as:
\begin{equation}
\mu(\Bx)=\phi(\Bx,\mathbf{P}(\Bx)) [\phi(\mathbf{P}(\Bx),\mathbf{P}(\Bx))]^{-1}f(\mathbf{P}(\Bx)).
\label{eq:ourmu}
\end{equation}
where $\phi(\cdot,\cdot)$ is a compactly supported, positive definite, radial basis kernel \citep{Wendland10}. Here, $\mu(\Bx)$ can be interpreted as a GP interpolant at $\Bx$, using boundary information at projected points $\mathbf{P}(\Bx)$ as data. By the interpolation property of GPs, the proposed mean function must therefore satisfy the desired boundary conditions:
\begin{equation}
\left\{\mu(\Bx): \Bx \in \mathcal{S}_j^{[0]}\right\} = \mathcal{F}_j^{[0]}, \; \forall j \in I^{[0]}, \quad \left\{\mu(\Bx): \Bx \in \mathcal{S}_j^{[1]}\right\} = \mathcal{F}_j^{[1]}, \; \forall j \in I^{[1]}.
\label{eq:meanbound}
\end{equation}
We find the radial basis kernel $\phi(\Bx_1,\Bx_2)=\max\{(1-||\Bx_1-\Bx_2||)^\nu,0\}$ \citep{Wendland10} to work well in practice.

% The mean function $\mu(\Bx)$ is in fact a GP-based interpolator of $f$ at $\Bx$ given its projected points $\BX$. Obviously $\mu(\Bx)=f(\Bx)$ for any $\Bx\in U_I$ because of the exactness of GP regression. On the other hand, the smoothness of $\mu(\Bx)$ is determined by $f(\BX(\Bx))$, which has the same order of smoothness as $f$.

Figure \ref{fig:mean_func_boundary} illustrates the proposed mean function $\mu(\cdot)$ in \eqref{eq:ourmu} using the earlier two-dimensional example (with known boundaries marked in blue). From the left plot, which shows the mean function $\mu(\cdot)$ for the full boundary case of $I^{[0]} = I^{[1]} = \{1,2\}$, we see that $\mu(\cdot)$ satisfies the desired boundary conditions from Figure \ref{fig:boundary}. The same is true for the middle and right plots, which shows the proposed $\mu(\cdot)$ given partial boundary information.

\subsection{Covariance function specification}
Consider next the specification of the covariance function $k(\cdot,\cdot)$ for BdryGP. We present below a new BdryMat\'ern covariance function which incorporates boundary information of the form \eqref{eq:bound}. We first discuss the properties of the BdryMat\'ern kernel for modeling, then provide an explicit derivation of this kernel.

% In this subsection, we first present  two one-dimensional non-stationary covariance functions related to Mat\'ern-$\frac{1}{2}$ kernel that give $0$ variance at the end points of the unit interval and discuss some of their properties that are similar to the stationary Matérn-$\frac{1}{2}$ covariance function. 

\subsubsection{The BdryMat\'ern kernel}
For variable $x_j$, the one-dimensional (1-d) BdryMat\'ern kernel is defined as:
 \begin{equation}
%  \label{eq:BdryMatern1d}
\small
     k_{\omega_j}^{\rm BM}(x,y)=\begin{cases}
    \dfrac{\sinh{[\omega_j (x\wedge y)]}\sinh{[\omega_j(1- x\vee y)]}}
    {\sinh(\omega_j)},& \ j \in I^{[0]} \cap I^{[1]} \quad \text{(full)}\\
    \sinh{[\omega_j (x\wedge y)]}\exp{[\omega_j(- x\vee y)]},& \ j \in I^{[0]} \cap \overline{I^{[1]}} \quad \text{(left)}\\
    \exp[\omega_j ( x\wedge y)]\sinh[\omega_j(1- x\vee y)],& \ j \in \overline{I^{[0]}} \cap I^{[1]} \quad \text{(right)}\\
    \exp(-\omega_j{|x-y|}) ,&\ j \in \overline{I^{[0]}} \cap \overline{I^{[1]}} \quad \text{(none)}.\\
    \end{cases}
    \label{eq:BdryMatern1d}
 \end{equation}
 \normalsize
Here, $x \wedge y = \min(x,y)$ and $x\vee y = \max(x,y)$, and $\sinh(\cdot)$ and $\cosh(\cdot)$ are the hyperbolic sine and cosine functions. The first case corresponds to known boundaries for both the left and right endpoints of $x_j$ (i.e., \textit{full} boundary information). The second and third cases correspond to known boundaries for only the left and only the right endpoints of $x_j$, respectively (i.e., \textit{partial} boundary information). The last case is for \textit{no} boundary information on $x_j$; this reduces to the Mat\'ern-1/2 correlation function.

Using \eqref{eq:BdryMatern1d}, we adopt the following product form for the BdryMat\'ern covariance over all $d$ variables:
\begin{equation}
k^{\rm BM}_\omega(\Bx,\By)= \sigma^2 \prod_{j=1}^d k^{\rm BM}_{\omega_j}(x_j,y_j),
\label{eq:BdryMatern}
\end{equation}
where $\sigma^2$ is a variance parameter. This product form of the BdryMat\'ern kernel $k_\omega^{\rm BM}(\Bx,\By)$ yields a very useful native space, which can be connected to FEM for proving improved GP convergence rates.
% $e^{-\omega{|x-y|}}$ is called   Matérn-$\frac{1}{2}$ covariance function \cmtS{why is there the 1/2 in equation (9)?} or exponential kernel and we will call the parameter $\omega$ of $k_\omega$ the generalized wavelength parameter.
% We can see the symmetry between left and right boundary condition this is the reason that in the following content, we only need to focus on the left boundary case for our analysis.

To see why the BdryMat\'ern kernel \eqref{eq:BdryMatern} can incorporate boundary information, consider a simple 1-d setting with $\sigma^2 = 1$ and $\omega=1$. Figure \ref{fig:variance_plot} visualizes the process variance $k_\omega^{\rm BM}(x,x)$ as a function of $\Bx$ over $[0,1]$, for the first three cases in \eqref{eq:BdryMatern1d}. The left plot shows $k_\omega^{\rm BM}(x,x)$ for the full boundary case, where \textit{both} left and right boundaries are known. Here, the process variance equals zero at the endpoints $x=0$ and $x=1$, meaning the BdryGP constrains sample paths to satisfy the left and right boundaries almost surely. The middle plot show $k^{\rm BM}_\omega(x,x)$ when \textit{only} the left boundary is known. Here, the process variance equals zero only when $x=0$, meaning all BdryGP sample paths satisfy the left boundary almost surely. A similar interpretation holds for the right plot, where \textit{only} the right boundary known.

The wavelength parameter $w_j$ in the BdryMat\'ern kernel \eqref{eq:BdryMatern} plays a similar role as the scale parameter in the Mat\'ern kernel: it controls the smoothness of sample paths from the BdryGP. To visualize this, Figure \ref{fig:wavelength_plot} plots the process covariance $k_\omega^{\rm BM}(0.5,x)$ between a point $x \in [0,1]$ and a fixed point at 0.5, for difference choices of wavelength $\omega$. The left plot shows, as $\omega \rightarrow \infty$, this covariance  converges to zero everywhere \textit{except} at $x=0.5$, which suggests that for larger wavelengths $\omega$, the sample paths from BdryGP become more rugged. The right plot shows, as $\omega \rightarrow 0^+$, this covariance converges to zero everywhere, \textit{including} at $x=0.5$. This suggests that the process variance $k_\omega^{\rm BM}(x,x)$ becomes smaller as $\omega \rightarrow 0^+$, which results in smoother sample paths.

The 1-d BdryMat\'ern kernel \eqref{eq:BdryMatern1d} also has an inherent connection to the covariance functions for the Brownian bridge and the Brownian motion. Suppose either the left or right boundary is known for variable $x_j$. Taking wavelength $\omega_j \rightarrow 0^+$ for the normalized BdryMat\'ern kernel, we get:
\begin{equation}
\label{eq:brownian_kernel1d}
    k_j^{\rm BR}(x,y)=\lim_{\omega_j \to 0^+}\frac{k_{\omega_j}(x,y)}{\omega_j}=\begin{cases}
    (x\wedge y)(1-x\vee y),& \; j \in I^{[0]} \cap I^{[1]} \quad \text{(full)}\\
    x\wedge y,& \; j \in I^{[0]} \cap \overline{I^{[1]}} \quad \text{(left)}\\
    1-x\vee y,& \; j \in \overline{I^{[0]}} \cap I^{[1]} \quad \text{(right)}.
     \end{cases}
\end{equation}
The first case is the covariance function of a Brownian bridge, and the second and third cases are variants of the covariance function for a Brownian motion. We will call $k_j^{\rm BR}(x,y)$ the 1-d \textit{Brownian kernel}, and its product form:
\begin{equation}
k^{\rm BR}(\Bx,\By) = \prod_{j=1}^d k_j^{\rm BR}(x_j,y_j)
\label{eq:brownian_kernel}
\end{equation}
the \textit{Brownian kernel}. The link between the BdryMat\'ern kernel (used in BdryGP) and the Brownian kernel will serve as the basis for proving improved convergence rates via finite-element modeling. We note that the Brownian kernel is \textit{not} used for modeling purposes, but rather as a theoretical tool for bridging the BdryGP model with FEM.

% Obviously, $k_\omega(x,y)$ is symmetry in $x$ and $y$ and from direct calculation, we can also see that  $k_\omega(x,y)=0$ if one of $x$ and $y$   is the end point of the unit interval $[0,1]$ at which boundary constraint must be satisfied. Furthermore, the value of $k_\omega(x,y)$ is decreasing with increasing distance between $x$ and $y$ in a non-stationary way. This is because of two reasons. Firstly, as $|x-y|$ become larger, then either  $x$ or $y$, or both, is closer to the end point of $[0,1]$ so either $x\wedge y$ or $1-x\vee y$, or both, becomes smaller; secondly, both $\sinh$ and $\cosh$ are monotonic functions on $[0 1]$, which will preserve  the changing direction of $k_\omega(x,y)$ with $|x-y|$.

\begin{figure}
\begin{center}
\includegraphics[width=0.3\textwidth]{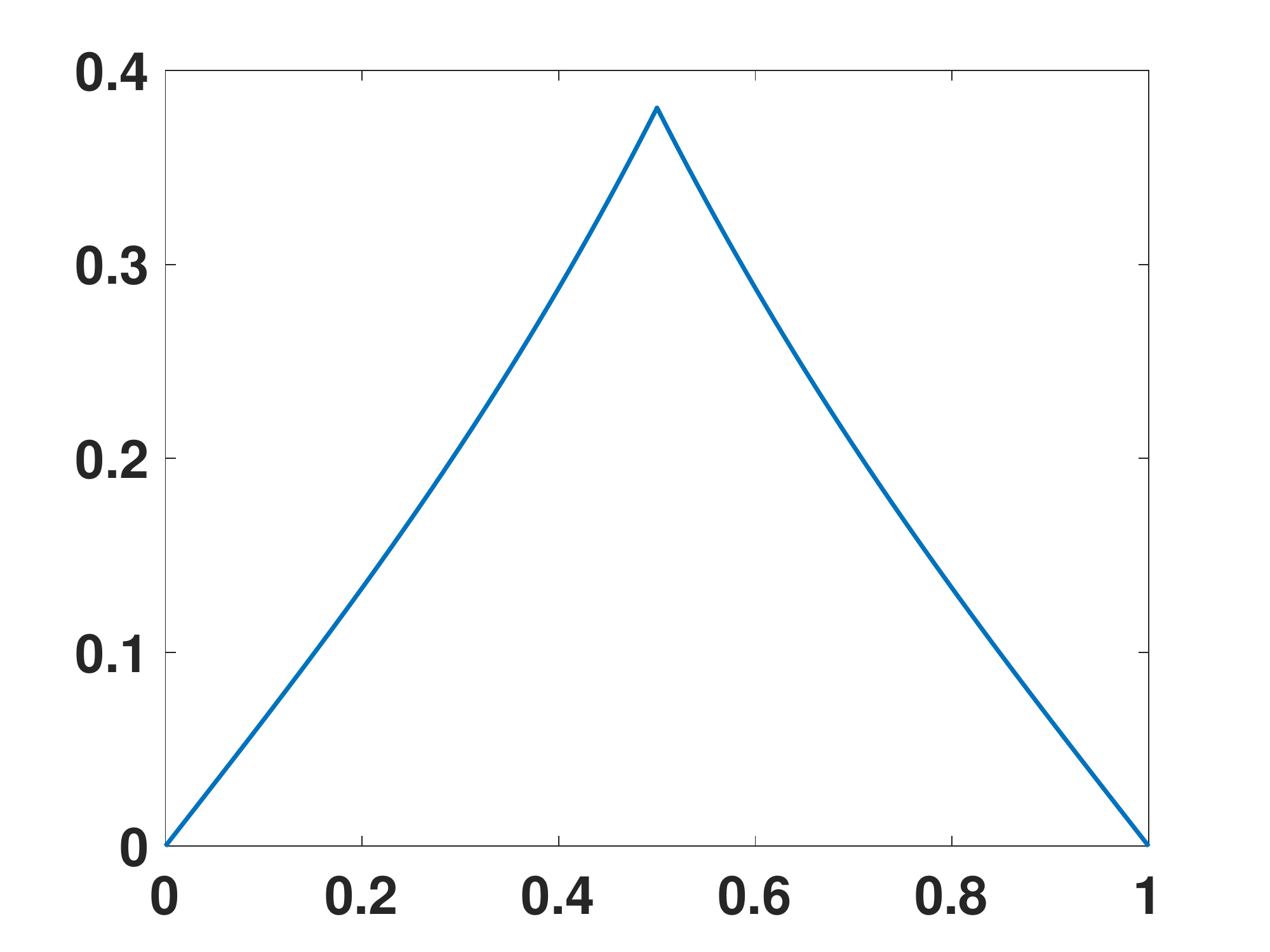}
\includegraphics[width=0.3\textwidth]{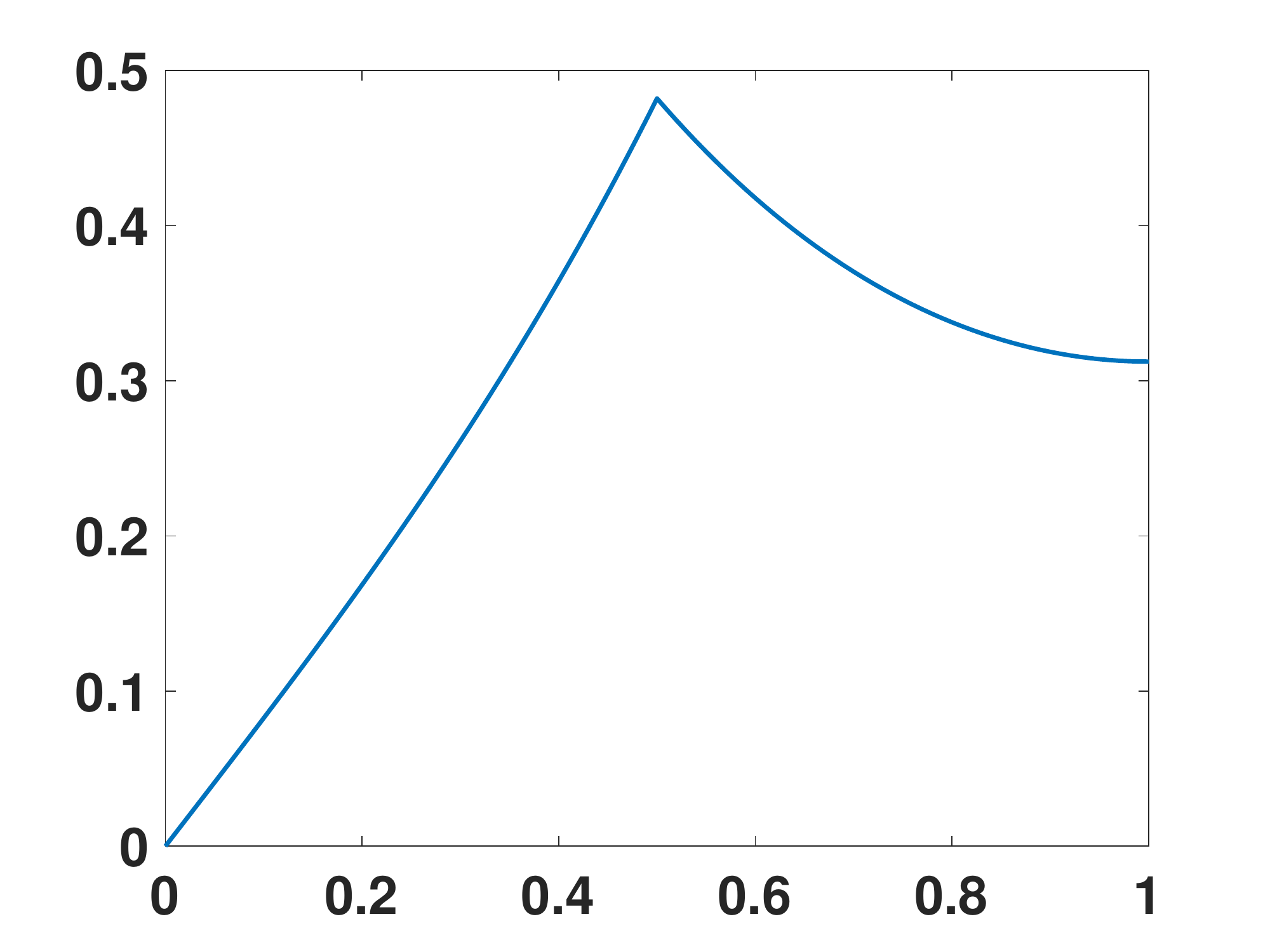}
\includegraphics[width=0.3\textwidth]{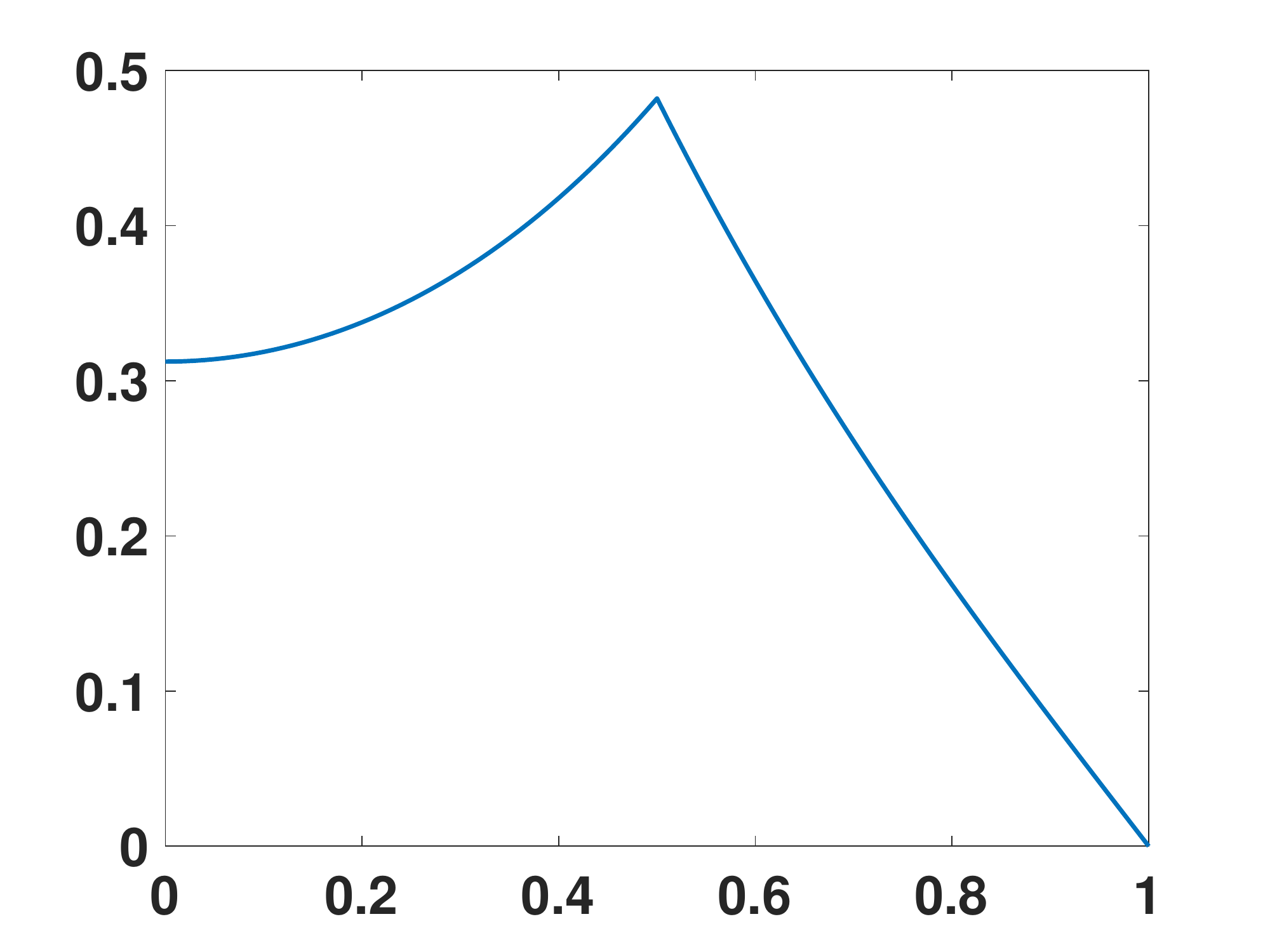}
\end{center}
\caption{Visualizing the BdryGP variance $k^{\rm BM}_\omega(x,x)$ over $x \in [0,1]$ for full boundary (left), left boundary (middle), and right boundary (right) information.}
\label{fig:variance_plot}
\end{figure}

\begin{figure}
\begin{center}
\includegraphics[width=0.4\textwidth]{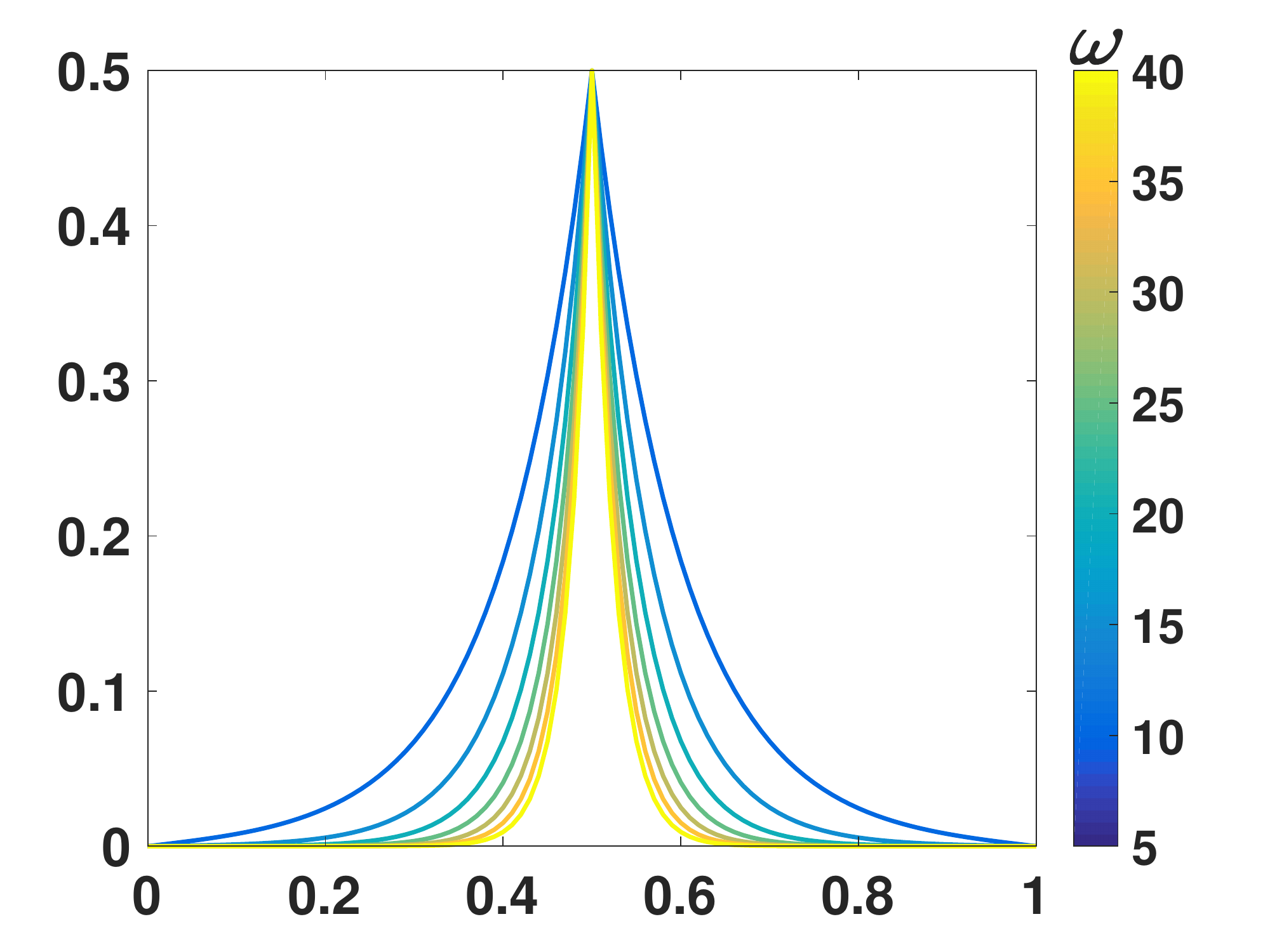}
\includegraphics[width=0.4\textwidth]{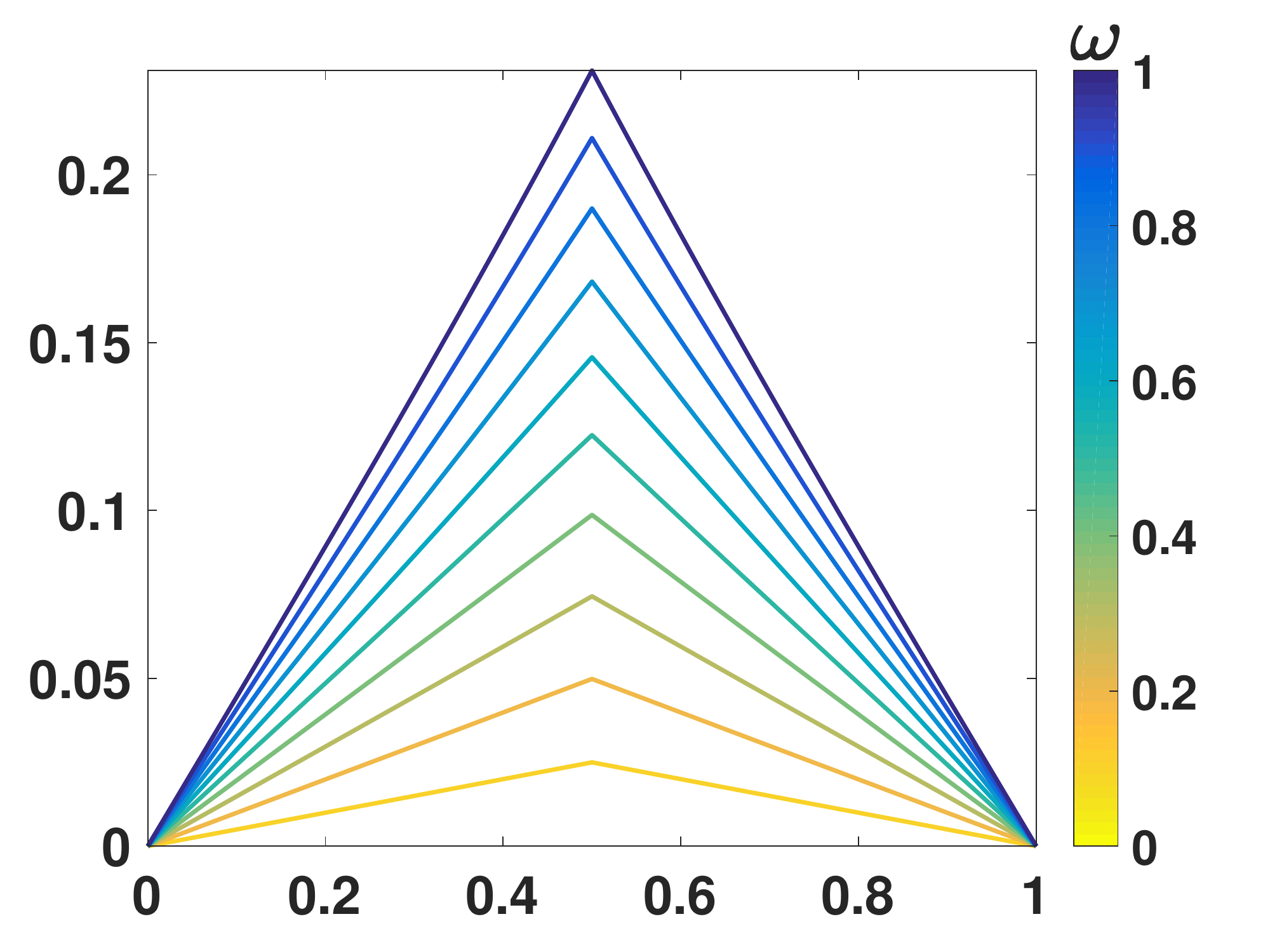}
\end{center}
\caption{Visualizing the BdryGP covariance $k_\omega^{\rm BM}(0.5,x)$ for full boundary information, with wavelength $\omega \rightarrow \infty$ (left) and $\omega \rightarrow 0^+$ (right).}
\label{fig:wavelength_plot}
\end{figure}

% Another important property of $k_\omega$ worth discussing is how $k_\omega(x,y)$ changes with the generalized wavelength $\omega$. Figure \ref{fig:wavelength_plot} shows that as $\omega$ converges to infinity and $0$, the full boundary condition $k_\omega(\frac{1}{2},x)$ converges to $\frac{1}{2}\bold{1}_{\{x=y\}}$ and $0$  respectively. In the case that $\omega\to \infty$, $k_\omega$ and exponential kernel coincides that the exponential kernel also converges to $\frac{1}{2}\bold{1}_{\{x=y\}}$. However, as $\omega \to 0$, the $k_\omega$ with full boundary condition converges to 0 while the exponential kernel converges to a constant. This is because $k_\omega$ has one less degree of freedom due to the boundary condition. A zero-valued wavelength parameter leads to a perfectly correlated random field. Because our GP is fixed at the boundary, a perfectly correlated random field becomes a deterministic constant function over $[0,1]$. From this perspective,  the generalized wavelength parameter still represents the parametrization of function smoothness  in the sense that the larger the wavelength, the smaller the correlation between two points.

% To generalize our kernel to higher dimension, we only need to take tensor product of 1D kernels:
% $$k_\omega(\Bx,\By)=\prod_{i=1}^pk_{\omega_i}(x_i,y_i)$$
% according to given boundary condition on each dimension. The tensor product of $k_\omega$ gives a very useful native space with boundary condition, which is commonly used in numerically solving PDE as we will show later.

\subsubsection{Derivation under boundary conditions}
% The strategy is to formulate the Mat\'ern kernel as the solution of a differential equation via its reproducing property, then solve for the kernel which satisfies the same equation under boundary constraints.
We now provide a derivation of the 1-d BdryMat\'ern kernel \eqref{eq:BdryMatern1d}. Consider the 1-d Mat\'ern kernel:
\begin{equation}
k_{\nu,\omega}(x,y) = \frac{2^{1-\nu}}{\Gamma(\nu)}(\sqrt{2\nu}\omega|x-y|)^\nu K_\nu(\sqrt{2\nu}\omega|x-y|),
\label{eq:matern1d}
\end{equation}
where $\nu$ is the smoothness parameter, $\omega$ is the scale parameter, and $K_\nu$ is the modified Bessel function \citep{abramowitz1965handbook}. Let $S_{\nu,\omega}(s)$ denote the spectral density of $k_{\nu,\omega}$ \citep{cressie1991}.
% The spectral density of this kernel can be written \citep{cressie1991} as:
% \begin{equation}
% S_{\nu,\omega}(s) = \bigg[C_{\nu,\omega}+C_{\nu,\omega}'s^2\bigg]^{-\nu-\frac{1}{2}}, \; C_{\nu,\omega} = 2\nu\omega^2C_{\nu,\omega}'', \; C'_{\nu,\omega} = 4\pi^2 C_{\nu,\omega}'',
% \label{eq:maternspec1d}
% \end{equation}
% where $C_{\nu,\omega}'' =\{{2\pi^{{1}/{2}}\Gamma(\nu+{1}/{2})(2\nu)^{\nu}\omega^{2\nu}}/{\Gamma(\nu)}\}^{-1/(\nu+{1}/{2})}$. Letting $\hat{f}$ be the Fourier transform of $f$, the RKHS of $k_{\nu,\omega}$ can be shown \citep{Wendland10} to be:
% \begin{equation}
% \CalH_{k_{\nu,\omega}}=\left\{f\in L^2(\Real)\bigcap C(\Real): \frac{\hat{f}}{\sqrt{S_{\nu,\omega}}}\in L^2(\Real)\right\}.
% \label{eq:maternrkhs}
% \end{equation}
% where $C(\Real)$ denotes the space of continuous function on $\Real$.
With $\simeq$ denoting equality up to an independent constant, the inner product of the RKHS $\CalH_{k_{\nu,\omega}}$ can be written as:
\begin{equation}
\langle f,g \rangle_{k_{\nu,\omega}}\simeq\int_\Real \frac{\hat{f}(s)\hat{g}(s)}{S_{\nu,\omega}(s)}ds, \quad f, g \in \CalH_{k_{\nu,\omega}},
\label{eq:materninprod}
\end{equation}
where $\hat{f}$ is the Fourier transform of $f$. Let $m=\nu+\frac{1}{2}$, and suppose $m\in \NatInt$. Equation \eqref{eq:materninprod} then simplifies to:
\begin{equation}
    % \langle f,g \rangle_{k_{\nu,\omega}} \simeq \sum_{n=0}^m(-1)^n{m \choose n}C^{m-n}_{\nu,\omega}K^{n}_{\nu,\omega}\int f(x) D^{2n}g(x) dx = \langle f,\mathcal{L}_{\nu,\omega}g \rangle_{L^2},
    \langle f,g \rangle_{k_{\nu,\omega}} \simeq \langle f,\mathcal{L}_{\nu,\omega}g \rangle_{L^2},
    \label{eq:materninprod2}
\end{equation}
where $D^l= d/dx^l$, $\mathcal{L}_{\nu,\omega}$ is the self-adjoint differential operator:
\begin{equation}
\mathcal{L}_{\nu,\omega}:=\sum_{l=0}^m(-1)^l{m \choose l}C^{m-l}_{\nu,\omega}{C'}^{l}_{\nu,\omega}D^{2l}, \; C_{\nu,\omega} = 2\nu\omega^2C_{\nu,\omega}'', \; C'_{\nu,\omega} = 4\pi^2 C_{\nu,\omega}'',
\label{eq:diffop}
\end{equation}
and $C_{\nu,\omega}'' =\{{2\pi^{{1}/{2}}\Gamma(\nu+{1}/{2})(2\nu)^{\nu}\omega^{2\nu}}/{\Gamma(\nu)}\}^{-1/(\nu+{1}/{2})}$.

With this, the reproducing property of the RKHS gives:
\begin{equation}
f(x)=\langle k_{\nu,\omega}(x,\cdot),f\rangle_{k_{\nu,\omega}}=\int_\Real k_{\nu,\omega}(x,y)\mathcal{L}_{\nu,\omega}f(y) dy, \quad \forall f \in \CalH_{k_{\nu,\omega}}.
\label{eq:reprod}
\end{equation}
This suggests that the Mat\'ern kernel $k_{\nu,\omega}(x,y)$ is the Green's function of the differential operator $\mathcal{L}_{\nu,\omega}$, and can therefore be \textit{uniquely} obtained by solving the following differential equation for $k$ \citep{ODEHandbook}:
\begin{equation}
\mathcal{L}_{\nu,\omega}k(x,y)=\delta(x-y),
\label{eq:diff}
\end{equation}
where $\delta(x-y)$ is the Dirac delta function. This link serves as the basis for deriving the 1-d BdryMat\'ern kernel.

Consider next the case of full boundary information on $f$. This information can be incorporated into the Mat\'ern RKHS $\CalH_{k_{\nu,\omega}}$ by restricting all functions $f \in \CalH_{k_{\nu,\omega}}$ to satisfy $f(0) = f(1) = 0$. The corresponding kernel $k$ for this constrained function space must satisfy the reproducing property:
\begin{equation}
f(x) = \int_\Real k(x,y)\mathcal{L}_{\nu,\omega}f(y) dy, \quad \forall f \in \CalH_{k_{\nu,\omega}}, \quad f(0) = f(1) = 0,
\label{eq:reprod2}
\end{equation}
or equivalently, the following constrained differential equation:
\begin{equation}
\mathcal{L}_{\nu,\omega}k(x,y)=\delta(x-y), \quad k(0,y) = k(1,y) = 0.
\label{eq:diff2}
\end{equation}
For the cases of only left and only right boundary information, a similar reasoning gives the differential equations:
\begin{align}
\mathcal{L}_{\nu,\omega}k(x,y)=\delta(x-y), &\quad k(0,y) = 0, \quad \text{and}\label{eq:diff3}\\
\mathcal{L}_{\nu,\omega}k(x,y)=\delta(x-y), &\quad k(1,y) = 0,
\label{eq:diff4}
\end{align}
respectively.

The following proposition shows that the 1-d BdryMat\'ern kernel cases in \eqref{eq:BdryMatern1d} satisfy the above differential equations with $\nu = 1/2$, with their corresponding RKHS related to the weighted first-order Sobolev space:
\begin{equation}
\label{eq:sobolev1}
\mathcal{H}^1_{\omega} = \left\{f : \omega \|f\|_{L^2}^2 + \frac{1}{\omega} \|Df\|_{L^2}^2 < \infty\right\}.
\end{equation}
% As a result, we can obtain the reproducing kernel of the native space of Mat\'ern kernel with boundary condition by imposing boundary condition on the previous equation. For example, in the case of full boundary condition: 
% $$\mahtcal{H}=\mathcal{H}_{k_\nu}\bigcap\{f:[0,1]\to \Real; f(0)=f(1)=0\}$$
% we only need to solve the  equation with the boundary condition $k(0,y)=k(1,y)=0$.

% In the case that $m=1$, the differential operator becomes
% \begin{equation}
% \mathcal{L}_{\frac{1}{2},\omega}=\omega-\frac{1}{\omega}D^2
% \label{eq:diffbound}
% \end{equation}
% so we can solve the  differential equation with associated boundary condition to get the kernel in (\ref{eq:BdryMatern1d}) with which the associated native space is exactly the weighted first order Sobolev space: 
% $$\CalH_{k_\omega}=\{f: \omega||f||^2_{L^2}+\frac{1}{\omega}||Df||^2_{L^2}<\infty; f\ \text{satisfies given boundary condition}\}.$$

\begin{prop}
Suppose $\nu = 1/2$. The unique kernel $k$ solving \eqref{eq:diff2}, \eqref{eq:diff3} and \eqref{eq:diff4} are the first three cases of the 1-d BdryMat\'ern kernel \eqref{eq:BdryMatern1d}, with corresponding RKHS equal to $\mathcal{H}^1_{\omega}$ with the additional constraint of $\{f(0)=0,f(1)=0\}$, $\{f(0)=0\}$ and $\{f(1)=0\}$, respectively.
\label{prop:BdryMatern1d}
\end{prop}

\begin{proof}
This can be proven by simply showing each of the first three cases of the 1-d BdryMat\'ern kernel \eqref{eq:BdryMatern1d} satisfies its corresponding differential equation \eqref{eq:diff2}, \eqref{eq:diff3} or \eqref{eq:diff4}. Since a solution to these differential equations is unique, the uniqueness of the kernel then follows. The RKHS claim follows from the equivalence between the differential equations and its corresponding reproducing property.

% The norm defined in equation (\ref{eq:sobolev1}) is  associated to the reproducing property (\ref{eq:reprod2}). So, for example, we can solve the full boundary case, which is equation (\ref{eq:diff2}), to prove the reproducing property:
% \begin{align*}
%     &\ \ \ \ \int_0^1k_{\omega}^{\text{BM}}(x,s)[\omega-\frac{1}{\omega}\frac{\partial^2}{\partial s^2}])f(s)ds\\
%     &=\omega\int_0^1 k_{\omega}^{\text{BM}}(x,s)f(s)ds-\frac{1}{\omega}k_{\omega}^{\text{BM}}(x,s)f'(s)\bigg|_{s=0}^{1}+\int_0^1\frac{1}{\omega}\frac{\partial}{\partial s}k_{\omega}^{\text{BM}}(x,s)f'(s)ds\\
%     &=f(x).
% \end{align*}
\end{proof}

\noindent This shows that the proposed BdryMat\'ern kernel indeed inherits the same smoothness properties as the Mat\'ern-1/2 kernel, while also satisfying the desired boundary conditions.

Unfortunately, the same kernel derivation does not appear to extend for more general smoothness parameters $\nu > 1/2$, since more constraints are needed on kernel $k$ in order to solve the corresponding differential equation. For example, when $\nu = 3/2$, a unique solution to \eqref{eq:diff2} requires boundary conditions on both $k$ and its first derivative, which implies further boundary information on $f$ than the Dirchlet boundaries assumed in the paper.

\section{Interpolation: BdryGP and FEM}
\label{sec:fem}
With the BdryGP in hand, we now reveal a useful connection between FEM and the BdryGP predictor. This connection allows us to extend results from FEM to prove improved convergence rates for BdryGP.
% Further details on FEM can be found in \cite{Johnson87}, \cite{Brenner08}, and \cite{Iserles09}.

% because the key step for the proof is to expand the BLUE or the MSE to some series which is closely related to FEM.
% and let $g \in \mathcal{H}^{-1}$, where $\mathcal{H}^{-1}$ is the dual space of $\CalH$
% Let us define the bilinear form $B(u,v):=\langle u,\mathcal{L}v \rangle_{\mathcal{H}}$, with corresponding functional $F(u):=B(u,u)-\langle  f,u\rangle_{{\mathcal{H}}}$. The weak solution of the PDE \eqref{eq:pde} then becomes the minimizer of the functional $F(u)$ over elements $u \in \mathcal{H}$. The Lax-Milgram theorem \citep{Evans15}  gives conditions under which the weak solution exists and is unique:
% \begin{theorem}[Lax-Milgram Theorem]
% \label{thm:Lax-Milgram}
% \cmtS{Is this used in any of our results? If not, I suggest remove this and just mention in text. paper is too long and tangential.} Assume the bilinear map has positive constants $K_1$ and $K_2$ such that:
%     \begin{align}
%     \begin{split}
%     |B(u,v)|&\leq K_1 ||u||_{\mathcal{H}}||v||_{\mathcal{H}}, \quad \forall u,v\in\mathcal{H},\\
%     |B(u,u)|&\geq K_2 ||u||_{\mathcal{H}}^2, \quad  \quad \forall u\in\mathcal{H}.
%     \end{split}
%     \label{eq:coercive}
%     \end{align}
% Let $f\in\mathcal{H}^{-1}$. Then there exists a unique element $u\in\mathcal{H}$ such that
% \begin{equation*}
%     B(u,v)=\langle  f,v\rangle_{{\mathcal{H}}}, \quad \forall v\in\mathcal{H}.
% \end{equation*}
% \end{theorem}

\subsection{Finite-Element Modeling}
We begin with a brief review of FEM. Consider first the following partial differential equation (PDE) system:
\begin{equation}
\begin{cases}
    \mathcal{L} f(\Bx) = g(\Bx), \quad &\Bx \in \mathcal{X},\\
    f(\Bx) = 0, \quad &\Bx \in \partial \mathcal{X}.
\end{cases}
\label{eq:pde}
\end{equation}
Here, $f$ is a solution on a Hilbert space $\mathcal{V}$, $\partial \mathcal{X}$ is the boundary of $\mathcal{X}$, and $\mathcal{L}$ is a differential operator on $\mathcal{V}$. Under regularity conditions, the Lax-Milgram Theorem \citep{Evans15} ensures the existence of a unique weak solution satisfying \eqref{eq:pde}.
% Suppose we have a finite-dimensional subspace $\mathcal{V}$ of $\mathcal{H}$. Assuming the bilinear form $B(u,v) = \langle u,\mathcal{L}v \rangle_{\mathcal{H}}$ satisfies condition \eqref{eq:coercive}, the Lax-Milgram Theorem then ensures the existence of a unique weak solution $u\in\mathcal{V}$. Much of FEM concerns the construction of the finite-dimensional subspace $\mathcal{V}$, and the corresponding minimizer of $F(u)$ over elements $u \in \mathcal{V}$.

%By applying the Sobolev extension theorem and variable transformation, we can, without loss of generality, assume that the domain $D\subset\mathbb{R}^d$ is the d-hyper-cube $T^d$. 

% \begin{figure}
% \begin{center}
% \begin{tikzpicture}[scale=0.8]
% \draw (0,1) grid (5,5);
% \end{tikzpicture}
% \hspace{2cm}
% \begin{tikzpicture}[scale=2.5]
% \draw (0,0)  node[below left] {$A$} --
% (1,0) node[below right] {$B$} --
% (1,1) node[above right] {$C$} --
% (0,1) node[above left] {$D$} -- cycle;
% \filldraw (0,0) circle (1pt);
% \filldraw (0,1) circle (1pt);
% \filldraw (1,1) circle (1pt);
% \filldraw (1,0) circle (1pt);
% \end{tikzpicture}
% \end{center}
% \caption{Mesh of $\mathcal{X}$ with 20 rectangles and vertices $A,B,C,D$ of some $K_i$. \cmtS{needed? figure takes valuable space. paper too long.}}
% \label{fig:mesh}
% \end{figure}

The idea behind FEM is to approximate \eqref{eq:pde} on a discretization of $\mathcal{X}$. This requires two ingredients: a discretization mesh on $\mathcal{X}$, and a finite-dimensional function space constructed from this mesh. Given a multi-index $\Balpha = (\alpha_1, \cdots, \alpha_d) \in \mathbb{N}^d$, let $\mathcal{X}$ be discretized on the full grid mesh:
\begin{equation}
    \BX_{\Balpha} = \left\{\Bx_{\Balpha,\Bbeta}=[\beta_j2^{-\alpha_j}]_{j=1}^d : \beta_j \in \mathcal{B}_{\Balpha} \right\},
    \label{eq:fullgrid}
\end{equation}
where $\mathcal{B}_{\Balpha}$ is the index set:
\begin{equation}
\mathcal{B}_{\Balpha} = \left\{ \mathbf{1}_{\{j \in I^{[0]}\}}, \cdots, 2^{\alpha_j} - \mathbf{1}_{\{j \in I^{[1]}\}} \right\}.
\label{eq:indexset}
\end{equation}
The mesh size of $\BX_{\Balpha}$ then becomes $h_{\Balpha}=(h_{\alpha_1},\cdots,h_{\alpha_d})=(2^{-\alpha_1},\cdots,2^{-\alpha_d})$.

% Figure \ref{fig:mesh} visualizes this in $d=2$ dimensions. Here, $\mathcal{P}(K_i)=\text{span}\{1,x_1,x_2,x_1x_2\}$ , and for any $u\in\mathcal{P}(K_i)$, $u$ can be determined by the values $u(A),u(B), u(C)$ and $u(D)$, where $A,B,C$ and $D$ are vertices of $K_i$.

Next, given the mesh $\BX_{\Balpha}$, let $\mathcal{V}_{\Balpha}$ be the finite-dimensional function space spanned by first-order polynomials within each hypercube formed by $\BX_{\Balpha}$ (we discuss $\mathcal{V}_{\Balpha}$ in greater detail in Section \ref{sec:funcspace}). The FEM solution is then defined as the projection of the weak solution $f$ on the finite-dimensional space $\mathcal{V}_{\Balpha}$. Using a connection to Lagrange polynomial interpolation (see Chapter 15.2 of \citealp{Wendland10}), this FEM solution can be equivalently represented as:
\begin{equation}\label{eq:LagrangePoly}
    \mathcal{I}_{\Balpha} f = \sum_{\Bx_{\Balpha,\Bbeta}\in\BX_{\Balpha}} f(\Bx_{\Balpha,\Bbeta})\phi_{\Balpha,\Bbeta}(\Bx),
\end{equation}
where:
\begin{equation}
\phi_{\Balpha,\Bbeta}(\Bx)=\prod_{j=1}^d\max\left\{1-\frac{|x-x_{\alpha_j,\beta_j}|}{2^{-\alpha_j}},0\right\}
\end{equation}
are piecewise-linear basis functions over each cube.

\subsection{FEM and the Brownian kernel}
We now reveal a novel connection between the FEM interpolator \eqref{eq:LagrangePoly} and the GP predictor under the Brownian kernel $k^{\rm BR}$, first for full grid designs then for sparse grid designs.

\subsubsection{Full Grids}
We first make this connection for full grid designs:

% the BLUE $\hat{u}$ of a zero-mean GP with the Brownian kernel conditioned on the full grid design $\BX_\alpha$ equals the interpolation projector $\mathcal{I}_\alpha$ defined in equation ). The following theorem tells us how to construct $\mathcal{I}_\alpha$ from a GP model.

\begin{theorem}\label{thm:LagrangeKrigingEquivalent}
Suppose $I^{[0]} \cup I^{[1]} = [d]$, and assume the full grid design $\BX_{\Balpha}$ with $n=|\BX_{\Balpha}|$ points. For any $f\in \CalH^{1,c}_{mix}$, the posterior predictor $\hat{f}_n^{\rm BR}$ of a GP with mean function $\mu(\cdot)$ in \eqref{eq:ourmu} and Brownian kernel $k^{\rm BR}$ is equivalent to the FEM solution $\mathcal{I}_{\Balpha} f$, i.e.:
\begin{equation}\label{eq:LagrangeKrigingEquivalent}
\hat{f}_n^{\rm BR}(\cdot) = \mu(\cdot)+k^{\rm BR}(\cdot,\bold{X_{\Balpha}}) [k^{\rm BR}(\bold{X}_{\Balpha},\bold{X}_{\Balpha})]^{-1} \left[f(\bold{X}_{\Balpha})-\mu(\bold{X}_{\Balpha})\right]=\mathcal{I}_{\Balpha} f(\cdot).
\end{equation}
% $\bold{X}_{\Balpha}$ be a full grid design with partial-boundary condition on $\mathcal{X}$ with $I^{[0]}=[d]$. Let $k$ be the Brownian kernel defined on $\mathcal{X}$ as follows:
% \begin{equation*}
%     k(\bold{x},\bold{y})=\prod_{i\in I^{[1]}} x_i\wedge y_i\bigg(1-x_i\vee y_i\bigg)\prod_{i\not\in I^{[1]}}x_i\wedge y_i .
% \end{equation*}
% where $\mathcal{I}_\alpha$ is the interpolation operator defined in (\ref{eq:LagrangePoly}) with which the lattice points of the mesh $\{K_i\}$ is on $\BX_\alpha$.
\end{theorem}

\noindent In other words, assuming $I^{[0]} \cup I^{[1]} = [d]$ (i.e., there exists left or right boundary information for each of the $d$ variables), the predictor $\hat{f}_n^{\rm BR}$ for a GP with Brownian kernel $k^{\rm BR}$ is equivalent to the FEM solution $\mathcal{I}_{\Balpha} f$, under the full grid design (or mesh) $\BX_{\Balpha}$.

The key idea in proving Theorem \ref{thm:LagrangeKrigingEquivalent} is to show that, under the Brownian kernel $k^{\rm BR}$, the matrix inverse $[k^{\rm BR}(\bold{X},\bold{X})]^{-1}$ has an explicit closed-form expression. Under this expression, the desired equivalence can be shown via an inductive argument on dimension $d$. This result can be viewed as an extension of Proposition 2 in \cite{DingZhang18}.

% $[k^{\rm BR}(\bold{X},\bold{X})]^{-1}f(\bold{X})$ is the hierarchical surplus of $\mathcal{I}_{\Balpha}$, we can then show the desired equivalence in \eqref{eq:LagrangeKrigingEquivalent}

\begin{proof}
Without loss of generality (WLOG), we assume the setting of only known left boundaries, i.e., $I^{[0]} = [d]$, $I^{[1]} = \emptyset$, since the setting of $I^{[0]} \cup I^{[1]} = [d]$ follows immediately. Furthermore, since the mean function $\mu(\cdot)$ in \eqref{eq:ourmu} satisfies the desired boundary conditions, we can simply show that the claim holds for a zero-mean GP with Brownian kernel $k^{\rm BR}$, under a boundary condition of zero. Let $k(\Bx,\By):=k^{\text{BR}}(\Bx,\By)$ and $k_j(x_j,y_j):=k^{\text{BR}}_j(x_j,y_j)$, $j = 1, \cdots, p$. We first prove the theorem for the base cases of $d=1$ and $d=2$, then show the claim holds for $d>2$ via induction.

Consider first the base case of $d=1$. Under the assumption of known left boundaries, $\bold{X}_{\Balpha}=\{x_i=i 2^{-\Balpha}:i=1,\cdots , 2^{\Balpha}\}$ with $x_{0}=0$ and $n:=|\BX_{\Balpha}|$. By Theorem 2 of \cite{DingZhang18}, $k^{-1}(\bold{X}_{\Balpha},\bold{X}_{\Balpha})$ is a symmetric tridiagonal matrix with entries
\begin{align*}
    [k^{-1}(\bold{X}_{\Balpha},\bold{X}_{\Balpha})]_{i,i}&=
    \begin{cases}
     \frac{x_{i+1}-x_{i-1}}{(x_{i}-x_{i-1})(x_{i+1}-x_{i})}, & \text{if} \ i< n\\
     \frac{1}{x_n-x_{n-1}}, & \text{if} \ i=n
    \end{cases},\\
    [k^{-1}(\bold{X}_{\Balpha},\bold{X}_{\Balpha})]_{i-1,i}&=[k^{-1}(\bold{X}_{\Balpha},\bold{X}_{\Balpha})]_{i,i-1}=-\frac{1}{x_{i}-x_{i-1}}
\end{align*}
for $i=1,2,\cdots,n$. Given any point $x \in [0,1]$, assume $x_{i}<x<x_{i+1}$. By straightforward calculations, we have
$$
    \hat{f}^{\text{BR}}_n(x)=\frac{x_{i+1}-x}{x_{i+1}-x_{i}}f(x_{i})+\frac{x-x_{i}}{x_{i+1}-x_{i}}f(x_{i+1})=\sum_{x_{{\Balpha},\beta}\in \BX_{\Balpha}}f(x_{{\Balpha},\beta})\phi_{{\Balpha},\beta}(x).$$
This proves the base case of $d=1$. 

% be a grid on $\mathbb{R}^2$ with mesh size $(2^{-\alpha_1},2^{-\alpha_2})$
For clarity in the later inductive step, we will also show the case of $d=2$. Here, $\bold{X}_{\Balpha}=\bigtimes_{j=1}^2\bold{X}_{\alpha_j} = \bigtimes_{j=1}^2 \{x_{\alpha_j,1}, x_{\alpha_j,2}, \cdots\}$, and  $k(\bold{x},\bold{y})=k_1(x_1,y_1)k_2(x_2,y_2)$. Let $\Bx$ be a point in $\mathcal{X}$, and let $K_{\Bx}$ be the hypercube in $\BX_{\Balpha}$ containing $\Bx$ with vertices $\{\Bx_{{\Balpha},(i_1,i_2)}:i_j=\beta_j,\beta_j+1,j=1,2\}$. The vector $k^{-1}(\bold{X}_{\Balpha},\bold{X}_{\Balpha})k(\bold{X}_{\Balpha},\bold{x})$  can be decomposed as: 
\begin{align*}
    &\ \ \ \ k^{-1}(\bold{X}_{\Balpha},\bold{X}_{\Balpha})k(\bold{X}_{\Balpha},\bold{x})\\
    &=\big\{k_1^{-1}(\bold{X}_{\alpha_1},\bold{X}_{\alpha_1})\bigotimes k_2^{-1}(\bold{X}_{\alpha_2},\bold{X}_{\alpha_2})\big\}\text{vec}\big([k_1(x_{1},x_{\alpha_1,\beta_1})k_2(x_{2},x_{\alpha_2,\beta_2})]_{\beta_1,\beta_2}\big)\\
    &=\text{vec}\big(k_2^{-1}(\bold{X}_{\alpha_2},\bold{X}_{\alpha_2})[k_1(x_{1},x_{\alpha_1,\beta_1})k_2(x_{2},x_{\alpha_2,\beta_2})]_{\beta_1,\beta_2}k_1^{-1}(\bold{X}_{\alpha_1},\bold{X}_{\alpha_1})\big)
\end{align*}
where vec$\big(\bold{M}\big)$ denotes the vectorization of the matrix $\bold{M}$. Define
\begin{equation*}
    \bold{\Phi}^{[2]}(\Bx):=k_2^{-1}(\bold{X}_{\alpha_2},\bold{X}_{\alpha_2})[k_1(x_{1},x_{\alpha_1,\beta_1})k_2(x_{2},x_{\alpha_2,\beta_2})]_{\beta_1,\beta_2}k_1^{-1}(\bold{X}_{\alpha_1},\bold{X}_{\alpha_1}).
\end{equation*}
%and, similar to the case $d=1$, let:
%\begin{align*}
%    &Z_d^{i+1}(\Bx):=\frac{x_d(1-x_d^{(i)})-(1-x_d)x_d^{(i)}}{x_d^{(i+1)}-x_d^{(i)}}=1-\frac{|x_d-x_d^{(i+1)}|}{x_d^{(i+1)}-x_d^{(i)}}\\
%    &Z_d^{i}(\Bx):=\frac{(1-x_d)x_d^{(i+1)}-x_d(1-x_d^{(i+1)})}{x_d^{(i+1)}-x_d^{(i)}}=1-\frac{|x_d-x_d^{(i)}|}{x_d^{(i+1)}-x_d^{(i)}}.
%\end{align*}
By straightforward calculations similar to the 1-d case, it follows that $\bold{\Phi}^{[2]}(\Bx)$ has only the four non-zero entries:
\begin{align*}
    &\bold{\Phi}^{[2]}_{\beta_1,\beta_2}(\Bx)=\phi_{\alpha_1,\beta_1}(x_1)\phi_{\alpha_2,\beta_2}(x_2),\ \ \bold{\Phi}^{[2]}_{\beta_1+1,\beta_2}(\Bx)=\phi_{\alpha_1,\beta_1+1}(x_1)\phi_{\alpha_2,\beta_2}(x_2),\\ &\bold{\Phi}^{[2]}_{\beta_1,\beta_2+1}(\Bx)=\phi_{\alpha_1,\beta_1}(x_1)\phi_{\alpha_2,\beta_2+1}(x_2), \ \ \bold{\Phi}^{[2]}_{\beta_1+1,\beta_2+1}(\Bx)=\phi_{\alpha_1,\beta_1+1}(x_1)\phi_{\alpha_2,\beta_2+1}(x_2),
\end{align*}
%  and $(\beta_1,\beta_2)$ is the multi-index which satisfies the condition that $\Bx\in K_i$ where $K_i$ is a hyper-cube with vertices $\{\Bx_{{\Balpha},(i_1,i_2)}:i_j=\beta_j,\beta_j+1,j=1,2\}$.
where $\phi_{{\Balpha},\beta}$ is the hat function defined previously. Thus, the predictor $\hat{f}_n^{\rm BR}$ can be rewritten as:
\begin{align*}
    \hat{f}^{\text{BR}}_n(\Bx)&=\text{vec}\big[\Phi^{[2]}(\Bx)\big]^\intercal f(\BX_{\Balpha})\\
    &=\sum_{i_1=\beta_1}^{\beta_1+1}\sum_{i_2=\beta_2}^{\beta_2+1}\bold{\Phi}^{[2]}_{i_1,i_2}(x)f(x_{\alpha_1,i_1},x_{\alpha_2, i_2})\\
    &=\sum_{\Bx_{{\Balpha},{\Bbeta}}\in\BX_{\Balpha}}f(\Bx_{{\Balpha},\beta})\phi_{{\Balpha},\beta}(\Bx),
\end{align*}
which proves the theorem for $d=2$.

Consider next the inductive step on $d$. Here, the full grid becomes $\bold{X}_{\Balpha}=\bigtimes_{j=1}^d\bold{X}_{\alpha_j}= \bigtimes_{j=1}^d \{x_{\alpha_j,1},x_{\alpha_j,2},\ldots\}$. Let $\Bx\in\mathcal{X}$ and let $K_i$ be the hyper-cube in $\BX_{\Balpha}$ containing $\Bx$ with vertices $\{\Bx_{{\Balpha},(i_1,\cdots,i_{d})}:i_j=\beta_j,\beta_j+1,j=1,\cdots,d\}$ as before. Suppose the inductive hypothesis:
\begin{align*}
    \hat{f}^{\text{BR}}_n(\Bx)&=k(\Bx,\BX_{\Balpha})\left[k(\BX_{\Balpha},\BX_{\Balpha})\right]^{-1}f(\BX_{\Balpha})\\
    &=\text{vec}\big[\Phi^{[d]}(\Bx)\big]^\intercal f(\BX_{\Balpha})\\
    &=\sum_{i_1=\beta_1}^{\beta_1+1}\sum_{i_2=\beta_2}^{\beta_2+1}\cdots\sum_{i_d=\beta_d}^{\beta_d+1}\bold{\Phi}^{[d]}_{i_1,\cdots,i_d}(\Bx)f(x_{\alpha_1,i_1},x_{\alpha_2,i_2},\cdots,x_{\alpha_d,i_d}),
\end{align*}
where:
\begin{equation*}
    \bold{\Phi}^{[d]}_{i_1,\cdots,i_d}(\Bx)=\prod_{j=1}^d\phi_{\alpha_j,i_j}(x_j).
\end{equation*}
 From this hypothesis, we can see that  there are at most $2^{d}$ non-zeros entries on $\bold{\Phi}^{[d]}(\Bx)$, namely, the entries $\bold{\Phi}^{[d]}_{i_1,\cdots,i_{d}}(\Bx)$ with $i_j=\beta_j$ or $\beta_j+1$. Since $\hat{f}^{\text{BR}}_n$ is the Lagrange polynomial interpolation of $f$ and is continuous, this assumption is equivalent to $\hat{f}^{\text{BR}}_n=\mathcal{I}_{\Balpha} f$.

Under this inductive hypothesis, consider the case for dimension $d+1$. Here, the full grid design becomes $\bold{X}_{\Balpha}=\BX_{{\Balpha}_{1:d}}\bigtimes\BX_{\alpha_{d+1}}$. Now let $\Bx\in \mathcal{X}$ and let $K_i$ be a hyper-cube  in $\BX_{\Balpha}$ containing $\Bx$ with vertices $\{\Bx_{{\Balpha},(i_1,\cdots,i_{d})}:i_j=\beta_j,\beta_j+1,j=1,\cdots,d+1\}$. Let $k(\Bx,\bold{y})=k(\Bx_{1:d},\bold{y}_{1:d})k_{d+1}(x_{d+1},y_{d+1})$. Then the vector $k^{-1}(\BX_{\Balpha},\BX_{\Balpha})k(\BX_{\Balpha},\Bx)$ becomes:
\begin{align*}
    &\ \ \ \ \ k^{-1}(\bold{X}_{\Balpha},\bold{X}_{\Balpha})k(\bold{X}_{\Balpha},\bold{x})\\
    &=\big\{k^{-1}(\bold{X}_{{\Balpha}_{1:d}},\bold{X}_{{\Balpha}_{1:d}})\bigotimes k_{d+1}^{-1}(\bold{X}_{\alpha_{d+1}},\bold{X}_{\alpha_{d+1}})\big\}\\
    & \quad \quad \quad \quad \text{vec}\big([k(\Bx_{1:d},\Bx_{{\Balpha}_{1:d},\Bbeta_{1:d}})k_{d+1}(x_{d+1},x_{\alpha_{d+1},\beta_{d+1}})]_{\Bbeta_{1:d},\beta_{d+1}}\big)\\
    &=\text{vec}\Big\{k_{d+1}^{-1}(\bold{X}_{\alpha_{d+1}},\bold{X}_{\alpha_{d+1}})\\
    & \quad \quad \quad [k(\Bx_{1:d},\Bx_{{\Balpha}_{1:d},\Bbeta_{1:d}})k_{d+1}(x_{d+1},x_{\alpha_{d+1},\beta_{d+1}})]_{\Bbeta_{1:d},\beta_{d+1}}k^{-1}(\bold{X}_{{\Balpha}_{1:d}},\bold{X}_{{\Balpha}_{1:d}})\Big\}.
\end{align*}
Similarly, define:
\begin{align*}
\bold \Phi(\Bx)^{[d+1]}& = k_{d+1}^{-1}(\bold{X}_{\alpha_{d+1}},\bold{X}_{\alpha_{d+1}})\\
& \quad \quad   [k(\Bx_{1:d},\Bx_{{\Balpha}_{1:d},\Bbeta_{1:d}})k_{d+1}(x_{d+1},x_{\alpha_{d+1},\beta_{d+1}})]_{\Bbeta_{1:d},\beta_{d+1}}k^{-1}(\bold{X}_{{\Balpha}_{1:d}},\bold{X}_{{\Balpha}_{1:d}}).
\end{align*}

From the inductive hypothesis, we know that
$$k^{-1}(\BX_{{\Balpha}_{1:d}},\BX_{{\Balpha}_{1:d}})k(\BX_{{\Balpha}_{1:d}},\Bx_{1:d})=\text{vec}\left(\bold{\Phi}^{[d]}(\Bx_{1:d})\right)$$
which is the vectorization of the sparse matrix $\bold{\Phi}^{[d]}(\Bx_{1:d})$, which has at most $2^d$ non-zero entries. Hence, $\Phi^{[d+1]}(\Bx)$ can be decomposed as:
\begin{align*}
    & \bold{\Phi}^{[d+1]}(\Bx)=k_{d+1}^{-1}(\BX_{\alpha_{d+1}},\BX_{\alpha_{d+1}})\bigg[\text{vec}\big(\bold{\Phi}^{[d]}(\Bx_{1:d})\big)k_{d+1}(x_{d+1},x_{\alpha_{d+1},\beta_{d+1}}) \bigg]_{\beta_{d+1}}
\end{align*}
So there are at most $2^{d+1}$ non-zeros entries on $\bold{\Phi}^{[d+1]}(\Bx)$, namely, the entries $\bold{\Phi}_{i_1,\cdots,i_{d+1}}(\Bx)$ where $i_j=\beta_j$ or $\beta_j+1$. Incorporating this, we then have:
\begin{align*}
    \hat{f}^{\text{BR}}_n(x)&=\sum_{i_1=\beta_1}^{\beta_1+1}\sum_{i_2=\beta_2}^{\beta_2+1}\cdots\sum_{i_{d+1}=\beta_{d+1}}^{\beta_{d+1}+1}\bold{\Phi}_{i_1,\cdots,i_d}(\Bx)f(x_{\alpha_1,i_1},x_{\alpha_2,i_2},\cdots,x_{\alpha_{d+1},i_{d+1}})\\
    &=\sum_{\Bx_{{\Balpha},\beta}\in\BX}\phi_{{\Balpha},\Bbeta}(\Bx)f(\Bx_{{\Balpha},{\Bbeta}})=\mathcal{I}_{\Balpha} f(\Bx),
\end{align*}
which completes the inductive step.
\end{proof}

% However, we notice that if  equation (\ref{eq:LagrangeKrigingEquivalent}) holds, then the vector $k^{-1}(\bold{X},\bold{X})f(\bold{X})$ must be the hierarchical surplus of $\mathcal{I}_\alpha$. For detailed introduction of hierarchical surplus, please refer to \cite{Garcke12}, \cite{Bungartz04} or see section 4. The key step of the proof is to write down the matrix  $k^{-1}(\bold{X},\bold{X})$ explicitly. We can extend the work in \cite{DingZhang18} to show that equation (\ref{eq:LagrangeKrigingEquivalent}) does hold. In general, theorem \ref{thm:LagrangeKrigingEquivalent} is inspired by the fact that the hierarchical surplus (see equation (11) of \cite{Garcke12}) is equal to tensor product of matrices in example 2 of \cite{DingZhang18}  up to a scale. Detailed proof is left in appendix.

\subsubsection{Sparse Grids}
One disadvantage of full grid designs is the so-called \textit{curse-of-dimensionality}: both the design size and its corresponding prediction error grow exponentially in dimension $d$. To this end, we extend next the earlier equivalence between FEM and the Brownian kernel for a broader class of designs called \textit{sparse grids} \citep{Bungartz04}, which ``sparsify'' a full grid by retaining only certain subgrids of interest. These designs are used later to prove the improved convergence rates for BdryGP.

We first provide a brief review of sparse grid designs. A sparse grid of level $k$, denoted as $\BX^{\rm SP}_k$, is defined as follows:
\begin{equation}
    \BX_k^{\rm SP}=\bigcup_{k\leq|{\Balpha}|\leq k+d-1}\BX_{\Balpha}, \quad |\Balpha| := \sum_{j=1}^d \alpha_j.
    \label{eq:sparsegrid}
\end{equation}
\noindent In words, the sparse grid $\BX_k^{\rm SP}$ is the union of full grids $\BX_{\Balpha}$ whose multi-indices $\Balpha$ sums between $k$ and $k+d-1$. Figure \ref{fig:sgg} shows sparse grids of levels 1 to 4 in two dimensions; we see that sparse grids provide a sizable reduction in design size compared to full grids. This reduction plays a key role in providing relief from dimensionality in many numerical approximation problems \citep{Wendland10,Dick13}.
\begin{figure}
\centering
\includegraphics[width=0.24\textwidth]{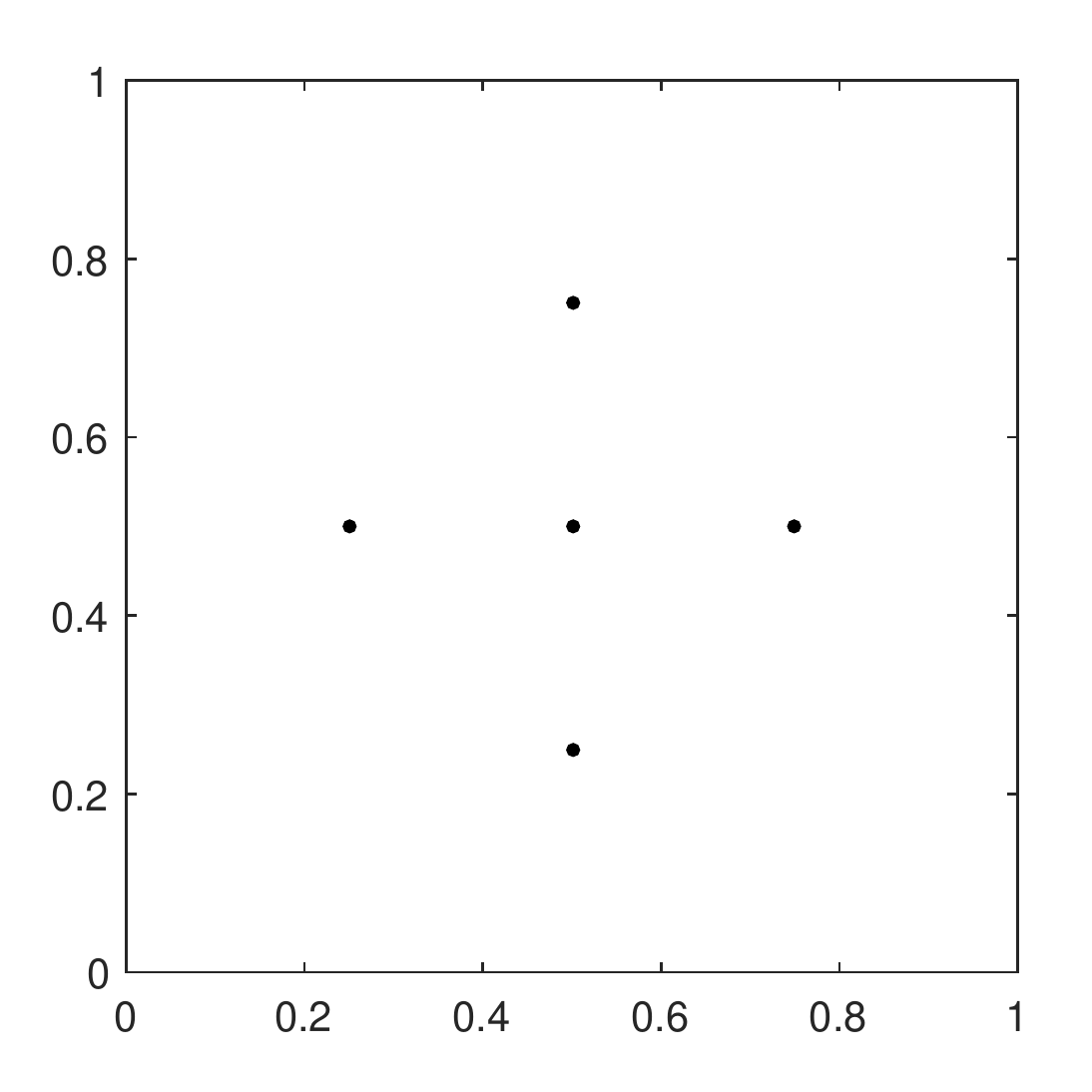}
\includegraphics[width=0.24\textwidth]{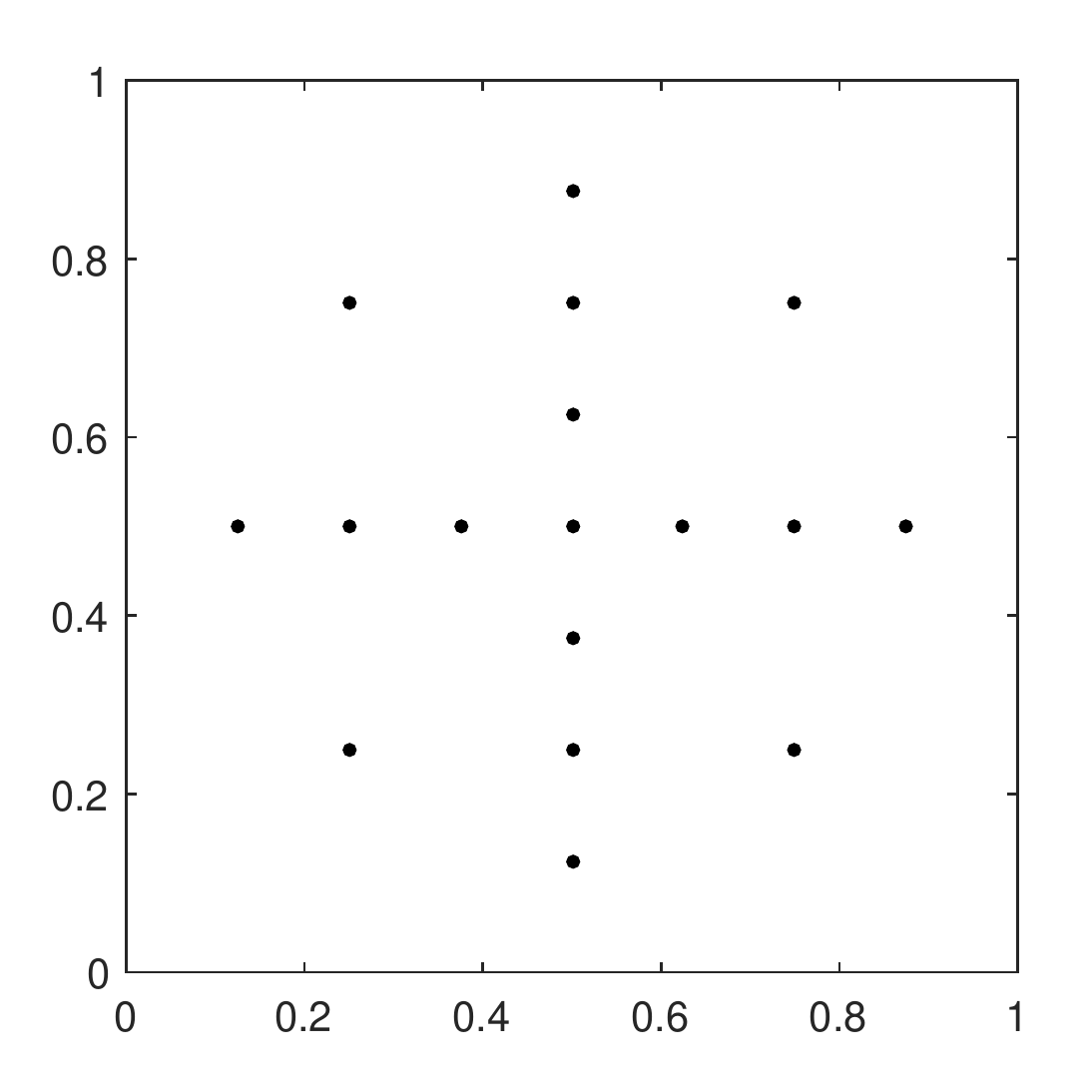}
\includegraphics[width=0.24\textwidth]{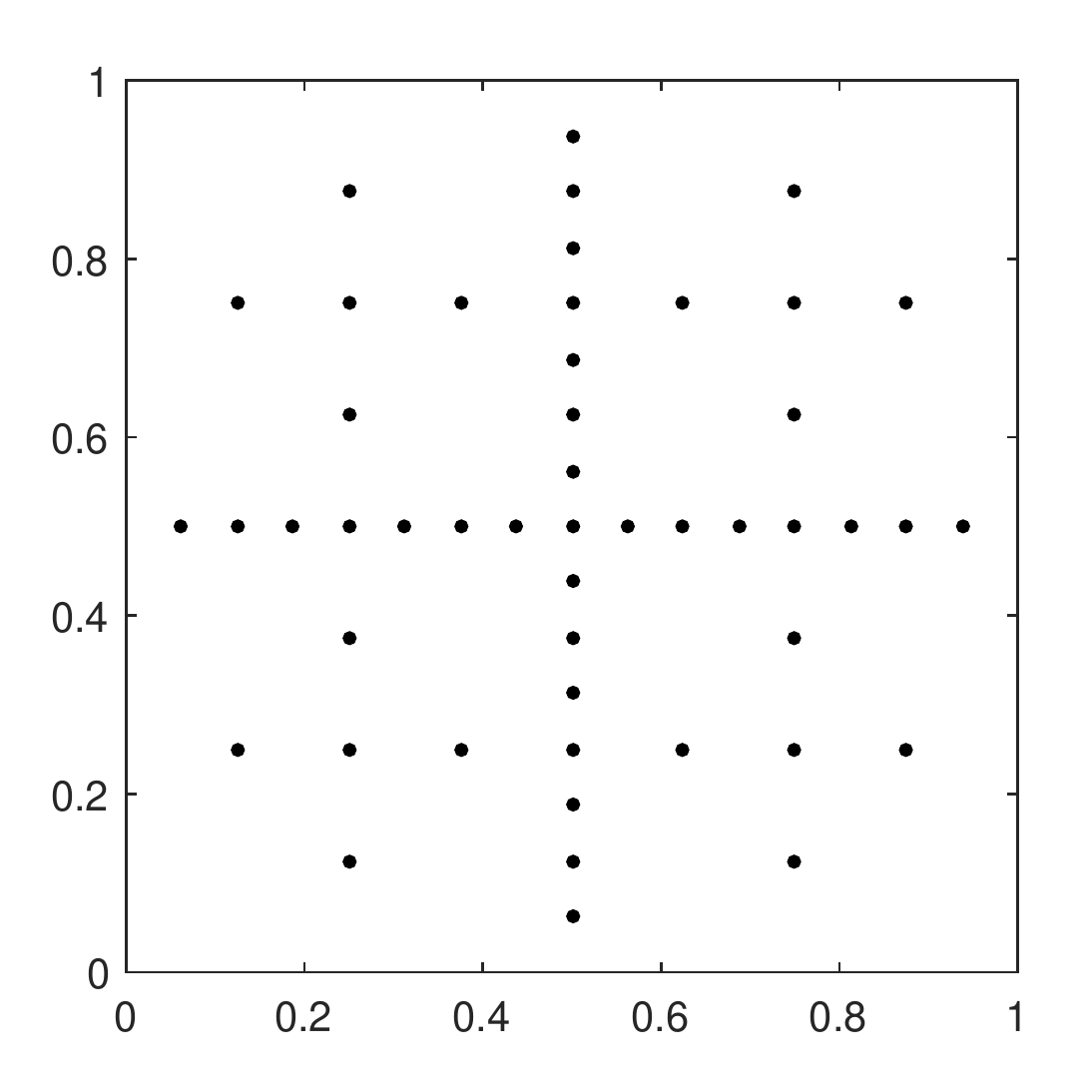}
\includegraphics[width=0.24\textwidth]{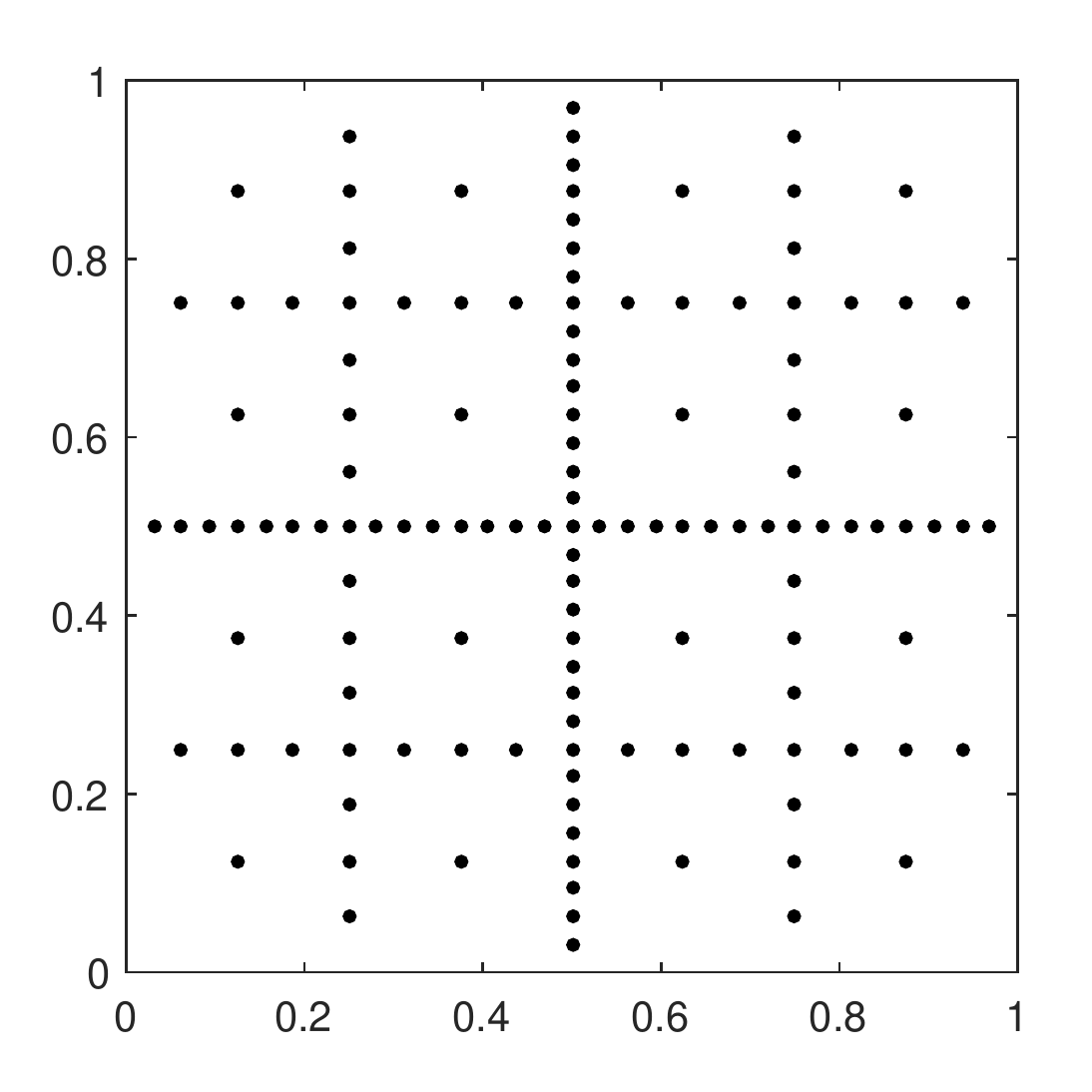}

\caption{Sparse grid designs of levels 1 to 4 in two dimensions.}
\label{fig:sgg}
\end{figure}

The FEM solution $\mathcal{I}_{\Balpha}f$ in \eqref{eq:LagrangePoly}, previously defined for the full grid $\BX_{\Balpha}$, can be extended analogously for sparse grids. Similar to before, let $\mathcal{V}_k^{\rm SP}$ be the sum of the finite-dimensional function spaces for each of the component full grids in the sparse grid \eqref{eq:sparsegrid}. The FEM solution on sparse grid $\BX_k^{\rm SP}$, defined as the projection of the weak solution $f$ on $\mathcal{V}_k^{\rm SP}$, can be shown (Equation (28) of \citealp{Garcke12}) to have the form:
\begin{equation}\label{eq:combinationtechnique}
    \begin{aligned}
    \mathcal{I}_k^{\rm SP} f &=\sum_{j=0}^{d-1}(-1)^j {d-1 \choose j}\sum_{|{\Balpha}|=k+d-1-j}\mathcal{I}_{\Balpha}f.
    \end{aligned}
\end{equation}

With this in hand, we show that under sparse grid designs, the FEM solution $\mathcal{I}_k^{\rm SP} f$ is also equivalent to the GP posterior mean $\hat{f}_n^{\rm BR}$ with the Brownian kernel $k^{\rm BR}$:
\begin{theorem}
\label{thm:sggFEMkriging}
Suppose $I^{[0]} \cup I^{[1]} = [d]$, and assume the sparse grid design $\BX_k^{\rm SP}$ with $n=|\BX_{\Balpha}|$ points. For any $f\in \CalH^{1,c}_{mix}$, the posterior predictor $\hat{f}_n^{\rm BR}$ of a GP with mean function \eqref{eq:ourmu} and Brownian kernel $k^{\rm BR}$ is equivalent to the FEM solution $\mathcal{I}_k^{\rm SP} f$ in \eqref{eq:combinationtechnique}.
\end{theorem}
\begin{proof}
In Theorem \ref{thm:LagrangeKrigingEquivalent}, we have shown the equivalence between the FEM solution $\mathcal{I}_{\Balpha}f$ and the GP predictor $\hat{f}_n^{\rm BR}$ on the full grid design $\BX_{\Balpha}$. Hence, $\mathcal{I}_{\Balpha}f$ can be replaced by $\hat{f}_n^{\rm BR}$ in Equation (\ref{eq:combinationtechnique}). The result then follows by Algorithm 1 of \cite{Plumlee14}.
\end{proof}
% Therefore, the BLUE of BdryGP with the Brownian kernel conditioned on $\BX_k^{\rm SP}$ can also be written in the form of equation (\ref{eq:combinationtechnique}).

\subsection{FEM and the BdryMat\'ern kernel}
Having proved the connection between FEM and the Brownian kernel $k^{\rm BR}$, we then show how this relates to the BdryMat\'ern kernel $k^{\rm BM}_\omega$ used in BdryGP. Of course, the GP predictors $\hat{f}_n^{\rm BR}$ and $\hat{f}_n^{\rm BM}$ under the Brownian and BdryMat\'ern kernels are not equivalent. However, we show below that the BdryGP approximation error $|f - \hat{f}_n^{\rm BM}|$ can be upper bounded by the approximation error $|f - \hat{f}_n^{\rm BR}|$ from the Brownian kernel:

% However, if we replace $k^{\text{BR}}$ by $k^{\text{BM}}_\omega$ in the previous theorem, the equality does not hold obviously. Nevertheless, we can show that the difference between $\mathcal{V}_{k}^s$ and the function space:
% $$\{k^{\text{BM}}_\omega(\Bx_{\Balpha,\Bbeta},\cdot):\Bx_{\Balpha,\Bbeta}\in\BX^{\rm SP}_k\}$$
% converges to $0$ faster than the the convergence rates we will show in next section. As a result, all BdryGPs conditioned on the same sparse grid design share the same convergence rate.
\begin{theorem}
\label{thm:diffBdryGP}
Suppose $I^{[0]} \cup I^{[1]} = [d]$, and assume the sparse grid design $\BX_k^{\rm SP}$ with $n=|\BX_k^{\rm SP}|$ points. Let $\hat{f}_n^{\rm BM}$ be the BdryGP predictor with mean function \eqref{eq:ourmu} and BdryMat\'ern kernel $k^{\rm BM}_{\omega}$. For any $f\in \CalH^{1,c}_{mix}$ and any $\omega > 0$:
%Assume a sparse grid design $\BX_k^{\rm SP}$ with $n$ design points. Let $\hat{f}_n^{\rm BM}$ and $\hat{f}_n^{\rm BR}$ be the BLUEs of the BdryGP with the BdryMat\'ern kernel \eqref{eq:BdryMatern} and the Brownian kernel \eqref{eq:brownian_kernel}, respectively. For any wavelength $\omega$, we have the following uniform asymptotic equality \cmtS{$f \in C(\mathcal{X})$? absolute error? may be more clear to explicitly write out $\mathcal{O}$ (since constants indep of both $f$ and $n$)}:
  \begin{equation}
  \begin{aligned}
      &||f-\hat{f}_n^{\rm BR}||_{L^\infty} \leq C||f-\hat{f}_n^{\rm BM}||_{L^{\infty}}. 
      \end{aligned}
      \label{eq:diffBdryGP}
  \end{equation}
  for some constant $C$ independent of $f$, $\omega$ and $n$.
\end{theorem}
\noindent The proof of this theorem can be found in Appendix \ref{apdix:pf4diffBdryGP}.

\section{Function spaces: BdryGP and FEM}
\label{sec:funcspace}
Next, we prove the equivalence between the RKHS of the Brownian kernel $k^{\rm BR}$, the constrained Sobolev space with mixed first derivatives, and the weak solution space for FEM.% \crd{[Move within proof] To keep notation simple, we assume in this section that $I^{[0]} = [d]$, i.e., all \textit{left} boundaries are known. These results hold analogously for the more general assumption $I^{[0]} \cup I^{[1]} = [d]$, i.e., at least \textit{one} boundary is known for each variable.}

% and introduce the sparse grid design from a function space perspective.

%In this section, we focus on the analysis of the native space induced by  BB kernel. Results in previous section are not surprising in that Lagrange polynomial interpolation really is related to kriging  as long as we can build the relation between  the Sobolev space with dominating mixed derivative $\mathcal{H}^{1}_{mix}(T^d)$ and the native space induced by BB kernel.

\subsection{Brownian Kernel RKHS and the Constrained Mixed Sobolev Space}

We first establish the equivalence between the Brownian kernel RKHS and the mixed Sobolev space under boundary constraints. Let $\CalH^1_{mix}$ be the Sobolev space of functions with mixed first derivative:
\begin{equation}
\CalH^1_{mix} := \left\{f: D^{\Balpha} f\in L^2(\Real^d), |\Balpha|_\infty\leq 1\right\}.
\label{eq:mixsob}
\end{equation}
where $|\Balpha|_\infty=\max_{j\in[d]}|\alpha_j|$. Further let $\CalH^{1,c}_{mix}$ be the space of functions in $\CalH^1_{mix}$ with boundary value zero on known boundaries:
\begin{equation}
\CalH^{1,c}_{mix}:=\left\{f\in\CalH^1_{mix}:f(\Bx)=0 \text{ if } x_i \leq 0, i \in I^{[0]} \text{ or if } x_i \geq 1, i\in I^{[1]} \right\}.
\label{eq:mixsobbound}
\end{equation}
\noindent We will call $\CalH^{1,c}_{mix}$ the \textit{constrained} mixed Sobolev space.
% \begin{equation}
% \CalH^1_{I^{[1]},mix}:=\left\{f\in\CalH^1_{mix}:f(\Bx)=0, \ \text{if\ } x_i\leq 0,\forall i\in[d]\  \text{or\ }  x_j\geq 1, \forall j\in I^{[1]} \right\}.
% \label{eq:mixsobbound}
% \end{equation}
 
The following proposition shows that the Brownian kernel RKHS $\mathcal{H}_{k^{\rm BR}}$ and the constrained Sobolev space $\CalH^{1,c}_{mix}$ are equivalent function spaces:
\begin{prop}\label{prop:NativeSpaceBB}
The function spaces $\mathcal{H}_{k^{\rm BR}}$ and $\mathcal{H}^{1,c}_{mix}$ are equivalent.
\end{prop}
\begin{proof}
By a straight-forward extension of Proposition \ref{prop:BdryMatern1d}, the Brownian kernel $k^{\rm BR}$ satisfies the following equation for any $f\in \mathcal{H}^{1,c}_{mix}$:
\begin{equation}
\label{eq:Brownian_RKHS}
    f(\Bx)=\int_{\Real^d}D^{\bold{1}}f(\bold{s})D^{\bold{1}}k^{\text{BR}}(\Bx,\bold{s})d\bold{s}
\end{equation}
From equation (\ref{eq:Brownian_RKHS}), the inner product of $\mathcal{H}_{k^{\rm BR}}$ is:
$$\langle f,g \rangle_{\mathcal{H}_{k^{\rm BR}}}=\int_{\Real^d}D^{\Balpha}fD^{\Balpha}gd\Bx, \quad f, g \in \mathcal{H}^{1,c}_{mix}.$$
Thus, we only need to show the norm equivalence identity:
\begin{equation*}
    C_1||f||^2_{\CalH_{k^{\text{BR}}}}\leq ||f||^2_{\mathcal{H}^{1,c}_{mix}}\leq C_2||f||^2_{\CalH_{k^{\text{BR}}}}.
\end{equation*}
Obviously, $C_1=1$. By the 1-d Poincar\'e inequality for locally absolutely continuous functions, there exists some constant $C$ such that:
\begin{equation*}
    \int_{\Real^d}[D^{\Balpha} f]^2\leq C \int_{\Real^d}[D^{\bold{1}}f]^2, \quad \text{for any } |\Balpha|_\infty\leq 1 \text{ and any } f\in \mathcal{H}^{1,c}_{mix}.
\end{equation*}
Iteratively applying the Poincaré inequality again, we get:
\begin{equation*}
    ||f||^2_{\mathcal{H}^{1,c}_{mix}}\leq 2^dC||f||^2_{\CalH_{k^{\text{BR}}}}
\end{equation*}
which proves the norm equivalence identity.
\end{proof}

\subsection{Hierarchical Difference Spaces}
% The two equivalent finite-dimensional spaces in theorem \ref{thm:LagrangeKrigingEquivalent} can be decomposed into subspaces of different levels. We call these subspaces hierarchical difference spaces.
Next, we introduce the idea of a hierarchical difference space, which is widely used in FEM analysis. These spaces will allow for a multi-level decomposition of the finite-dimensional function spaces for FEM, and thereby the FEM solution as well.

% In this subsection, we first introduce the concept of hierarchical difference space. To this end, we first briefly define some basic concepts related to $\mathcal{H}^{1,c}_{mix}$. WLOG, we assume only left boundary condition is given on each dimension which implies that $I^{[1]}=\emptyset$. The case that $I^{[1]}\neq\emptyset$ will be discussed at the end of this subsection. \cmtS{repetitive?} Let $\BX_{\Balpha}$ be a grid on $\CalX$ with mesh size $h_{\Balpha}=(h_{\alpha_1},\cdots,h_{\alpha_d})=(2^{-\alpha_1},\cdots,2^{-\alpha_d})$ as before. We first define a function space $\mathcal{V}_{{\Balpha}}$ which consists of d-linear functions:

    % \mathcal{V}_{{\Balpha}}:&=\text{span}\{\phi_{{\Balpha},{\Bbeta}}:\beta_i=1,\cdots,2^{\alpha_i}-1\ \text{if}\ i\in I^{[1]}\ \text{and}\ \beta_i=1,\cdots,2^{\alpha_i}\ \text{if}\ i\not\in I^{[1]} \}\\
    % &=\text{span}\{\phi_{{\Balpha},{\Bbeta}}:\ \beta_i=1,\cdots,2^{\alpha_i} \} 

% he last line is from our assumption and $\phi_{{\Balpha},{\Bbeta}}$

Let us define the finite-dimensional function space $\mathcal{V}_{{\Balpha}}$ for the FEM solution on full grid $\BX_{\Balpha}$:
\begin{equation}
\label{eq:va}
    \mathcal{V}_{{\Balpha}}:=\text{span}\{\phi_{{\Balpha},{\Bbeta}}: \Bx_{\Balpha,\Bbeta}\in \BX_{\Balpha} \}
\end{equation}
where $\phi_{{\Balpha},{\Bbeta}}$ is the earlier hat function with $\phi_{0,1}(x)=x$ and $\phi_{0,2}(x)=1-x$. It is clear that $\mathcal{V}_{\Balpha}$ is the tensor product of these 1-d spaces, i.e., $\mathcal{V}_{{\Balpha}}=\bigotimes_{j=1}^d\mathcal{V}_{\alpha_j}$.

% The following proposition tells us that $\mathcal{V}_{{\Balpha}}$ has a so-called hierarchical decomposition 
% \begin{prop}
% \label{prop:hierdiff}
% \begin{equation}
% \label{eq:subspacedecomp}
%     % \mathcal{V}_\bold{n}:=\mathcal{V}_{(n,\cdots,n)}=\bigoplus_{0\leq|{\Balpha}|_{\infty}\leq n} W_{{\Balpha}},
%     \mathcal{V}_{\Balpha}=\bigoplus_{ \bm{0} \leq \Balpha' \leq \Balpha} W_{\Balpha'},
% \end{equation}
% where $W_{{\Balpha}}=\text{span}\{\phi_{{\Balpha},{\Bbeta}}:{\Bbeta}\in B_{{\Balpha}}\}$ is called a \textit{hierarchical difference space}.
% \end{prop}

Furthermore, $\mathcal{V}_{\Balpha}$ can be represented as the following multi-level subspace decomposition:
\begin{equation}
\label{eq:subspacedecomp}
 % \mathcal{V}_\bold{n}:=\mathcal{V}_{(n,\cdots,n)}=\bigoplus_{0\leq|{\Balpha}|_{\infty}\leq n} W_{{\Balpha}},
 \mathcal{V}_{\Balpha}=\bigoplus_{ \bm{0} \leq \Balpha' \leq \Balpha} W_{\Balpha'},
\end{equation}
where $W_{{\Balpha}}=\text{span}\{\phi_{{\Balpha},{\Bbeta}}:{\Bbeta}\in B_{{\Balpha}}\}$ is called a \textit{hierarchical difference space}.
Further details on these spaces can be found in \cite{Yse86} and \cite{Bungartz04}. We note that, in order to incorporate partial boundaries, the hierarchical difference space used here is slightly modified from that in the literature. In the case of full boundaries (i.e., $I^{[0]} = I^{[1]} = [d]$), the two spaces are equivalent.

The subspace decomposition \eqref{eq:subspacedecomp} allows for the following useful multi-level decomposition of the FEM solution. Consider first the FEM solution $\mathcal{I}_{\Balpha}f$ on the full grid $\BX_{\Balpha}$. From equation \eqref{eq:subspacedecomp}, $\mathcal{I}_{\Balpha}f$ can be decomposed as:
 \begin{equation}\label{eq:hieraDiffSpa_proj}
     \mathcal{I}_{\Balpha}f = \sum_{ \bm{0} \leq \Balpha' \leq \Balpha}f_{\Balpha'}(\Bx).
 \end{equation}
 Here, $f_{{\Balpha}}$ is the projection of $f$ on $W_{\Balpha}$, given by:
  \begin{equation}\label{eq:Proj_W_a}
     f_{{\Balpha}}(\Bx)=\sum_{{\Bbeta}\in B_{\Balpha}}c_{{\Balpha},{\Bbeta}}\phi_{{\Balpha},{\Bbeta}}(\Bx),
 \end{equation}
 and the constant $c_{{\Balpha},{\Bbeta}}$ is known as the \textit{hierarchical surplus}, defined as:
 \begin{equation}\label{eq:HierarchicalSurplus}
     c_{{\Balpha},{\Bbeta}}=\bigg(\prod_{j=1}^dA_{\alpha_j,\beta_j}\bigg)f(\bold{X}_{\Balpha}), \quad      A_{\alpha_j,\beta_j}=\begin{cases}
     [-\frac{1}{2}\ \  1 \ \ -\frac{1}{2}]& \text{if}\ \alpha_j\geq 1\\
     [  -1\ \  1] & \text{if}\ \alpha_j=0
     \end{cases}.
 \end{equation}
% which gives the coefficients for a linear combination of nodal values
% of its argument $f$
Here, $\prod_{j=1}^dA_{\alpha_j,\beta_j}$ denotes the Kronecker product of vectors $A_{\alpha_j,\beta_j}$; this is the standard
stencil notation used in numerical analysis. Similarly, the sparse grid FEM solution $\mathcal{I}_k^{\rm SP}$ can be decomposed as:
\begin{equation}
\label{eq:hieraDiffSpa_proj2}
\mathcal{I}_k^{\rm SP} f = \sum_{0 \leq |\Balpha| \leq k+d-1} f_{\Balpha}(\Bx).
\end{equation}
\noindent This decomposition, along with the equivalences in Section \ref{sec:fem}, provides the basis for proving improved convergence rates for BdryGP.

% \crd{(Simon: need?) The definition of hierarchical subspace is similar to the one used in numerical computation except that only left boundary condition is given in our case. When $I^{[1]}=[d]$, then two hierarchical subspaces  are equivalent. The only difference between the two subspaces is that we have some extra subspaces $W_{\Balpha}$ when $\exists \alpha_i=0$, that is also the reason why we can assume $I^{[1]}=\emptyset$.} 
% For more concrete introduction of hierarchical difference spaces, please refer to \cite{Yse86},\cite{Bungartz04} or \cite{Garcke12}.
 
 \subsection{Brownian Kernel RKHS and Hierarchical Difference Spaces}
%  We can see in theorem \ref{thm:LagrangeKrigingEquivalent} that, conditioned on observations from a full grid design  $\BX_{\bold{n}}$ with mesh size $h_{\bold{n}}$ where $\bold{n}:=(n, \cdots,n)$, the BLUE of a zero-mean GP of Brownian kernel can be written as:
%  \begin{equation*}
%      \hat{f}(\Bx)=\sum_{\Bx_{\bold{n},{\Bbeta}}\in \BX_{\bold{n}}}f(\Bx_{\bold{n},{\Bbeta}})\phi_{\bold{n},{\Bbeta}}(\Bx)
%  \end{equation*}
%  which is very similar to the projection $\pi_\nu[f]$ defined in equation (\ref{eq:hieraDiffSpa_proj}).  In this subsection, we show that they are equal:
%  \begin{equation}
%      \pi_{\mathcal{V}_\bold{n}}[f]=\hat{f}
%  \end{equation}
% and then we can show the equivalence of three function spaces.
 
Consider now the limiting function space $\mathcal{V}$:
\begin{equation}\label{eq:hierarchicalInterpolation2Kriging}
     \mathcal{V}=\bigotimes_{j=1}^d\overline{\lim_{\alpha_j\to\infty}\mathcal{V}_{\alpha_j}}.
\end{equation}
\noindent In other words, $\mathcal{V}$ is the tensor product of the limiting 1-d finite-dimensional spaces in equation \eqref{eq:va}. The space $\mathcal{V}$ can be viewed as the weak solution space on which FEM aims to solve the PDE system \eqref{eq:pde} in the limit.

% The reason that we are interested in $\mathcal{V}$ is because we can decompose and function $f\in\mathcal{V}$ by its projection on $W_{\Balpha'}$ according to equation (\ref{eq:subspacedecomp}). This is similar to principle analysis in that the projection of $f$ on $W_{\Balpha'}$ contains less information of $f$ if $\Balpha'$ is very large.  \cmtS{I will add a sentence or two describing why we need $\mathcal{V}$.}

The following proposition shows the equivalence of $\mathcal{V}$  to the native space of $k^{\text{BR}}$:

%Let $k$ denote the BB kernel. The goal of this subsection  is to show  the following spaces  are equivalent:
%\begin{equation}\label{eq:SpacesEquivalent}
%    \mathcal{H}_k\sim\mathcal{H}^1_{0,mix}\sim\overline{\mathcal{V}}.
%\end{equation}
%We now generalize equation ($\ref{eq:hierarchicalInterpolation2Kriging }$) to the follow theorem.
\begin{prop}\label{thm:hierarchicalInterpolation2Kriging}
% \crd{(Simon: new notation?)Let $k$ be the Brownian kernel defined in equation (\ref{eq:brownian_kernel})
% Let $\hat{f}_\bold{n}$ be the posterior mean defined in equation (\ref{eq:LagrangeKrigingEquivalent}).} Let $\pi_{\mathcal{V}_\bold{n}}f$ be the projection of $f$ on the space $\mathcal{V}_{\bold{n}}$. Then
% \begin{equation*}
%     \hat{f}_\bold{n}=\pi_{\mathcal{V}_{\bold{n}}}[f].
% \end{equation*}
% \cmtS{Mention later in proof to Theorem \ref{thm:4spacesequal}?} As a result, the following spaces are equivalent for any $\bold{n}$:
% $$\mathcal{V}_{\bold{n}}=\{k^{\text{BR}}(\Bx_{\bold{n},\Bbeta},\cdot):\Bx_{\bold{n},\Bbeta}\in\BX_{\bold{n}}\}$$
% The function spaces $\mathcal{V}_{\Balpha}$ is equivalent  to the following finite dimensional finite dimensional function space:
% $$\text{span}\{k^{\text{BR}}(\Bx_{\Balpha,\Bbeta},\cdot):\Bx_{\Balpha,\Bbeta}\in \BX_{\Balpha}\}$$
% for any $\Balpha\in\mathbb{N}^d$.
The function spaces $\mathcal{V}$ and $\mathcal{H}_{k^{\rm BR}}$ are equivalent.
\label{thm:hierdiff}
\end{prop}

\noindent The proof of this proposition requires the following lemma, which shows that the finite-dimensional RKHS of $k^{\rm BR}$ on grid $\BX_{\Balpha}$ is equivalent to $\mathcal{V}_{\Balpha}$.

\begin{lem}
The finite-dimensional spaces $\mathcal{V}_{\Balpha}$ and $\{k^{\rm BR}(\Bx_{\Balpha,\Bbeta},\cdot):\Bx_{\Balpha,\Bbeta}\in\BX_{\Balpha}\}$ are equivalent for any ${\Balpha}\in\mathbb{N}^d$.
\label{lem:finitespace}
\end{lem}

\noindent The proof of Lemma \ref{lem:finitespace} is given in Appendix \ref{apdix:pf4hierdiff}.

\begin{proof}[Proof (Proposition \ref{thm:hierdiff})]
From Lemma \ref{lem:finitespace}, we know that the projectors to $\mathcal{V}_{\Balpha}$ and $\{k^{\text{BR}}(\Bx_{\Balpha,\Bbeta},\cdot):\Bx_{\Balpha,\Bbeta}\in\BX_{\Balpha}\}$ are equal for any ${\Balpha}\in\mathbb{N}^d$. Since $\mathcal{H}_{k^{\rm BR}}$ is the completion of the space $\lim_{\{\alpha_j\to\infty\}_{j=1}^d}\{k^{\text{BR}}(\Bx_{\Balpha,\Bbeta},\cdot):\Bx_{\Balpha,\Bbeta}\in\BX_{\Balpha}\}$, it is therefore the function space defined in equation \eqref{eq:hierarchicalInterpolation2Kriging}.
\end{proof}

% Theorem \ref{thm:LagrangeKrigingEquivalent} and theorem \ref{thm:hierarchicalInterpolation2Kriging} have shown that
% $$\mathcal{V}_{\Balpha}=\{\phi_{\Balpha,\Bbeta}:\Bx_{\Balpha,\Bbeta}\in\BX_{\Balpha}\}=\{k^{\text{BR}}(\Bx_{\Balpha,\Bbeta},\cdot):\Bx_{\Balpha,\Bbeta}\in\BX_{\Balpha}\}$$
% for any $\Balpha \in \NatInt^d$ because the projectors of these spaces are equal. When $\bold{n}\to \infty$, we have the following theorem, which tells that the RKHS's of $k^{\text{BR}}$ and $k^{\text{BM}}_\omega$, $\CalH^{1,c}_{mix}$ and $\mathcal{V}$ are in fact equivalent function spaces.

Combining Propositions \ref{prop:NativeSpaceBB} and \ref{thm:hierdiff}, we can then prove the desired equivalence between the two RKHSs $\CalH_{k^{\text{BR}}}$ and $\CalH_{k^{\text{BM}}_\omega}$, the constrained mixed Sobolev space $\CalH^{1,c}_{mix}$, and the weak solution space $\mathcal{V}$:

\begin{theorem}
$\CalH_{k^{\text{BR}}}$, $\CalH_{k^{\text{BM}}_\omega}$, $\CalH^{1,c}_{mix}$ and $\mathcal{V}$ are equivalent function spaces.
\label{thm:4spacesequal}
\end{theorem}
\begin{proof}
Proposition \ref{prop:NativeSpaceBB} shows the equivalence between the RKHS $\CalH_{k^{\text{BR}}}$ and the constrained mixed Sobolev space $\CalH^{1,c}_{mix}$. Following the same reasoning (i.e., via the norm equivalence identity), the equivalence between the two RKHSs $\CalH_{k^{\text{BR}}}$ and $\CalH_{k^{\text{BM}}_\omega}$ can also be shown for any $\omega\in(0,\infty)$. The equivalence between $\mathcal{H}_{k^{\rm BR}}$ and $\mathcal{V}$ (Proposition \ref{thm:hierdiff}) then completes the proof.

% Because $\CalH^{1,c}_{mix}=\CalH_{k^{\text{BM}}_\bold{1}}$, we only need to show the equivalence between $\mathcal{V}$ and $\CalH_{k^{\text{BR}}}$, we then can finish the proof.

% From theorem \ref{thm:hierarchicalInterpolation2Kriging}, we can easily derive that  $\mathcal{V}= \CalH_{k^{\text{BR}}}$ by letting $n_i\to\infty$ for $i=1,\cdots,d$.

% As a result, the following finite-dimensional function spaces are equivalent for any $\Balpha\in\NatInt^d$:
% \begin{equation}
% \textup{span}\{\phi_{\Balpha,\Bbeta}(\cdot):\Bx_{\Balpha,\Bbeta}\in\BX_{\Balpha}\}=\textup{span}\{k^{\textup{BR}}(\Bx_{\Balpha,\Bbeta},\cdot):\Bx_{\Balpha,\Bbeta}\in\BX_{\Balpha}\}
% \label{eq:finitespace}
% \end{equation}

\end{proof}

% From the above theorem, we can decompose the RKHS of  $\CalH_{k^{\text{BR}}}$ into hierarchical difference subspaces of different levels:
% $$ \CalH_{k^{\text{BR}}}=\bigoplus_{{\Balpha}\in\mathbb{Z}^d_{\geq0}} W_{{\Balpha}}.$$

This function space equivalence allows for the decomposition of the RKHS $\CalH_{k^{\text{BR}}}$ (and its corresponding interpolator) into hierarchical difference subspaces (and its corresponding projections) of different levels. This decomposition plays a key role in proving the following convergence rates.

\section{Convergence rates for BdryGP}
\label{sec:convrates}
With these equivalences in hand, we now prove the desired rates for BdryGP under sparse grids. All of these rates assume that $I^{[0]} \cup I^{[1]} = [d]$, i.e., at least one boundary is known for each of the $d$ variables. Of course, the same rates also hold in the full boundary setting of $I^{[0]} = I^{[1]} = [d]$, where \textit{all} boundaries of $f$ are known.

\subsection{$L^p$ and $L^\infty$ Convergence Rates}
\label{sec:lprate}
Suppose $f$ is a \textit{deterministic} function from the constrained mixed Sobolev space $\mathcal{H}^{1,c}_{mix}$. Under boundary information, the following theorem proves the $L^p$ and $L^\infty$ convergence rates for the proposed BdryGP (with BdryMat\'ern kernel $k^{\rm BM}_\omega$):

\begin{theorem}\label{thm:L1error}
Suppose $I^{[0]} \cup I^{[1]} = [d]$, and assume the sparse grid design $\BX_k^{\rm SP}$ with $n = |\BX_k^{\rm SP}|$ points. For any $f \in \mathcal{H}^{1,c}_{mix}$ and any wavelength $\omega$, the BdryGP has an $L^p$ convergence rate of:
\begin{equation}\label{eq:L1error}
    ||f-\hat{f}_n^{\rm BM}||_{L^p}=\smallO(n^{-1}), \quad 1 \leq p < \infty
\end{equation}
and an $L^\infty$ convergence rate of:
\begin{equation}\label{eq:Linferror}
||f-\hat{f}_n^{\rm BM}||_{L^\infty}=\mathcal{O}(n^{-1}[\log n]^{2(d-1)}).
\end{equation}
\end{theorem}

%  Now, we are well prepared to prove theorem \ref{thm:L1error}:
The proof of Theorem \ref{thm:L1error} requires the following three lemmas. The first lemma (from \citealp{Bungartz04}) provides a big-O approximation of the number of points in the sparse grid $\BX_k^{\rm SP}$: 
\begin{lem}\label{lem:num4SPG}[Lemma 3.6 in \cite{Bungartz04}]
Let $n = |\BX_k^{\rm SP}|$ be the number of points in a $d$-dimensional sparse grid of level $k$. Then:
\begin{equation*}
    n=\mathcal{O}(2^k[\log 2^k]^{d-1}).
\end{equation*}
\end{lem}
\noindent The second lemma upper bounds the hierarchical surplus in $c_{\Balpha,\Bbeta}$ \eqref{eq:HierarchicalSurplus}:
\begin{lem}\label{prop:BoundofSurplus}
Let $f\in\mathcal{H}^{1,c}_{mix}$. Then there exists constants $C>0$ and $\gamma\in(0,1]$ independent of $f$, $\Balpha$ and $\Bbeta$, such that: \begin{equation}
|c_{\Balpha,\Bbeta}|\leq C2^{-(\gamma|\alpha|_{\infty}+|\alpha|)}
\label{eqn:lemi}
\end{equation}
for almost all $(\Balpha,\Bbeta)$, where $\Balpha \in \NatInt^d$ and $\Bbeta\in B_{\Balpha}$. Moreover, for any $\Balpha \in \NatInt^d$:
\begin{equation}
\sup_{{\Bbeta}\in B_{\Balpha}}|c_{\Balpha,\Bbeta}|\leq C2^{-|\Balpha|}.
\label{eqn:lemii}
\end{equation}
\end{lem}
\noindent The last lemma provides a useful identity:
\begin{lem}
For any $x \in (0,1)$,
\[ \sum_{i=0}^\infty x^i{i+k+d-1 \choose d-1} = \sum_{j=0}^{d-1}{k+d-1 \choose j}\left(\frac{x}{1-x}\right)^{d-1-j}\frac{1}{1-x}.\]
\label{lem:id}
\end{lem}

The proofs of Lemma \ref{prop:BoundofSurplus} and \ref{lem:id} are found in the Appendix.

\begin{proof}
Consider first the prediction error $f - \hat{f}_n^{\rm BR}$ for some $f \in \mathcal{H}_{mix}^{1,c}$, where $\hat{f}_n^{\rm BR}$ is the GP predictor using the Brownian kernel $k^{\rm BR}$. Using (i) the function space equivalence $\mathcal{H}_{mix}^{1,c} = \mathcal{V}$ (Theorem \ref{thm:4spacesequal}) and (ii) the equivalence between $\hat{f}_n^{\rm BR}$ and the sparse grid FEM solution $\hat{f}_n^{\rm BR}$ (Theorem \ref{thm:sggFEMkriging}), this prediction error can be decomposed via \eqref{eq:hieraDiffSpa_proj2}:
\begin{equation}\label{eq:error4SPG}
    f-\hat{f}_n^{\rm BR} = f-\mathcal{I}_k^{\rm SP} f =\sum_{{\Balpha}\in\mathbb{Z}_{\geq 0}^d}f_{\Balpha}-\sum_{0\leq |{\Balpha}|\leq k+d-1}f_{\Balpha}=\sum_{|{\Balpha}|\geq k+d}f_{\Balpha}.
\end{equation}
Therefore, the error can be bounded by the infinite series:
\begin{equation}
    \|f-\hat{f}_n^{\rm BR}\|\leq \sum_{|\alpha|\geq k+d}||f_\alpha||
    \label{eq:infser}
\end{equation}
for any norm $||\cdot||$.

Let us first take the $L^p$ norm for $\| \cdot \|$ in \eqref{eq:infser}. Note that

\begingroup
\allowdisplaybreaks
\begin{align*}
    \|f-\hat{f}_n^{\rm BR}\|_{L^p}&\leq \sum_{|{\Balpha}|\geq k+d}||f_{\Balpha}||_{L^p}\\
    &=\sum_{|{\Balpha}|\geq k+d}\left\|\sum_{{\Bbeta}\in B_{\Balpha}}c_{{\Balpha},{\Bbeta}}\phi_{{\Balpha},{\Bbeta}(\Bx)}\right\|_{L^p}\stepcounter{equation}\tag{\theequation}\label{eq:deriv1}\\
    &=\sum_{|{\Balpha}|\geq k+d}\left[\sum_{{\Bbeta}\in B_{\Balpha}}c_{{\Balpha},{\Bbeta}}^p\int^{\Bx_{{\Balpha},{\Bbeta}}+\bold{h}_{\Balpha}}_{\Bx_{{\Balpha},{\Bbeta}}-\bold{h}_{\Balpha}}\phi_{{\Balpha},{\Bbeta}}^p(\Bx) \; d\Bx \right]^{\frac{1}{p}}\\
    &=\sum_{|{\Balpha}|\geq k+d}\left[\frac{2^{d-1}}{{(p+1)}^d |B_{\Balpha}|} \sum_{{\Bbeta}\in B_{\Balpha}} c^p_{\Balpha,\Bbeta} \right]^{\frac{1}{p}}\\
    &\leq C \sum_{|{\Balpha}|\geq k+d} 2^{-(\gamma|{\Balpha}|_{\infty}+|{\Balpha}|)}\\
    &\leq C \sum_{|{\Balpha}|\geq k+d} 2^{-(1+\varepsilon)|{\Balpha}|},
\end{align*}
\endgroup
where $C$ and $\epsilon$ are positive constive constant independent of $\Balpha$ (note that the constant $C$ is used to show big-O convergence, and may change in value throughout the proof). Here, the third line follows from the fact that $\{\phi_{\Balpha,\Bbeta}\}_{{\Bbeta}\in B_{\Balpha}}$ is pairwise disjoint, the fourth line follows from the fact that $\int^{\Bx_{\Balpha,\Bbeta}+\bold{h}_{\Balpha}}_{\Bx_{\Balpha,\Bbeta}-\bold{h}_{\Balpha}} \phi_{\Balpha,\Bbeta}^p(\Bx) d\Bx =[2/(p+1)]^d2^{-|{\Balpha}|}$, and the fifth line follows from Lemma \ref{prop:BoundofSurplus} (Equation \ref{eqn:lemi}).

% \cmtS{Lemma \ref{prop:BoundofSurplus} (i)?}

We can further upper bound the last equation as follows:
\begin{align}
\begin{split}
    C \sum_{|{\Balpha}|\geq k+d}2^{-(1+\varepsilon)|{\Balpha}|}&=C \sum_{i=k+d}^\infty2^{-(1+\varepsilon)i}\sum_{|{\Balpha}|=i}1\\
    &=C \sum_{i=k+d}^\infty2^{-(1+\varepsilon)i}{i-1 \choose d-1}\\
    &\leq C 2^{-(1+\varepsilon)k} \cdot 2^{-(1+\varepsilon)d}\sum_{i=0}^\infty2^{-i}{i+k+d-1 \choose d-1},
    \end{split}
    \label{eq:deriv2}
\end{align}
where the second line follows since there are ${i-1 \choose d-1}$ ways to represent $i$ as a sum of $d$ natural numbers. With $x=2^{-1}$, Lemma \ref{lem:id} gives:
\begin{align}
\begin{split}
    \sum_{i=0}^\infty2^{-i}{i+k+d-1 \choose d-1} = 2\sum_{j=0}^{d-1}{k+d-1 \choose j}=2\frac{k^{d-1}}{(d-1)!}+\mathcal{O}(k^{d-2}).
    \end{split}
    \label{eq:deriv3}
\end{align}
Plugging \eqref{eq:deriv3} into \eqref{eq:deriv1}, we get:
\begin{align}
\begin{split}
    \|f-\hat{f}_n^{\rm BR}\|_{L^p}&\leq C\sum_{|{\Balpha}|\geq k+d}2^{-(1+\varepsilon)|{\Balpha}|}\leq C2^{-(1+\varepsilon)(k+d)}\frac{k^{d-1}}{(d-1)!}\\
    &=C2^{-(1+\varepsilon)k}\left[2^{-(1+\varepsilon)(d-1)}\frac{k^{d-1}}{(d-1)!}\right].
    %=\mathcal{O}(2^{-\varepsilon k}2^{-k}[\log 2^k]^{d-1}).
\end{split}
\label{eq:deriv4}
\end{align}
Using the upper bound on grid points for sparse grids (Lemma \ref{lem:num4SPG}), the above prediction error can be stated in terms of sample size $n$:
\begin{align*}
\begin{split}
    \|f-\hat{f}_n^{\rm BR}\|_{L^p} &\leq C2^{-\varepsilon k}2^{-k}[\log 2^k]^{d-1}=2^{-\varepsilon k}n^{-1}[k\log 2]^{2d-2}\\
    &=\mathcal{O}(n^{-(1+\delta)}) = o(n^{-1})
    \label{eq:deriv5}
    \end{split}
\end{align*}
for some $\delta>0$.

For $L^{\infty}$ convergence, we can take the $L^\infty$ norm for $\| \cdot \|$ in \eqref{eq:infser} and mimic the same proof technique for $L^{p}$ convergence, with the key distinction being the use of Lemma \ref{prop:BoundofSurplus} (ii) in \eqref{eq:deriv1} to upper bound $\|\sum_{{\Bbeta}\in B_{\Balpha}} c_{\Balpha,\Bbeta}\phi_{\Balpha,\Bbeta}\|_{L^{\infty}} = \sup_{{\Bbeta}\in B_{\Balpha}}|c_{\Balpha,\Bbeta}|=\CalO(2^{-|{\Balpha}|})$. This yields the following $L^{\infty}$ rate in $n$:
\begin{equation}
\|f-\hat{f}_n^{\rm BR}\|_{L^\infty} = \CalO(n^{-1}[\log n]^{2(d-1)}).
\label{eq:deriv6}
\end{equation}

Finally, using Theorem \ref{thm:diffBdryGP}, the $L^p$ and $L^\infty$ convergence rates for $\|f-\hat{f}_n^{\rm BR}\|$ also hold for the BdryGP error $\|f-\hat{f}_n^{\rm BR}\|$ as well, which completes the proof.
\end{proof}
%The general idea of the proof is to write $||f-f^s_k||_{L^1}$ in the form of equation (\ref{eq:error4SPG}) and then compute an upper bound of $c_{\alpha,\beta}$ defined in equation (\ref{eq:HierarchicalSurplus}). We leave the proof in appendix.
%If we make a more restriction that $f\in\mahtcl{H}^{2}_{0,mix}$, then the convergence rate given in equation (\ref{eq:SGConvergenceRate_L2}) can be found in \cite{Bungartz04}.
%\begin{rem}

\textit{Remark 1}: In Theorem \ref{thm:L1error}, the intuition behind the slower $L^{\infty}$ rate (compared to the $L^{p}$ rate, $1 \leq p < \infty$), is that $D^{\bold{1}}f(\Bx)$ can be ill-behaved on a  measure-zero set on $\CalX$. Because of this, the pointwise convergence rate on this set can be be much slower. The effect from this measure-zero set can be ignored under integration for $\mathcal{L}^p$ with $p < \infty$.
%In fact, this is the finite difference representation of $D^\bold{1}D^\bold{1}f(\Bx_{\alpha,\beta})=D^\bold{2}f(\Bx_{\alpha,\beta})$.  If $f\in\mathcal{H}^2_{0,mix}$, then $c_{\alpha,\beta}$ can be bounded by a much smaller number as shown in \cite{Bungartz04}. Therefore, we can say that $c_{\alpha,\beta}$ measures the distance of $f$ to $\mathcal{H}^2_{0,mix}$. We can also make restriction on the Hölder condition of $D^{\bold{1}}f$. For example, we can assume that $D^{\bold{1}}f\in C^{0,\frac{1}{2}}$ then the convergence rate can be improved to an order between $n^{-1}$ and $n^{-2}$. Such an assumption is reasonable because $C^{0,\frac{1}{2}}$ is a large enough space that it contains all the sample paths of Brownian motion. However, these concepts are related to fractional Sobolev space and are off our goal. At this point, we argue that $n^{-1}$ is a good enough convergence rate since it is independent of dimension $d$.
%\end{rem}

\textit{Remark 2}: We can further improve the convergence rate in Theorem \ref{thm:L1error} if we restrict $f$ to the smaller function space $\CalH^{2,c}_{mix}$, the constrained Sobolev space with mixed \textit{second} derivatives. Using the same proof strategy, but plugging in Lemma 3.5 in \cite{Bungartz04}, we can then show that $||\hat{f}^{\text{BM}}_n-f||_{L^2}=\CalO(n^{-2}[\log n]^{d-1})$ for $f \in \CalH^{2,c}_{mix}$. The function space equivalence (Theorem \ref{thm:4spacesequal}), however, does not hold under this extension, since $\CalH^{2,c}_{mix}$ is smaller than the RKHS of the BdryMat\'ern kernel $\CalH^{1,c}_{mix}$.

% \cite{Bungartz04} has shown that for any $f\in\CalH^{2,c}_{mix}$, the FEM sparse grid interpolator converges to $f$ under $L^2$ norm in $\CalO(n^{-2}[\log n]^{d-1})$ provided full boundary condition. So according to theorem \ref{thm:sggFEMkriging} and \ref{thm:diffBdryGP}, we can plug equation (\ref{eq:diffBdryGP}) in the proof of lemma 3.5 in \cite{Bungartz04} to have the same result $||\hat{f}^{\text{BM}}_n-f||_{L^2}=\CalO(n^{-2}[\log n]^{d-1})$. However, the space $\CalH^{2,c}_{mix}$ is smaller than the RKHS of $k^{\text{BM}}_\omega$, which is $\CalH^{1,c}_{mix}$. That is why we need to generalize the result so that it is consistent with the RKHS of the kernel used for interpolation. Otherwise, the convergence rate does not satisfy the assumptions of GP regression.

\subsection{Probabilistic Uniform Rate}
\label{sec:probrate}
Next, we prove a probabilistic convergence rate  for BdryGP, where $f$ is assumed to be \textit{random}, following a GP with sample paths in the constrained mixed Sobolev space $\mathcal{H}^{1,c}_{mix}$. This is motivated by the probabilistic convergence rates in \cite{Wang18} for GPs without boundary constraints. Define first the following kernel space:
\begin{equation}
\CalH^{1,c}_{mix}( \CalX \times \CalX) := \{k(\Bx,\By): k(\Bx,\cdot),k(\cdot,\By)\in \CalH^{1,c}_{mix}(\CalX), \; \forall \Bx,\By\in \CalX\}.
\label{eq:kerspace}
\end{equation}
Such a space ensures that a GP with kernel $k \in \CalH^{1,c}_{mix}( \CalX \times \CalX)$ has sample paths in $\mathcal{H}_{mix}^{1,c}$.

The following theorem gives a probabilistic uniform rate for BdryGP when $f$ follows a GP with kernel $k \in \CalH^{1,c}_{mix}( \CalX \times \CalX)$:
\begin{theorem}
\label{thm:proberror}
Suppose $I^{[0]} \cup I^{[1]} = [d]$, and assume the sparse grid design $\BX_k^{\rm SP}$ with $n = |\BX_k^{\rm SP}|$. Let $Z(\cdot)$ be a GP with kernel $k \in \CalH^{1,c}_{mix}(\CalX\times\CalX)$, and  $\mathcal{I}_n^{\rm BM}$ be the BdryGP interpolation operator satisfying $\mathcal{I}_n^{\rm BM} f = \hat{f}_n^{\rm BM}$. Then:
\begin{equation}
    \mathbb E\left[\sup_{\Bx\in\CalX}|Z(\Bx)-\mathcal{I}_n^{\rm BM} Z(\Bx)|^p\right]^{\frac{1}{p}}=\CalO(n^{-1}[\log n]^{2d-\frac{3}{2}}), \quad 1 \leq p < \infty,
\label{eq:lpprob}
\end{equation}
and:
\begin{equation}
\sup_{\Bx\in \CalX}|Z(\Bx)-\mathcal{I}_n^{\rm BM}Z(\Bx)|=\CalO_{\mathbb{P}}(n^{-1}[\log n]^{2d-\frac{3}{2}}). 
\label{eq:linfprob}
\end{equation}
% For any $1\leq p< \infty$.
\end{theorem}
\begin{proof}
% Because of the asymptotic equivalence between the Brownian kernel $k^{\rm BR}$ and the BdryMat\'ern kernel $k^{\rm BM}$ (Theorem \ref{thm:diffBdryGP}), we only need to prove this for $k^{\text{BR}}$.
Let $Z(\cdot)$ be a GP with kernel $k \in \CalH^{1,c}_{mix}(\CalX\times\CalX)$, and let $\mathcal{I}^{\rm BM}_n|_{\Bx}$ and $\mathcal{I}^{\rm BM}_n|_{\By}$ be the projection operator $\mathcal{I}^{\rm BM}_n$ in arguments $\Bx$ and $\By$. Consider the following hierarchical expansion of the so-called ``natural distance'' $\boldsymbol{\sigma}$:
\begin{align*}
    \boldsymbol{\sigma}^2(\Bx,\By)&:=\mathbb{E}\big[\big(Z(\Bx)-\mathcal{I}^{\rm BM}_nZ(\Bx)\big)\big(Z(\By)-\mathcal{I}^{\rm BM}_nZ(\By)\big)\big]\\
    &=k(\Bx,\By)-\mathcal{I}^{\rm BM}_n|_{\By}k(\Bx,\By)-\mathcal{I}^{\rm BM}_n|_{\Bx}k(\Bx,\By)-\mathcal{I}^{\rm BM}_n|_{\Bx} \mathcal{I}^{\rm BM}_n|_{\By}k(\Bx,\By)\\
    &=\{\text{I}-\mathcal{I}^{\rm BM}_n|_{\Bx}\}\{\text{I}-\mathcal{I}^{\rm BM}_n|_{\By}\}k(\Bx,\By).
\end{align*}%According to the proof of theorem 1 in \cite{Wang18}, we only need to show that the upper bound of the so called natural distance $\boldsymbol{\sigma}(\Bx,\By)$ for any $\Bx,\By\in\CalX$ induced by the GP $Z(\Bx)-\mathcal{I}^s_kZ(\Bx)$ is in the order $\CalO(n^{-1}[\log n]^{2(d-1)})$.
By Theorem \ref{thm:L1error}, we have:
$$\boldsymbol{\sigma}(\Bx,\By)=\CalO(n^{-1}[\log n)]^{2(d-1)})$$
for any $\Bx,\By\in\CalX$. This can then be plugged into the proof of Theorem 1 of  in \cite{Wang18} to prove the result.
\end{proof}
% \crd{[condense proof + put in main paper.]}

% The proof of this theorem can be found in appendix \ref{apdix:pfproberror}. The key ingredient for proving the above theorem is to expand the kernel function of $Z$ by a new form of expansion which we call hierarchical expansion. From theorem \ref{thm:4spacesequal}, we can expand any kernel $k\in\CalH^{1,c}_{mix}(\CalX\times\CalX)$ by the basis functions $\{\phi_{\Balpha,\Bbeta}\}$:
% \begin{equation}
% \begin{aligned}
%     k(\Bx,\By)&=\sum_{\Balpha\in\mathbb{Z}^d_{\geq 0}}\sum_{\boldsymbol{\gamma}\in\mathbb{Z}^d_{\geq 0}}\sum_{\Bbeta\in B_{\Balpha}}\sum_{\boldsymbol{\zeta}\in B_{\boldsymbol{\gamma}}}c_{(\Balpha,\Bbeta),(\boldsymbol{\gamma},\boldsymbol{\zeta})}[k(\cdot,\cdot)]\phi_{\Balpha,\Bbeta}(\Bx)\phi_{\boldsymbol{\gamma},\boldsymbol{\zeta}}(\By)\\
%     \end{aligned}
% \end{equation}
% where
% $$c_{(\Balpha,\Bbeta),(\boldsymbol{\gamma},\boldsymbol{\zeta})}[k(\cdot,\cdot)]:=\bigg(\prod_{i=1}^dA_{\alpha_i,\beta_i}\bigg)\bigg(\prod_{i=1}^dA_{\gamma_i,\zeta_i}\bigg)k(\bold{X}_{\Balpha},\bold{X}_{\boldsymbol{\gamma}}).$$
%  If the GP is without boundary information, the  connection between FEM and GP is lost and hence the probabilistic convergence rate one can prove is only a direct result of \cite{Wang18}, which is in a lower order as we will show in the next subsection.

\subsection{Comparison with Existing Results}
\label{sec:comp}

\begin{table}[t]
\caption{Convergence rates for BdryGP and existing rates in the literature. ``Optimal design'' refers to optimally-chosen points under a statistical criterion or error bound.}
\label{tab:convergence}
\centering
\begin{tabular}{ c c c c } 
\toprule
\textit{Work} & \textit{Design} & \textit{Type} & \textit{Rate} \\[1ex] 
\toprule
Current & Sparse grid & Deterministic, $L^p$ & $\smallO(n^{-1})$\\
Current & Sparse grid & Deterministic, uniform & $\CalO(n^{-1}[\log n]^{2(d-1)})$\\
\cite{Geer00} & Optimal design & Deterministic, $L^2$ & $\CalO(n^{-\frac{1}{2+d}})$\\
\cite{Wu1993} & Optimal design & Deterministic, uniform & $\CalO(n^{-\frac{1}{2d}})$ \\
\hline
Current & Sparse grid & Probabilistic, uniform & $\CalO_\mathbb{P}(n^{-1}[\log n]^{2d-\frac{3}{2}})$\\
\cite{Wang18} & Optimal Design & Probabilistic, uniform  & $\CalO_{\mathbb{P}}(n^{-\frac{1}{2d}}[\log n^{\frac{1}{2d}}]^{\frac{1}{2}})$\\
% \hline
\cite{Stein1990} & Full grid & Mean square, pointwise & $\CalO(n^{-\frac{1}{2d}})$\\
% \hline
\cite{Ritter00} & Optimal design & Mean square, $L^2$ & $\CalO(n^{-\frac{1}{2d}})$\\
% \hline
\toprule
\end{tabular}
\end{table}

We now compare these BdryGP rates to existing GP rates which do not incorporate boundary information. Table \ref{tab:convergence} summarizes several key results for the latter. Consider first the \textit{deterministic} rates, where $f$ is a deterministic function within a function space. For $f \in \CalH^1(\mathcal{X})$ (the first-order Sobolev space), \cite{Wu1993} proved a $L^\infty$ minimax rate of $\mathcal{O}(n^{1/(2d)})$ for radial basis interpolators. Under the same assumptions, \cite{Geer00} and \cite{Gu02} also proved a $L^2$ minimax rate of $\CalO(n^{-{1}/(2+d)})$ for kernel ridge regression. Without additional information on $f$, these rates are in general not improvable \citep{Stone82}. To contrast, by incorporating boundary information, the proposed BdryGP enjoys quicker convergence rates in sample size $n$, with an $L^p$ rate of $\smallO(n^{-1})$ and an $L^\infty$ rate of $\mathcal{O}(n^{-1}[\log n]^{2(d-1)})$. Furthermore, the BdryGP rates are more resistant to the ``curse-of-dimensionality''. As dimension $d$ grows large, the existing error rate $\mathcal{O}(n^{1/(2d)})$ grows exponentially in sample size $n$, whereas the BdryGP rates grow exponentially in a lower-order term $\log n$ (for $L^\infty$) or in constants (for $L^p$). This shows that, by incorporating boundary information, the BdryGP not only yields lower prediction errors for fixed dimension $d$, but maintains relatively good performance as dimension $d$ grows large.

Consider next the \textit{probabilistic} uniform rates, where $f$ follows a GP with kernel $k \in \CalH^{1,c}_{mix}( \CalX \times \CalX)$, which ensures sample paths are contained in the constrained mixed Sobolev space $\mathcal{H}_{mix}^{1,c}$. These probabilistic uniform GP rates were first studied in \cite{Tuo17} for the Mat\'ern kernel without boundary information. There, the authors proved an $L^p$ rate over the stochastic process (uniform in $x$) of $\mathcal{O}(n^{-1/(2d)} \sqrt{[\log n^{1/(2d)}]})$, and a probabilistic rate (uniform in $x$) of $\CalO_{\mathbb{P}}(n^{-1/(2d)} \sqrt{[\log n^{1/(2d)}]})$. To contrast, by incorporating boundary information, the same uniform rates are improved to $\mathcal{O}(n^{-1} [\log n]^{2d - 3/2})$ and $\CalO_{\mathbb{P}}(n^{-1} [\log n]^{2d - 3/2})$ in Theorem \ref{thm:proberror}, respectively. This again shows that, by incorporating boundary information, the BdryGP can yield lower prediction errors. 

It is worth mentioning that the constrained \textit{mixed} Sobolev space used here imposes greater smoothness than the Sobolev spaces used in existing rates, which may also contribute to our rate improvements. To parse out the effect from different function spaces, we can directly extend results from \cite{rieger17} and \cite{Tuo17} to show that, under \textit{unconstrained} function spaces of comparable smoothness to Theorems \ref{thm:L1error} and \ref{thm:proberror}, we achieve only $L^p$ rates of $\CalO(n^{-1/2}[\log n]^{(5/2)(d-1)})$ and $\CalO_{\mathbb{P}}(n^{-1/2}[\log n]^{(5/2)d-2})$ (a full proof is provided in the Appendix). These rates are of an order slower than the BdryGP rates in Theorems \ref{thm:L1error} and \ref{thm:proberror}, which confirms that boundary information indeed improves predictive performance.

\section{Numerical Experiment}
% Thus far, our paper has proved the convergence rate of kriging with BdryMat\'ern kernel given the boundary condition. In fact, Kriging prediction performances between sparse grid design and space-filling design has been numerically compared in \cite{Plumlee14}  which found that performance of sparse grid is competitive for smooth function and inferior for rough function. Our theorems perfectly explain this experimental result in that the convergence rate is almost $\CalO(n^{-2})$ for any underlying function $f\in\CalH^{2,c}_{mix}$ according to \cite{Bungartz04} while the convergence rate is almost $\CalO(n^{-1})$ for underlying function $f\in\CalH^{1,c}_{mix}$ according to theorem \ref{thm:L1error}.

% We first numerically compare the convergence rate of our Kriging algorithm on a set  of deterministic functions defined on $[0,1]^{10}$, which are 10-dimensional functions. We will compare the log values of $L^1$ error among BdryGP with full boundary information $I^{[0]}=I^{[1]}=[d]$, partial boundary information $I^{[0]}=[d]$, $I^{[1]}=\emptyset$ and GP whose kernel is the tensor product of Matérn kernels associated to $\mathcal{H}^1$, which corresponds to the case of no boundary information. More precisely, the kernel is in the following form:
% $$k(\Bx,\By)=e^{-c||\Bx-\By||_1}.$$

We now provide a small simulation study verifying the improved error convergence rate of the proposed BdryGP model over standard GP models (which do not incorporate boundary information). The set-up is as follows. We use three $d=10$-dimensional test functions from the emulation literature, taken from \cite{Surjano16}:
\begin{align*}
    \textit{Corner peak:} \quad f(\Bx) &= \left(1+\frac{\sum_{j=1}^dx_j}{d}\right)^{-d-1},\\
    \textit{Product peak:} \quad f(\Bx) &= \prod_{j=1}^d\left(1+10(x_j-0.25)^2\right)^{-1},\\
    \textit{Rosenbrock:} \quad f(\Bx) &= 4\sum_{j=1}^{d-1}(x_j-1)^2+400\sum_{j=1}^{d-1}\left((x_j-0.5)-2(x_j-0.5)^2\right)^2.
\end{align*}
We will compare two variants of the BdryGP model: (i) the BdryGP with \textit{full} boundary information (i.e., $I^{[0]} = I^{[1]} = [d]$), and (ii) the BdryGP with only \textit{partial} information on left boundaries (i.e., $I^{[0]} = [d]$, $I^{[1]} = \emptyset$), with a standard GP model with the product Mat\'ern-1/2 kernel. All models use a wavelength parameter of $\omega=1.0$, and are compared on the prediction error $\|f - \hat{f}\|_{L^1}$, which is approximated using 1000 uniformly sampled points in $\mathcal{X}$.

% The  errors will be estimated by the sample average of the absolute and square prediction errors at 1000 randomly selected points respectively. We will consider the following functions: 

%\begin{equation*}
%    \mu^*(\Bx)=\int_{[0,1]^d}\Phi_{\text{dist}(\Bx,\partial B)}(\Bx-\By)\mu(\By)d\By
%\end{equation*}
%where $\Phi_{\text{dist}(\Bx,\partial B)}(\cdot)$ is a mollifier that smooths $\mu$ on the interior of $[0,1]^d$ and leaves everything unchanged on the boundary.  For our numerical experiments,  $\mu$ is used as our mean function.

\begin{figure}
\centering
\includegraphics[width=0.32\textwidth]{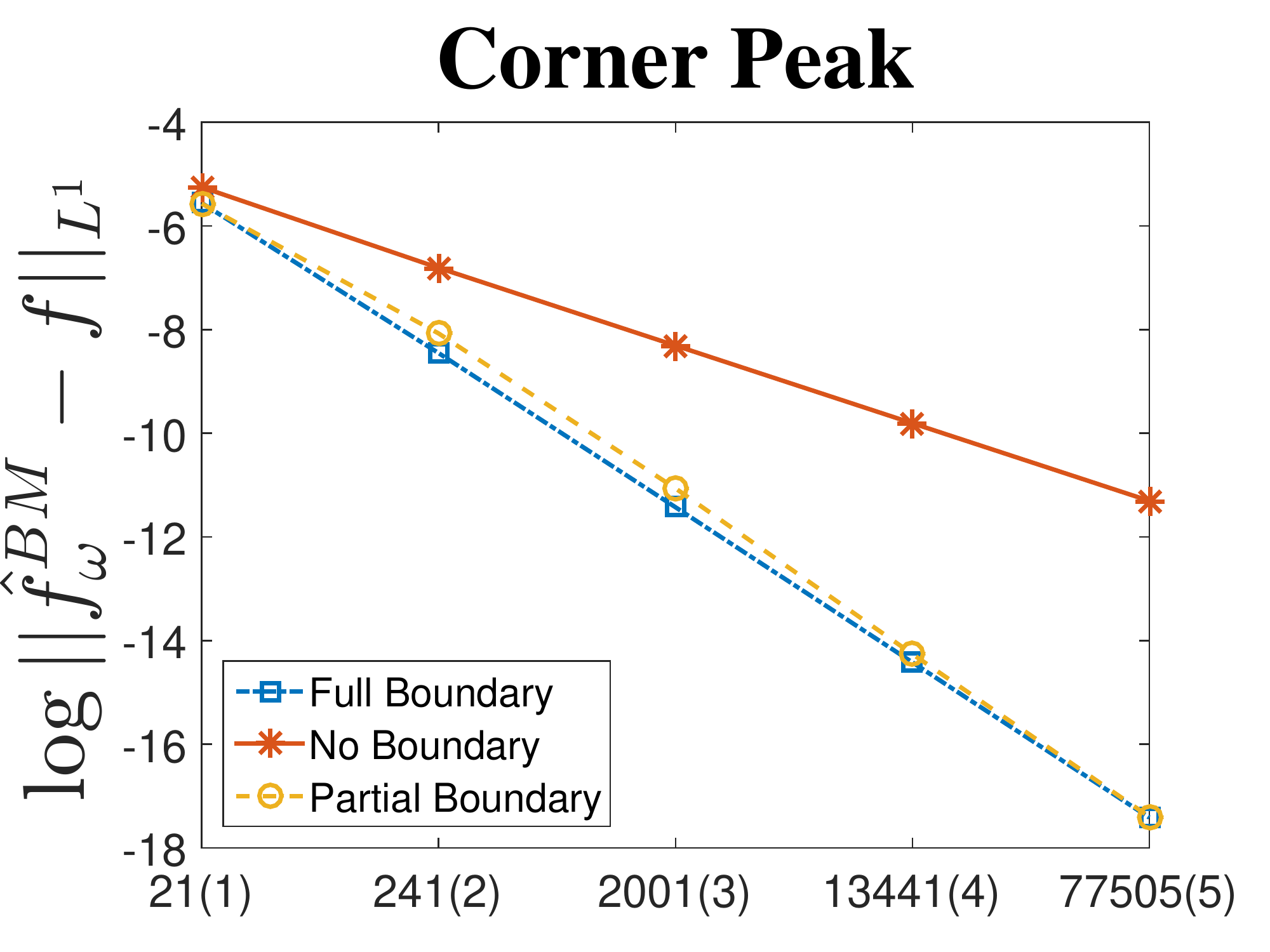}
\includegraphics[width=0.32\textwidth]{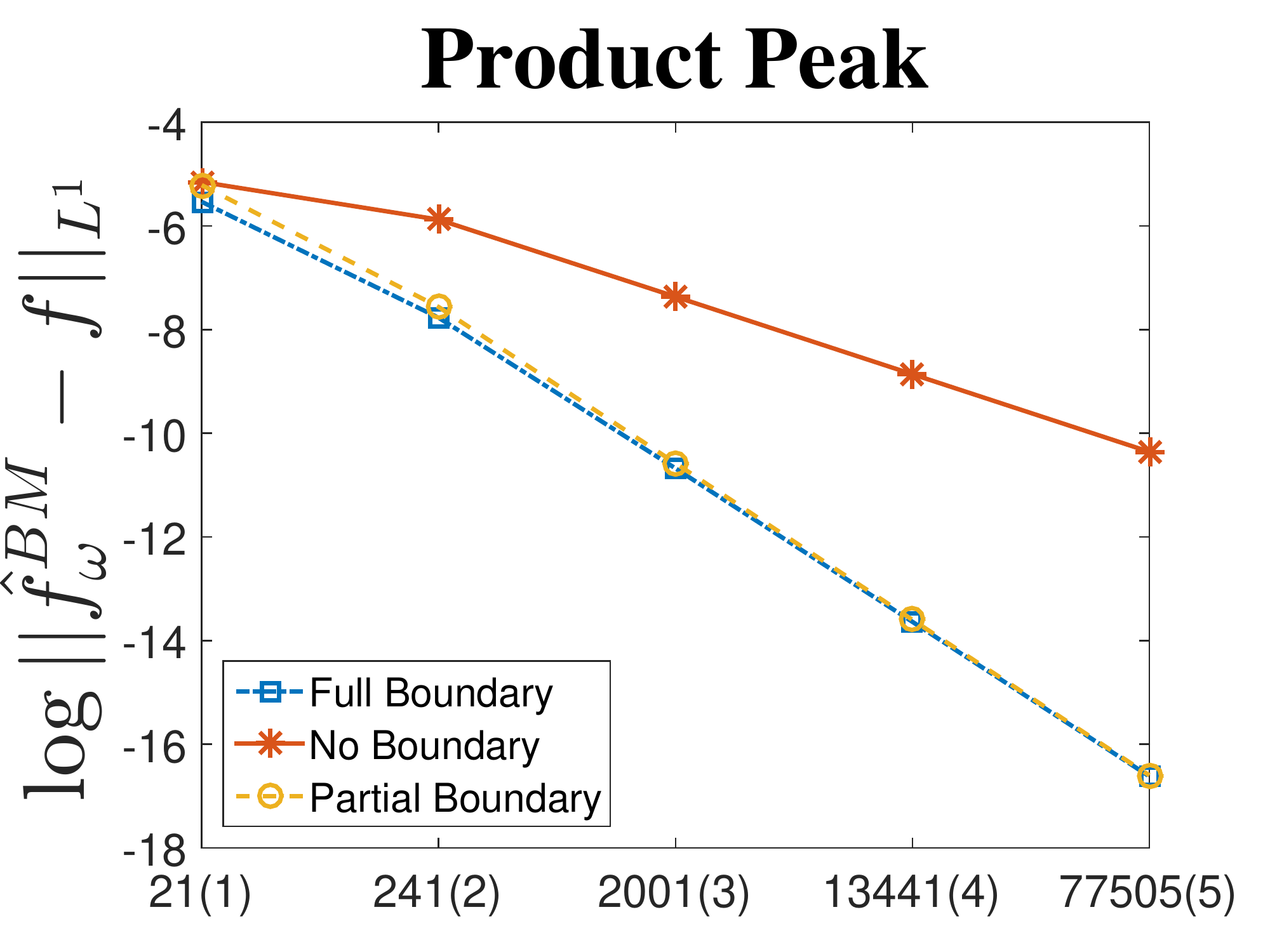}
\includegraphics[width=0.32\textwidth]{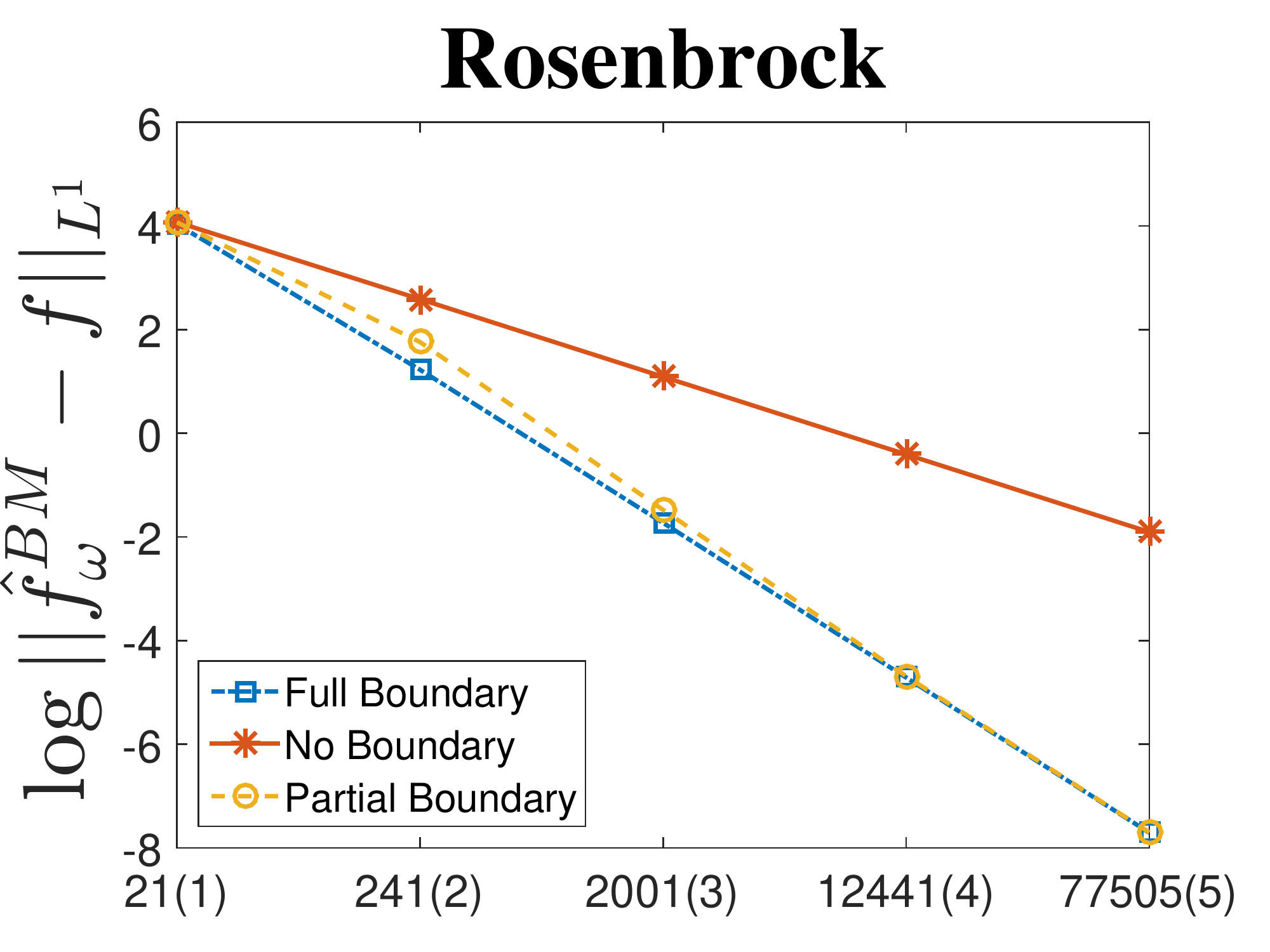}

\hspace{0.05cm}
\includegraphics[width=0.305\textwidth]{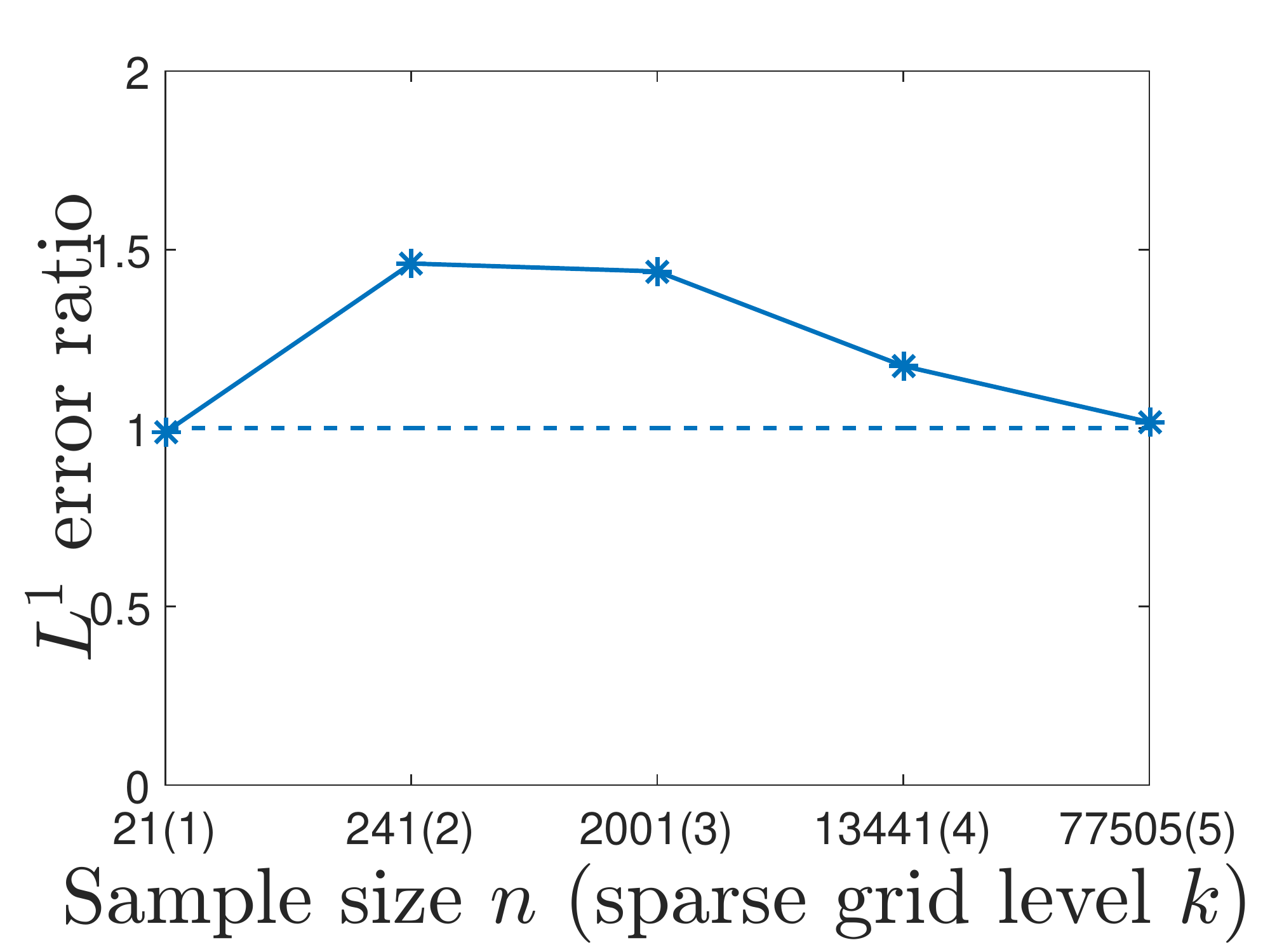}
\hspace{0.05cm}
\includegraphics[width=0.305\textwidth]{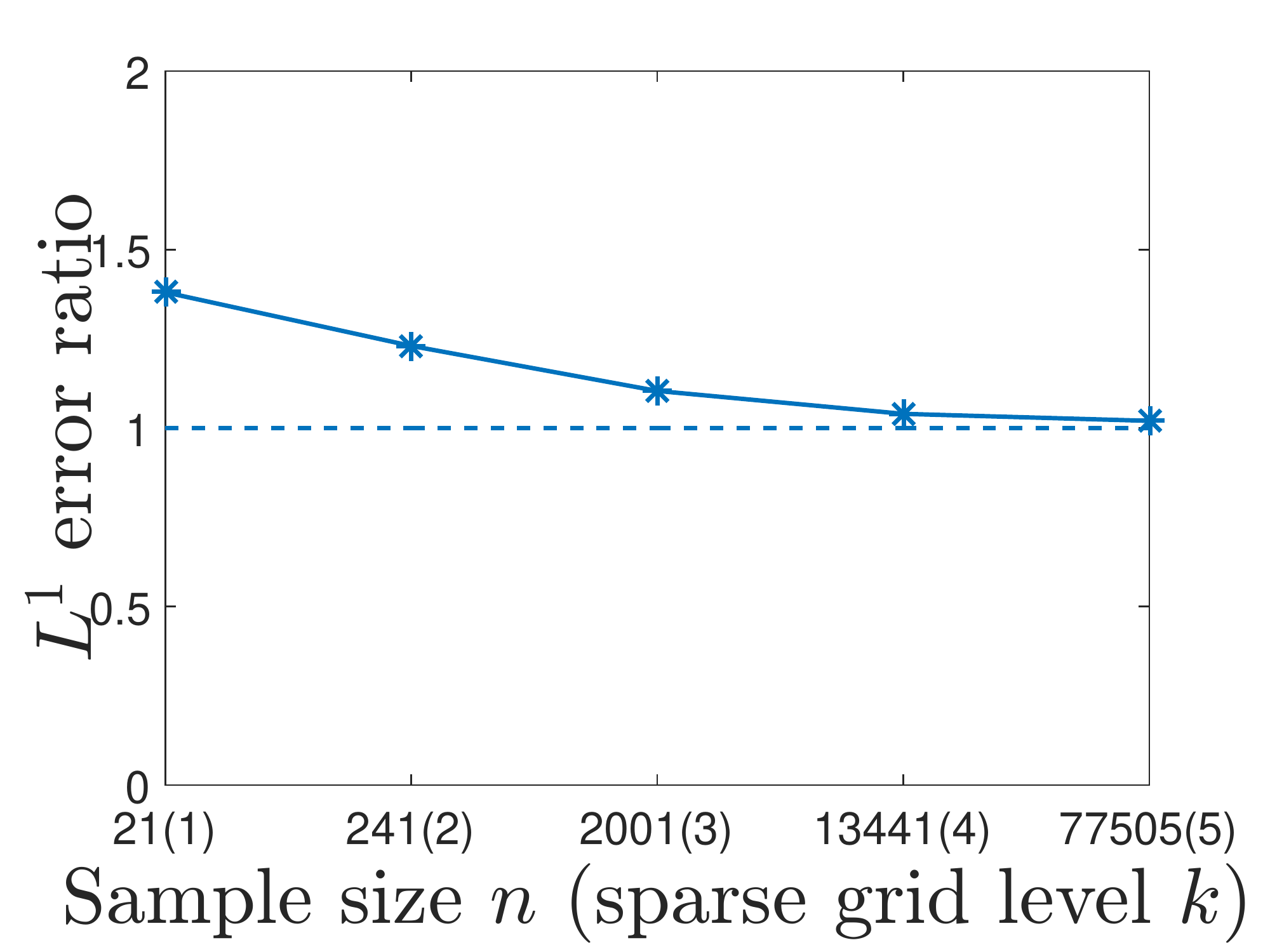}
\hspace{0.05cm}
\includegraphics[width=0.305\textwidth]{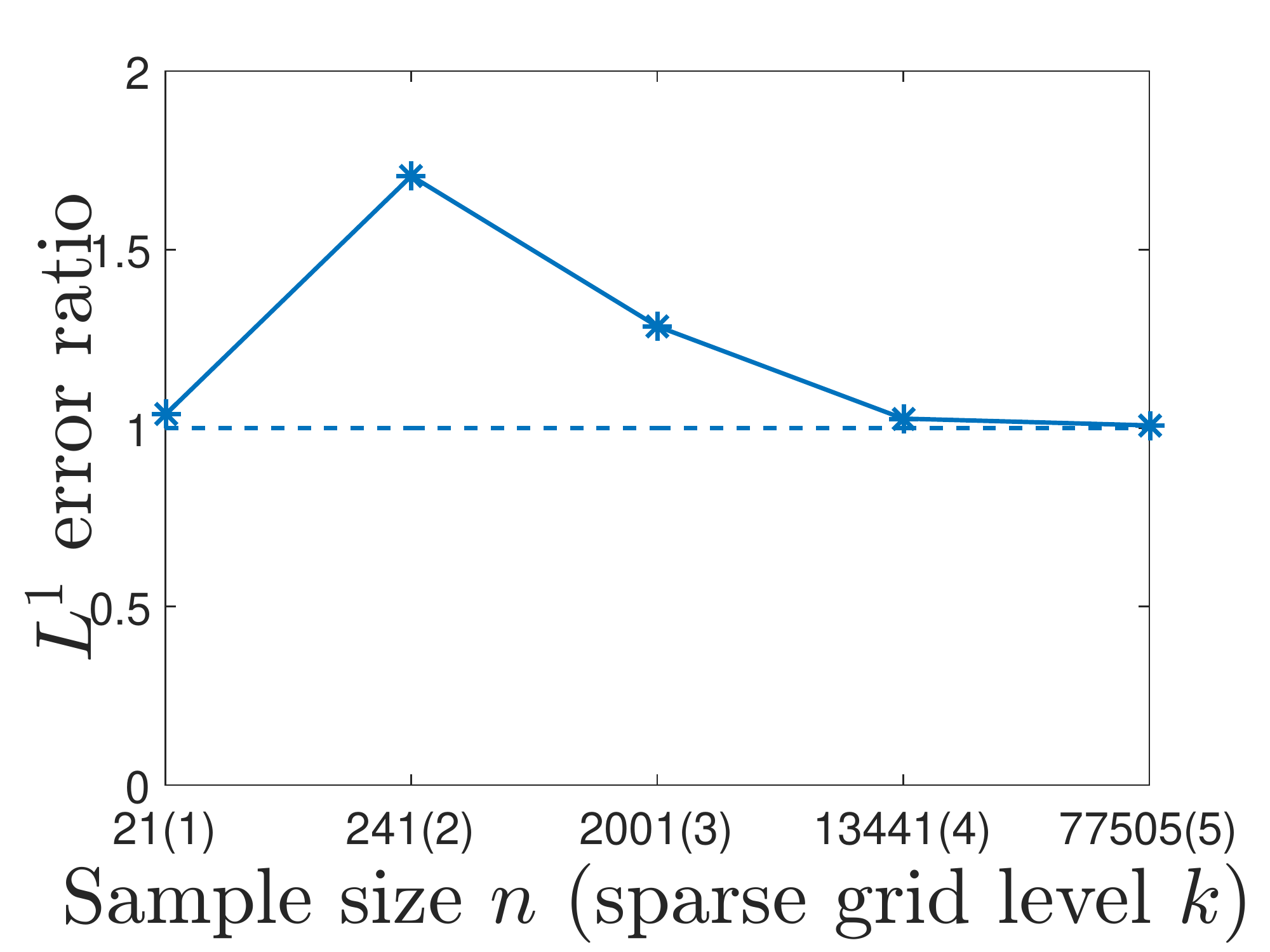}

\caption{(Top): Log-error $\log \|\hat{f}_n^{\rm BM}-f\|_{L^1}$ as a function of sample size $n$ (and sparse grid level $k$). (Bottom): The $L^1$ error ratio of the full boundary GP over the partial boundary GP, as a function of sample size $n$ (and sparse grid level $k$).}
\label{fig:logError}
\end{figure}

Figure \ref{fig:logError} (top) plots the log-error $\|f - \hat{f}\|_{L^1}$ as a function of sample size $n$ (and sparse grid level $k$). For all three functions, these log-errors appear to be linearly decreasing in sparse grid level $k$. Furthermore, the two BdryGP models (both of which incorporate some form of boundary information) yield much lower errors than the standard GP without boundary information, with the error decay slopes for BdryGP roughly double that for the standard GP model. This is in line with the convergence rates proven in Section \ref{sec:convrates}, which show that the $L^1$ error rates for BdryGP are on the order of $\CalO(2^{-k})$, but increase to $\CalO(2^{-{k}/{2}})$ without boundary information. 

To highlight the error gap between full and partial boundary information, Figure \ref{fig:logError} (bottom) plots the $L^1$ error ratio of the full boundary BdryGP over the partial boundary BdryGP. All ratios are above 1.0, which shows that full boundaries indeed yield more information on $f$ compared to partial boundaries. However, this improvement seems to diminish as sample size $n$ grows large; this suggests that the information on $f$ from design points can outweigh the additional information from full boundaries (over partial boundaries) for large sample sizes.

% According to the figures for the results of our study, we can obviously see that the log errors are approximately linear in the level of the associated sparse grid design.  This is because the $L^1$ errors for BdryGP and tensor Matérn are approximately $\CalO(2^{-k})$ and $\CalO(2^{-\frac{k}{2}})$ according to theorem \ref{thm:L1error} and theorem \ref{thm:lp_noBoundary}, where $k$ is the level of the sparse grid. The log error of BdryGP is almost one half of that of tensor Matérn according to our theorem and according to our numerical experiments. This illustrate our theorem and show that boundary information does help to improve the convergence rate. 
\iffalse
% We then investigate the mean square  error resulting from our Kriging algorithm. This can be thought of as the average mean squared prediction error over all possible sample paths drawn from BdryGP. We compared the average mean squared  error (MSE) resulting from the design strategies computed through 1000 Monte Carlo samples on $[0,1]^5$ from level $1$ to level $30$.
\begin{figure}[h]
\centering
\includegraphics[width=0.50\textwidth]{figures/logerror/logMSE.pdf}

\caption{log of MSE versus the level of associated sparse grid design on $[0,1]^5$\label{fig:logError} \crd{[update]}}
\end{figure}

We can notice that the log MSE curve is approximately linear in level $k$ when $k$ becomes large enough. According to theorem \ref{thm: HighProb_con}, the slope of log MSE is close to $-1$ for large $k$ which coincides with our numerical experiments. 
\fi

\section{Conclusion}
This paper presents a new Gaussian process model, called BdryGP, for incorporating one type of boundary information with provably improved convergence rates. The key novelty in BdryGP is a new BdryMat\'ern covariance function, which inherits the same smoothness properties of a tensor Mat\'ern kernel, while constraining sample paths to satisfy boundary information almost surely. Using a new connection between finite-element modeling and GP interpolation, we then show that under sparse grid designs, BdryGP enjoys improved convergence rates over standard GP models, which do not account for boundary information. By incorporating boundaries, our BdryGP rates are also more resistant to the well-known ``curse-of-dimensionality'' in nonparametric regression. Numerical simulations confirm these improved convergence rates, and demonstrate the improved performance of BdryGP over standard GP models.

While this paper provides an appealing theoretical framework for the BdryGP model, there are further developments which would be useful for practical implementation. For computational efficiency, one can leverage the equivalence between FEM and BdryGP (Section \ref{sec:fem}) to eliminate matrix computation steps for prediction and likelihood evaluations, which improves the scalability of BdryGP for big datasets. It would also be useful to investigate the behavior of BdryGP (e.g., consistency and convergence rates) under maximum likelihood estimation of model parameters.

\bibliographystyle{imsart}
\bibliography{imsart}

\pagebreak
\appendix
\section{}

%\crd{[pls update / correct notation and other points in proofs. i've highlighted a few below, but it's not an exhaustive list.]}
\subsection{Proof for Lemma \ref{lem:finitespace}}
\label{apdix:pf4hierdiff}
\begin{proof}
WLOG, we assume $I^{[0]}=[d]$ and $I^{[1]}=\emptyset$. Let
\begin{equation*}
    \bigotimes_{j=1 }^d\bigtriangleup_if(\Bx_{\Balpha,{\Bbeta}}):=\bigg(\prod_{j=}^d A_{\alpha_i,\beta_i}\bigg)f(\bold{X})
\end{equation*}
denote the Hierarchical Surplus $c_{\Balpha,\Bbeta}$ where the operator $\bigtriangleup_i$ is defined as:
\begin{equation*}
     \bigtriangleup_if(\Bx_{\Balpha,\Bbeta}):=
     \begin{cases}
     -\frac{1}{2}f(\Bx_{\Balpha,\Bbeta}+2^{-\alpha_i}e_i)+f(\Bx_{\Balpha,\Bbeta})-\frac{1}{2}f(\Bx_{\Balpha,\Bbeta}-2^{-\alpha_i}e_i) &\text{if} \ \alpha_i\geq 1,\\
     f(1) - f(0) &\text{if} \ \alpha_i=0.
     \end{cases}
\end{equation*}
Obviously, $\bigtriangleup_i$ is a linear operator. Let $\pi_{\mathcal{V}_{\Balpha}}[\cdot]$ be the projector to the space $\mathcal{V}_{\Balpha}$ as the one in equation (3.19) in \cite{Bungartz04}. The projection operator $\pi_{\mathcal{V}_\bold{n}}[\cdot]$  can be written as:
\begin{align*}
    &\ \ \ \ \  \pi_{\mathcal{V}_\bold{n}}[\cdot]\\
    &=\sum_{|\Balpha|_\infty\leq \bold{n}}\sum_{\Bbeta\in B_{\Balpha}}\phi_{\Balpha,\Bbeta}\bigotimes_{i=1}^d\bigtriangleup_i\\
    &=\sum_{|{\Balpha}|_\infty\leq \bold{n}}\sum_{\beta_d\in B_{\Balpha}}\phi_{\alpha_d,\beta_d}\bigtriangleup_d\sum_{\beta_{d-1}\in B_{\Balpha}}\phi_{\alpha_{d-1},\beta_{d-1}}\bigtriangleup_{d-1}\cdots\sum_{\beta_1\in B_{\Balpha}}\phi_{\alpha_1,\beta_1}\bigtriangleup_1\\
    &=\sum_{\alpha_d\leq n_d}\sum_{\beta_d\in B_{\Balpha}}\phi_{\alpha_d,\beta_d}\bigtriangleup_d\sum_{\alpha_{d-1}\leq n_{d-1}}\sum_{\beta_{d-1}\in B_{\Balpha}}\phi_{\alpha_{d-1},\beta_{d-1}}\bigtriangleup_{d-1}\cdots\sum_{\alpha_1\leq n_1}\sum_{\beta_1\in B_{\Balpha}}\phi_{\alpha_1,\beta_1}\bigtriangleup_1.
\end{align*}
As a result the posterior mean of BdryGP with Brownian kernel $\hat{f}_\bold{n}^{\text{BR}}$  and $\pi_{\mathcal{V}_{\bold{n}}}$ lie in  tensor product of spaces and we only need to show the equation holds for 1-d functions. We prove the equation by induction. when $n=1$, then according to  Theorem \ref{thm:LagrangeKrigingEquivalent}:
\begin{align*}
    \hat{f}_1^{\text{BR}}(x)&=
    \begin{cases}
    & 2f(\frac{1}{2})x \ \ \text{if} \  x\leq \frac{1}{2}\\
    & 2f(\frac{1}{2})(1-x)+f(1)(2x-1)\ \ \text{if} \ x>\frac{1}{2}
    \end{cases}\\
    &=[-\frac{1}{2}f(0)+f(\frac{1}{2})-\frac{1}{2}f(1)]\phi_{1,1}(x)+[-f(0)+f(1)]\phi_{0,1}(x)\\
    &=\pi_{\mathcal{V}_1}[f](x).
\end{align*}
Suppose the equation holds for $n=k$, and WLOG, suppose $x\in(x_{k,\beta_k},x_{k,\beta_k+1})$ for some $x_{k,\beta_k},x_{k,\beta_k+1}\in\BX_\bold{n}$ and $\beta_k$ is odd. So we have:
\begin{align*}
    \hat{f}_k^{\text{BR}}(x)&=f(x_{k,\beta_k})\phi_{k,\beta_k}(x)+f(x_{k,\beta_k+1})\phi_{k,\beta_k+1}(x)\\
    &=\sum_{n\leq k}\sum_{\beta\in B_n}c_{n,\beta}\phi_{n,\beta}(x)=\pi_{\mathcal{V}_k}[f]
\end{align*}
with $x\in \text{supp}[\phi_{n,\beta_n}]$ and $\beta_n\in B_n$. When $n=k+1$, we have:
\begin{align*}
    \pi_{\mathcal{V}_{k+1}}[f]&=\sum_{n\leq k+1}\sum_{\beta\in B_n}c_{n,\beta}\phi_{n,\beta}(x)\\
    &=f(x_{k,\beta_k})\phi_{k,\beta_k}(x)+f(x_{k,\beta_k+1})\phi_{k,\beta_k+1}(x)+c_{k+1,\beta_{k+1}}\phi_{k+1.\beta_{k+1}}(x).
\end{align*}
According to the following identities:
\begin{align*}
    &x_{k,\beta_k}=x_{k+1,\beta_{k+1}-1}\\
    &x_{k,\beta_k}=x_{k+1,\beta_{k+1}+1}
\end{align*}
and, WLOG, conditioned on the assumption $x\in(x_{k,\beta_k},x_{k+1,\beta_{k+1}})$, we have
\begin{align*}
    &\phi_{k,\beta_k}(x)=\frac{x_{k,\beta_k+1}-x}{2^{-k}}\\
    &\phi_{k,\beta_k+1}(x)=\frac{x-x_{k,\beta_k}}{2^{-k}}\\
    &\phi_{k+1,\beta_{k+1}}(x)=\frac{x-x_{k,\beta_k}}{2^{-k-1}}.
\end{align*}
Now we plug in equation (\ref{eq:HierarchicalSurplus}) and the above identities, we can have the result:
\begin{align*}
    \pi_{\mathcal{V}_{k+1}}[f]&=f(x_{k+1,\beta_{k+1}-1})\phi_{k+1,\beta_{k+1}-1}(x)+f(x_{k+1,\beta_{k+1}})\phi_{k+1,\beta_{k+1}}(x)\\
    &=\hat{f}^{\text{BR}}_{k+1}(x).
\end{align*}
\end{proof}

\subsection{Proof of Theorem \ref{thm:diffBdryGP}}
\label{apdix:pf4diffBdryGP}
\begin{proof}
WLOG, we assume that the mean function $\mu=0$. Let $f\in\CalH^{1,c}_{mix}$, then the difference between the two interpolator $\hat{f}^{\rm BR}_{n}$ and $\hat{f}^{\rm BM}_{n}$ conditioned on $\BX^{\rm SP}_k$ can be written as
\begin{align*}
    \delta(\Bx)&:= |\hat{f}^{\rm BR}_{n}(\Bx)-\hat{f}^{\rm BM}_{n}(\Bx)|\\
    &=|[f(\Bx)-\hat{f}^{\rm BR}_{n}(\Bx)]-[f(\Bx)-\hat{f}^{\rm BM}_{n}(\Bx)]|.
\end{align*}
We first define the following hierarchical difference functions:
\begin{align*}
    \Delta^{\text{BR}}_{\alpha_j}[f](x_j)
    &:=\{k^{\text{BR}}(x_j,\BX_{\alpha_j})[k^{\text{BR}}(\BX_{\alpha_j},\BX_{\alpha_j})]^{-1}-\\
    &\ \ \ \ k^{\text{BR}}(x_j,\BX_{\alpha_j-1})[k^{\text{BR}}(\BX_{\alpha_j-1},\BX_{\alpha_j-1})]^{-1}\}f(\BX_{\alpha_j})\\
    \Delta^{\text{BM}}_{\alpha_j}[f](x_j)
    &:=\{k^{\text{BM}}_{\omega_j}(x_j,\BX_{\alpha_j})[k^{\text{BM}}_{\omega_j}(\BX_{\alpha_j},\BX_{\alpha_j})]^{-1}-\\
    &\ \ \ \ k^{\text{BM}}_{\omega_j}(x_j,\BX_{\alpha_j-1})[k^{\text{BM}}_{\omega_j}(\BX_{\alpha_j-1},\BX_{\alpha_j-1})]^{-1}\}f(\BX_{\alpha_j}).
\end{align*}
According to equation (2) and (3) in \cite{Barthelmann00},  we have the following expansion of the error terms:
\begin{align*}
    &f(\Bx)-\hat{f}^{\rm BR}_{n}(\Bx)=\sum_{|\Balpha|\geq k+d}\bigotimes_{j=1}^d\Delta^{\text{BR}}_{\alpha_j}[f](x_j)\\
    &f(\Bx)-\hat{f}^{\rm BM}_{n}(\Bx)=\sum_{|\Balpha|\geq k+d}\bigotimes_{j=1}^d\Delta^{\text{BM}}_{\alpha_j}[f](x_j)
\end{align*}
where
$$\sum_{|\Balpha|\geq k+d}\bigotimes_{j=1}^d\Delta^{\text{BR}}_{\alpha_j}[f](x_j)= \sum_{|\Balpha|\geq k+d}f_{\Balpha}(\Bx)$$
We now want to write the expansion of $f(\Bx)-\hat{f}^{\rm BM}_{n}(\Bx)$  in terms of $\{\Delta^{\text{BR}}_{\alpha_j}[f](x_j)\}$. According to Theorem 2 in \cite{DingZhang18}, for any 1-d function $f\in\mathcal{H}^{1,c}_{mix}$, we can write the BLUE of kernel $k^{\text{BM}}_{\omega_j}$ explicitly:
\begin{align*}
    & \ \ \ \ k^{\text{BM}}_{\omega_j}(x_j,\BX_{\alpha_j})[k^{\text{BM}}_{\omega_j}(\BX_{\alpha_j},\BX_{\alpha_j})]^{-1}f(\BX_{\alpha_j})\\
    &=\frac{\sinh[\omega_j(x_{\alpha_j,\beta_j+1}-x_j)]}{\sinh[\omega_j2^{-\alpha_j}]}f(x_{\alpha_j,\beta_j})+\frac{\sinh[\omega_j(x_j-x_{\alpha_j,\beta_j})]}{\sinh[\omega_j2^{-\alpha_j}]}f(x_{\alpha_j,\beta_j+1})\\
    &=\frac{x_{\alpha_j,\beta_j+1}-x_j}{2^{-\alpha_j}}f(x_{\alpha_j,\beta_j})+\frac{x_j-x_{\alpha_j,\beta_j}}{2^{-\alpha_j}}f(x_{\alpha_j,\beta_j+1})+\CalO(2^{-2\alpha_j})\\
    &=k^{\text{BR}}(x_j,\BX_{\alpha_j})[k^{\text{BR}}(\BX_{\alpha_j},\BX_{\alpha_j})]^{-1}f(\BX_{\alpha_j})+\CalO(2^{-2\alpha_j})
\end{align*}
where $x_{\alpha_j,\beta_j}$ and $x_{\alpha_j,\beta_j+1}$ are the points that satisfy $x_j\in [x_{\alpha_j,\beta_j},x_{\alpha_j,\beta_j+1}]$, the second equality of the above equation is from Taylor expansion, and the last equality is from the proof of Theorem \ref{thm:LagrangeKrigingEquivalent}.  So the following equality holds:
\begin{align*}
    \bigotimes_{j=1}^d\Delta^{\text{BM}}_{\alpha_j}(x_j)&=\bigotimes_{j=1}^d\{\Delta^{\text{BR}}_{\alpha_j}[f](x_j)+\CalO(2^{-2\alpha_j})\}\\
    &=\bigotimes_{j=1}^d\Delta^{\text{BR}}_{\alpha_j}[f](x_j)+\sum_{j=1}^d\CalO\big(2^{-\alpha_j}\bigotimes_{j=1}^d\Delta^{\text{BR}}_{\alpha_j}[f](x_j)\big)
\end{align*}
where the second equality is from the fact that $\Delta^{\text{BR}}_{\alpha_j}[f](x_j)$ is in an order no smaller than $\CalO(2^{-2\alpha_j})$. Therefore, we can have the final result:
\begin{align*}
    f(\Bx)-\hat{f}^{\rm BM}_{n}(\Bx)&=\sum_{|\Balpha|\geq k+d}\bigotimes_{j=1}^d\Delta^{\text{BM}}_{\alpha_j}[f](x_j)\\
    &=\sum_{|\Balpha|\geq k+d}\left[1+\sum_{j=1}^d\CalO(2^{-\alpha_j})\right]f_{\Balpha}(\Bx)\\
    &=\CalO\left(\sum_{|\Balpha|\geq k+d}f_{\Balpha}\right).
\end{align*}
\end{proof}

\subsection{Proof of Lemma \ref{prop:BoundofSurplus}}

\begin{proof} Let $\bold{i}$ and $\bold{h}$  denote $(i_1,i_2,\cdots,i_d)$ and $(2^{-\alpha_1},\cdots,2^{-\alpha_d})$ respectively, and let $\bold{i}\bold{h}$ denote $(i_12^{-\alpha_1},\cdots,i_d2^{-\alpha_d})$. Let $f(\Bx_{I};\Bx)$ denote $f$ with fixed $x_i, i\not\in I$. According to equation (\ref{eq:HierarchicalSurplus}), we write $c_{\Balpha,\Bbeta}$ as
\begin{align*}
    c_{\Balpha,\Bbeta}&=\bigg(\prod_{i=1}^dA_{\alpha_i,\beta_i}\bigg)f(\bold{X})\\
    &=\sum_{i_d=-1}^{1}\cdots\sum_{i_1=-1}^1\bigg(\frac{-1}{2}\bigg)^{\sum_{j=1}^d|i_j|}f(\Bx_{\Balpha,\Bbeta}+\bold{ih})\\
    &=\sum_{i_d=-1}^{1}\cdots\sum_{i_{2}=-1}^1\bigg(\frac{-1}{2}\bigg)^{\sum_{j=2}^{d}|i_j|}\left(-\frac{1}{2}\right)\cdot\\
    &\quad \quad \quad \quad \quad \int_{\Bx_{\alpha_1,\beta_1}}^{\Bx_{\alpha_1,\beta_1}+h_1}\partial_{x_1}[f(s_1;\Bx_{\Balpha,\Bbeta})-f(s_1-h_1;\Bx_{\Balpha,\Bbeta})]ds_1\\
    %&=\sum_{i_d,\cdots,i_3=-1}^{1}\bigg(\frac{-1}{2}\bigg)^{\sum_{j=3}^{d}|i_j|}\int_{\Bx_{\alpha_2,\beta_2}}^{\Bx_{\alpha_2,\beta_2}+h_2}\int_{\Bx_{\alpha_1,\beta_1}}^{\Bx_{\alpha_1,\beta_1}+h_1}\partial_{x_1x_2}\sum_{i_1,i_2=0}^1(-1)^{|i_1|+|i_2|}f(s_1-i_1h_1,s_2-i_2h_2;\Bx_{\alpha,\beta})ds_1ds_2\\
    &=\left(-\frac{1}{2}\right)^d\int_{\Bx_{\alpha_d,\beta_d}}^{\Bx_{\alpha_d,\beta_d}+h_d}\cdots\int_{\Bx_{\alpha_1,\beta_1}}^{\Bx_{\alpha_1,\beta_1}+h_1}D^{\bold{1}}\sum_{i_1,\cdots,i_d=0}^1(-1)^{\sum_{j=1}^d|i_j|}f(\bold{s}-\bold{ih})d\bold{s}.
\end{align*}
When $d=1$, any function in $\CalH^{1,c}_{mix}\subset\CalH^1(\CalX)$ can be extended to trace-zero function, which is the limit of a sequence of smooth functions under the $\CalH^{1,c}_{mix}$ norm (Theorem 5.5.2 of \cite{Evans15}). When $d>1$, any $f\in\CalH^{1,c}_{mix}$ is also the limit of a sequence of smooth functions $\{g^n\}$ under  the $\CalH^{1,c}_{mix}$ norm because $\CalH^{1,c}_{mix}$ is the tensor product of 1-d function spaces. Therefore, according to Lebesgue differentiation theorem, for almost all $\Bx \in\mathcal{X}$:
\begin{align*}
    &\ \ \ \  \int_{\Bx}^{\Bx+h}|D^{\bold{1}}f(\bold{s})-D^{\bold{1}}f(\bold{s}-h)|d\bold{s}\\
    &\leq \int_{\Bx-h}^{\Bx+h}|D^{\bold{1}}f(\bold{s})-D^{\bold{1}}g^n(\bold{s})|d\bold{s}+\int_{\Bx}^{\Bx+h}|D^{\bold{1}}g^n(\bold{s})-D^{\bold{1}}g^n(\bold{s}-h)|d\bold{s}\\
    &\leq Ch^{1+\gamma}
\end{align*}
where the last line is because the first term of the second line can be arbitrarily small by letting $n$ large from the trace-zero theorem and the second term is the difference of two smooth functions and hence we can use Hölder's condition to have an upper bound.

Now, WLOG, we assume $\alpha_1=|{\Balpha}|_\infty$ and then, as long as there is no singular point $\Bx$ that does not satisfy the Hölder condition near $\Bx_{\Balpha,\Bbeta}$, we can have:
\begin{align*}
    |c_{\Balpha,\Bbeta}|&=\left(\frac{1}{2}\right)^d|\int_{\Bx_{\alpha_d,\beta_d}}^{\Bx_{\alpha_d,\beta_d}+h_d}\cdots\int_{\Bx_{\alpha_1,\beta_1}}^{\Bx_{\alpha_1,\beta_1}+h_1}D^{\bold{1}}\sum_{i_1,\cdots,i_d=0}^1(-1)^{\sum_{j=1}^d|i_j|}f(\bold{s}-ih)d\bold{s}|\\
    &\leq \left(\frac{1}{2}\right)^d\int_{\Bx_{\alpha_d,\beta_d}}^{\Bx_{\alpha_d,\beta_d}+h_d}\cdots\int_{\Bx_{\alpha_1,\beta_1}}^{\Bx_{\alpha_1,\beta_1}+h_1}\\
    & \quad \quad \quad \quad  \sum_{i_2,\cdots,i_d=0}^1|D^{\bold{1}}f(s_1;\bold{s}-\bold{i}\bold{h})-D^{\bold{1}}f(s_1-h_1;\bold{s}-\bold{i}\bold{h})|d\bold{s}\\
    &\leq Ch_1^\gamma\prod_{i=1}^dh_i\\
    &=C2^{-\{\gamma|{\Balpha}|_{\infty}+|{\Balpha}|\}},
\end{align*}
 where the third line is from the inequality from Lebesgue differentiation theorem we proved previously.  If the Hölder condition fails at a specific point, then we can begin with the second line of the above equation to derive that $|c_{\Balpha,\Bbeta}|=\CalO(2^{-|{\Balpha}|})$ via the inequality:
\begin{align*}
  &  \left(\frac{1}{2}\right)^d \int_{\Bx_{\alpha_d,\beta_d}}^{\Bx_{\alpha_d,\beta_d}+h_d}\cdots\int_{\Bx_{\alpha_1,\beta_1}}^{\Bx_{\alpha_1,\beta_1}+h_1}\\ & \ \ \ \ \sum_{i_2,\cdots,i_d=0}^1|D^{\bold{1}}f(s_1;\bold{s}-ih)-D^{\bold{1}}f(s_1-h_1;\bold{s}-ih)|d\bold{s}\\
  &\leq \frac{1}{2^d}2^{-|{\Balpha}|}\sum_{i_2,\cdots,i_d=0}^1 ||D^{\bold{1}}f(s_1;\bold{s}-ih)-D^{\bold{1}}f(s_1-h_1;\bold{s}-ih)||_{L^2(\mathcal{X})}.
\end{align*}
This proves the claim.

\end{proof}

\subsection{Proof of Lemma \ref{lem:id}}
\begin{proof}
The result follows from direct calculations:
\begin{align*}
    & \ \ \ \ \sum_{i=0}^\infty x^i{i+k+d-1 \choose d-1}\\
    &=\frac{x^{-k}}{(d-1)!}(\sum_{i\geq 0}x^{i+k+d-1})^{(d-1)}\\
    &=\frac{x^{-k}}{(d-1)!}\left(x^{k+d-1}\frac{1}{1-x}\right)^{(d-1)}\\
    &=\frac{x^{-k}}{(d-1)!}\sum_{j=0}^{d-1}{d-1 \choose j}(x^{k+d-1})^{(j)}\left(\frac{1}{1-x}\right)^{(d-1-j)}\\
    &=\sum_{j=0}^{d-1}{d-1 \choose j}\frac{(k+d-1)!}{(k+d-1-j)!}x^{d-1-j}\frac{(d-1-j)!}{(d-1)!}\left(\frac{1}{1-x}\right)^{d-1-j+1}\\
    &=\sum_{j=0}^{d-1}{k+d-1 \choose j}\left(\frac{x}{1-x}\right)^{d-1-j}\frac{1}{1-x}.
\end{align*}
\end{proof}

%Now we are prepared for the  proof of theorem \ref{thm:L1error}:

\subsection{$L^P$ Convergence Rate without Boundary Information}
\label{apdix:pflperror_noboundary}
\begin{theorem}\label{thm:lp_noBoundary}
Let $f\in\mathcal{H}^1_{mix}$ and $\Phi$ be the kernel whose native space is equivalent to $\mathcal{H}^1_{mix}$.  Let $\hat{f}^{\rm SP}_k$ be the posterior mean of the GP with kernel $\Phi$ conditioned on a sparse grid design $\BX^{\rm SP}_k$ with $n$ design points. Then:
$$||f-f^s_k||_{L^\infty}=\CalO(n^{-\frac{1}{2}}[\log n]^{\frac{5}{2}(d-1)}).$$
 \end{theorem}
% This theorem is a direct result of Corollary 2 in \cite{rieger17}.

\begin{proof}
We replace $|f(x+\delta)-f(x)|$ with $[\int_0^1|f(x+\delta)-f(x)|^pdx]^{\frac{1}{p}}$ in equation (20) in \cite{rieger17} and use Hölder's inequality to get:
$$\int_0^1|f(x+\delta)-f(x)|^pdx=\int_0^1\left|\int_x^{x+\delta}f'(s)ds\right|^pdx\leq \delta^{\frac{1}{2}}||f'||_{L^2}. $$
As a result, we can have the following inequality:
$$\mathcal{E}(f;\pi_m(I))_{L^p}\leq cm^{-1+\frac{1}{2}}||f||_{\mathcal{H}^1}$$
where $\mathcal{E}(f;V)_{L^p}$ is the best approximation error for a given $f$ from $V$ measured in $L_p$ norm and $\pi_m$ is the set of polynomials of degree less than $m$. On the other hand, Theorem 8 in \cite{Barthelmann00} also holds true for $L^p$ norm, therefore,  by following the proof for Theorem 9 in \cite{rieger17}, we can have the following inequality:
$$||f||_{L^p(T^d)}\leq C{q-1 \choose d-1 }n^{-\frac{1}{2}}[\log n]^{\frac{3}{2}(d-1)}||f||_{\mathcal{H}^1_{mix}}+{q-d \choose d-1 }\max |f(\BX^s_q)|.$$
We then replace $f$ with $f-f^s_k$ . Because $f^s_k$ is exact on $\BX^{\rm SP}_q$ the second term on the right hand side vanishes. We have shown in Theorem \ref{thm:L1error} that ${q-1 \choose d-1 }=\CalO([\log n]^{d-1})$ which leads to the final result.
\end{proof}

\subsection{Probabilistic Convergence Rate without Boundary Information}
\label{apdix:pfproberror_noboundary}
\begin{theorem}
  \label{thm:proberror_noboundary}
  Suppose $I^{[0]} \cup I^{[1]} = \emptyset$, and assume the sparse grid design $\BX_k^{\rm SP}$ with $n = |\BX_k^{\rm SP}|$. Let $Z(\cdot)$ be a GP with kernel $k \in \CalH^{1}_{mix}(\CalX\times\CalX)$ with no boundary information, and  $\mathcal{I}_n^{\rm s}$ be the GP interpolation operator satisfying $\mathcal{I}_s^{\rm BM} f = \hat{f}_n^{\rm s}$, where  $\hat{f}_n^{\rm s}$ is the posterior mean. Then:
 \begin{equation*}
     \mathbb E\left[\sup_{\Bx\in\CalX}|Z(\Bx)-\mathcal{I}^s_kZ(\Bx)|^p\right]^{\frac{1}{p}}=\CalO(n^{-\frac{1}{2}}[\log n]^{\frac{5}{2}d-2}), \quad 1 \leq p < \infty,
     \end{equation*}
     and
     \begin{equation*}
     \sup_{\Bx\in \CalX}|Z(\Bx)-\mathcal{I}^s_kZ(\Bx)|=\CalO_{\mathbb{P}}(n^{-\frac{1}{2}}[\log n]^{\frac{5}{2}d-2}).
 \end{equation*}
  \end{theorem}
\begin{proof}
We can follow the proof for Theorem \ref{thm:proberror}, with the only difference being that for any kernel $k(\Bx,\By)\in\CalH^1_{mix}(\Real^d\times\Real^d)$ without boundary information, the uniform bound of the induced natural distance becomes:
$$\boldsymbol{\sigma}(\Bx,\By)=\CalO(n^{-\frac{1}{2}}[\log n]^{\frac{5}{2}(d-1)}).$$
The claim can then be shown by performing the same substitution as in the proof in Theorem \ref{thm:proberror}.
\end{proof}

\end{document}